%% file: main.tex
\documentclass[preprint,12pt,authoryear]{elsarticle}
\usepackage[utf8]{}
\usepackage{parskip}
\usepackage{enumitem}
\usepackage{booktabs, multirow}
\usepackage{amsmath, amssymb, amsthm}
\usepackage{graphicx}
\usepackage{tabularx, arydshln} 
\usepackage{listings}
\usepackage[section]{minted}
\usepackage[dvipsnames]{xcolor} 
\usepackage[font=small]{caption} 
\usepackage{subcaption} 
\usepackage{tabularray, seqsplit} 
\usepackage{enumitem} 
\usepackage{xspace} 
\usepackage[normalem]{ulem}
\usepackage{color}
\usepackage{hyperref}
\usepackage[capitalise]{cleveref} 

\input{config}

\begin{document}

\begin{frontmatter}
\title{Towards a Formal Verification of Secure \\Vehicle Software Updates\tnoteref{t1}}

\tnotetext[t1]{Preprint accepted for publication in Computers \& Security 2025 \citep{HAGEN2025104751}.}

\author[chalmers]{Martin Slind Hagen}
\author[chalmers]{Emil Lundqvist}
\author[chalmers]{Alex Phu}
\author[chalmers]{Yenan Wang}
\author[chalmers,volvocars]{\newline Kim Strandberg}
\author[chalmers]{Elad Michael Schiller}

\affiliation[chalmers]{organization={Chalmers University of Technology},
            addressline={Computer Science and Engineering}, 
            city={Gothenburg},
            postcode={41296}, 
            country={Sweden}}
\affiliation[volvocars]{organization={Volvo Car Corporation},
            addressline={Department of Research and Development}, 
            city={Gothenburg},
            postcode={40531}, 
            country={Sweden}}

\begin{abstract}
\small
With the rise of software-defined vehicles (SDVs), where software governs most vehicle functions alongside enhanced connectivity, the need for secure software updates has become increasingly critical. Software vulnerabilities can severely impact safety, the economy, and society. In response to this challenge, Strandberg et al.~[escar Europe, 2021] introduced the Unified Software Update Framework (UniSUF), designed to provide a secure update framework that integrates seamlessly with existing vehicular infrastructures. 

Although UniSUF has previously been evaluated regarding cybersecurity, these assessments have not employed formal verification methods. To bridge this gap, we perform a formal security analysis of UniSUF. We model UniSUF's architecture and assumptions to reflect real-world automotive systems and develop a ProVerif-based framework that formally verifies UniSUF’s compliance with essential security requirements — confidentiality, integrity, authenticity, freshness, order, and liveness —demonstrating their satisfiability through symbolic execution. Our results demonstrate that UniSUF adheres to the specified security guarantees, ensuring the correctness and reliability of its security framework.

\end{abstract}


\end{frontmatter}

\input{new_intro_extension}
\input{preliminaries}

\input{architecture}
\input{subproblems}

\input{methods}

\input{conclusions}

\newpage
\appendix
\input{appendix}

\end{document}

%% file: config.tex
\setitemize{leftmargin=*, align=left}
\newcommand{\Req}{Request\xspace}
\newcommand{\Suc}{Success\xspace}

\newcommand{\CA}{Consumer\xspace}
\newcommand{\CDA}{CDA\xspace}
\newcommand{\CIA}{CIA\xspace}
\newcommand{\CLS}{Consumer Local Storage\xspace}
\newcommand{\CMS}{CMS\xspace}
\newcommand{\CSA}{CSA\xspace}

\newcommand{\DI}{Download Instructions\xspace}
\newcommand{\DKM}{DKM\xspace}
\newcommand{\II}{Installation Instructions\xspace}
\newcommand{\IKM}{IKM\xspace}
\newcommand{\MKM}{MKM\xspace}
\newcommand{\OA}{Order Agent\xspace}
\newcommand{\OCS}{Order Cloud Service\xspace}
\newcommand{\PA}{Producer\xspace}
\newcommand{\PDA}{PDA\xspace}
\newcommand{\PIA}{PIA\xspace}
\newcommand{\PK}{PublicKey\xspace}
\newcommand{\PLS}{Producer Local Storage\xspace}
\newcommand{\Prod}{Supplier\xspace}
\newcommand{\PSA}{PSA\xspace}
\newcommand{\PSS}{PSS\xspace}
\newcommand{\SK}{PrivateKey\xspace}
\newcommand{\SKA}{SKA\xspace}
\newcommand{\SL}{Software List\xspace}
\newcommand{\SR}{Software Repository\xspace}
\newcommand{\SW}{Software\xspace}
\newcommand{\VCS}{Vehicle Cloud Service\xspace}
\newcommand{\VCM}{VCM\xspace}
\newcommand{\Vehc}{Vehicle\xspace}
\newcommand{\VD}{VIN Database\xspace}
\newcommand{\VIN}{VIN\xspace}
\newcommand{\VSO}{VSO\xspace}
\newcommand{\VUUP}{VUUP\xspace}
\newcommand{\ECU}{ECU\xspace}
\newcommand{\SA}{SecurityAccess\xspace}

\definecolor{backcolour}{rgb}{0.95,0.95,0.92}
\setminted[text]{
    framesep=2mm, 
    breaklines, 
    baselinestretch=1.2, 
    bgcolor=backcolour, 
    fontsize=\footnotesize, 
    linenos
}
\newenvironment{longlisting}{\captionsetup{type=listing}}{}
\newcommand{\appendixcode}[3]{
\begin{longlisting}
    \inputminted{text}{#1}
    \captionof{listing}{#2}
    \label{#3}
\end{longlisting}
}

\newtheorem{metaRequirement}{System-level Requirement}[subsection]
\newtheorem{theorem}{Theorem}[section]
\newtheorem{corollary}{Corollary}[theorem]
\newtheorem{lemma}[theorem]{Lemma}

\newlist{labelenumerate}{enumerate}{1}
\setlist[labelenumerate]{label=$\ell_{\arabic*}$:}

\crefname{metaRequirement}{System-level Requirement}{System-level Requirements}

%% file: new_intro_extension.tex
\section{Introduction}

Connected cars are quickly becoming the norm, with 96\% of manufactured cars in 2030 expected to have connectivity features~\citep{statista-connected-cars}.
These connected cars feature a multitude of electronic control units (ECUs), with more than 100 ECUs per vehicle~\citep{number-of-ecus-in-cars}.
Maintaining and regularly updating these ECUs is critical to prevent security vulnerabilities.
The massive scale of the automotive industry, with over 70 million cars sold worldwide annually~\citep{statista-cars-sold}, makes it a prime target for malicious actors.
If an attacker compromises the software update process, the result can be malware installation, sensitive data leakage, and even vehicle hijacking, leading to devastating financial, social, and potentially fatal consequences.

Despite the importance of secure updates, ensuring the confidentiality, integrity, and correct execution of the software update process for connected vehicles remains a highly challenging task.
In particular, attackers could exploit weaknesses in the update process to install malicious software, eavesdrop on sensitive information, or revert vehicle software to an obsolete version containing vulnerabilities. 
Moreover, the sheer number of vehicles and ECUs in each car presents a scalability challenge, complicating the implementation of robust security measures across the entire fleet.

\subsection{Existing Solutions and Their Shortcomings}
\label{sec:goals}
Frameworks have been proposed to address these challenges, including the Unified Software Update Framework (UniSUF)~\citep{UniSUF}.
UniSUF proposes a reference architecture intended to serve as input to standards for secure software updates.
UniSUF's specifications were driven by the following development goals~\citep{UniSUFAnalysis}.

\input{goals}

\Cref{fig:highlevel-arch} provides a high-level overview of the UniSUF architecture, illustrating the main entities involved in the update process: the Producer, the Consumer, Software Suppliers, the Software Repository, and Electronic Control Units (ECUs).
Software update packages, referred to as VUUPs, are generated by the Producer based on the latest software versions provided by the Software Suppliers, which have also passed internal testing.
The Producer delivers the VUUPs to the Consumer and uploads the software to the Software Repository. Upon receiving a VUUP, the Consumer processes it and installs the updates on the respective ECUs according to the instructions specified in the VUUP.
\Cref{sec:UniSUF Architecture} presents detailed architectures and protocol steps.

\begin{figure}[t]
  \centering
  \includegraphics[width=\columnwidth]{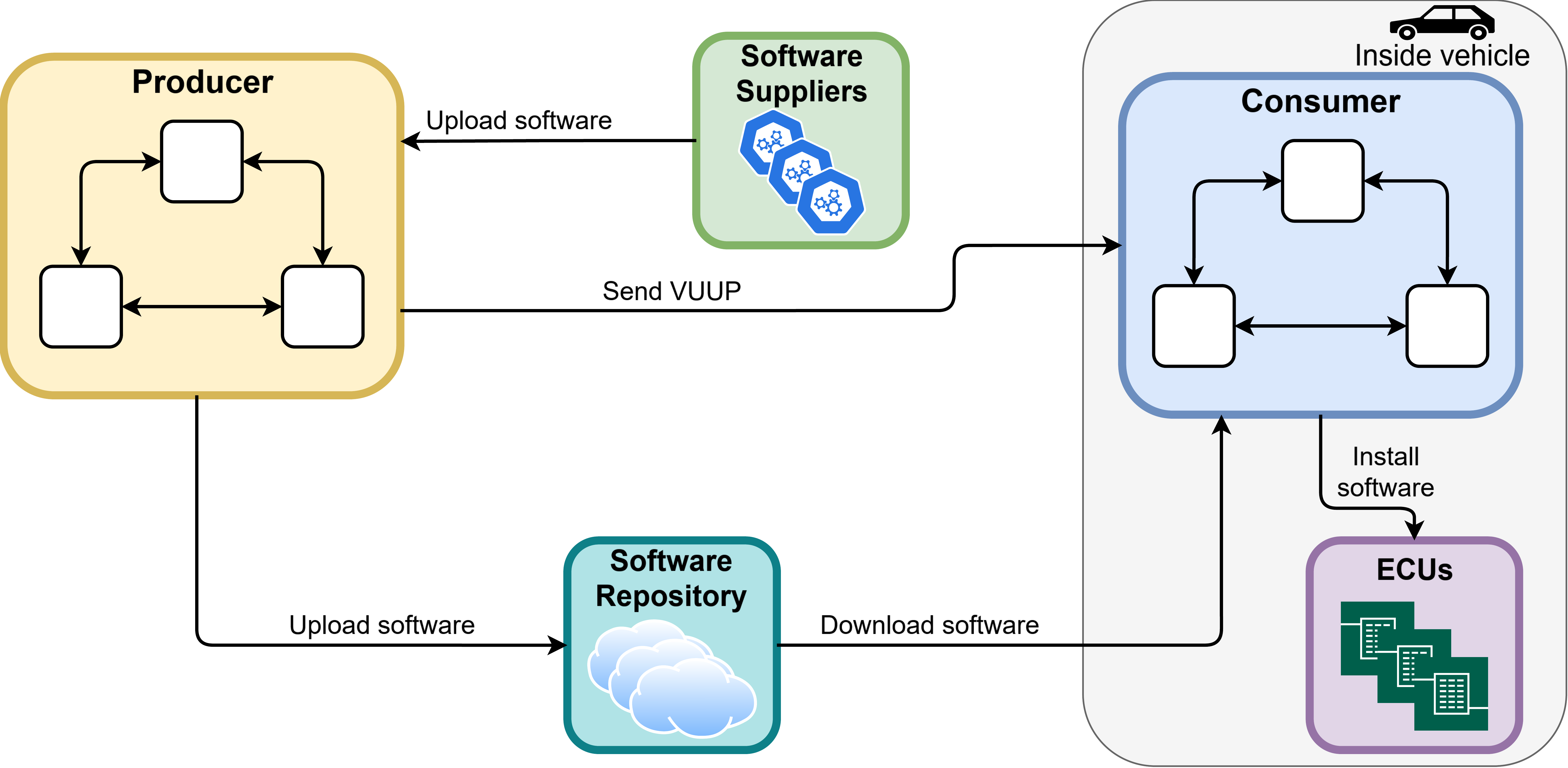}
  \caption{High-level overview of the UniSUF architecture showing the main entities and update flow. 
  \Cref{ch:subproblems} details the architecture components, see Figures~\ref{fig:preparation} to \ref{fig:stream_update_to_ecu}.}
  \label{fig:highlevel-arch}
\end{figure}

UniSUF's prior evaluations focus primarily on practical deployments and lack formal verification of its security guarantees.
This leaves the potential for subtle vulnerabilities that attackers could exploit, notably by exposing cryptographic keys or violating the sequence of operations during the update process.

Several scientific challenges are associated with the formal verification of systems like UniSUF. These include \textbf{[1: Confidentiality]} ensuring that the software update process maintains the confidentiality of secret information throughout execution; \textbf{[2: Integrity and Authenticity]} verifying that no unauthorized modifications occur during the update process; and \textbf{[3: Order and Liveness]} guaranteeing that the update process follows the correct sequence of actions and terminates appropriately.
These challenges are compounded by the complexity of modeling such systems in formal verification tools such as ProVerif~\citep{ProVerif-survey,proverif-manual}, which requires a precise representation of the security assumptions and the adversary model.

\label{sec:research-questions}
\label{sec:problem-description}

The challenges associated with the formal security of UniSUF raise the following research questions.
\begin{itemize}
     \item[\textbf{RQ1}:] UniSUF has certain secrets, essential for its operation, for example, cryptographic keys and disseminated software. The question is whether UniSUF's operation might expose any of its secrets.
    \item[\textbf{RQ2}:] How can we guarantee that the software that UniSUF disseminates is obtained from an authentic source and not manipulated? 
    \item[\textbf{RQ3}:] How can we guarantee that in UniSUF, it is impossible to perform a software update with obsolete software versions?
    \item[\textbf{RQ4}:] Appropriate software updates require a specified order of actions. How can we guarantee that UniSUF's update operation proceeds orderly?
    \item[\textbf{RQ5}:] How do we know that the software update process always ends?
\end{itemize}

\subsection{Our Contribution}
\label{sec:our_contribution}
We address the above challenges and research questions by conducting a formal security analysis using symbolic execution in ProVerif. Our key contributions to advancing the state of the art are as follows:
\begin{itemize}
    \item We model UniSUF’s architecture (\Cref{sec:UniSUF Architecture}) and assumptions (\Cref{ch:preliminaries}) to reflect real-world automotive systems. This model is represented in ProVerif (\Cref{ch:subproblems}) to show the formal satisfiability of essential security properties, including confidentiality, integrity, authenticity, freshness, order, and liveness.
    \item We introduce novel techniques (\Cref{sec:methods}) to simulate UniSUF in a symbolic execution environment within ProVerif, allowing us to formally analyze critical aspects, such as software authenticity and the correct sequence of operations.
    \item Through formal proofs and experimental results, we ensure that UniSUF’s update process terminates and follows the correct procedural order.
\end{itemize}

\noindent The primary outcomes of our results are as follows.
\begin{itemize}
    \item A rigorous formulation of UniSUF’s security requirements, emphasizing confidentiality, integrity, authenticity, freshness, order, and liveness in the software update process.
    \item An open-source ProVerif-based framework to formally verify UniSUF's compliance with these security guarantees.
    \item Through ProVerif's symbolic execution environment, we demonstrate that UniSUF can satisfy the proposed security requirements under realistic adversary models while considering a system architecture that represents real-world deployment.
\end{itemize}

Our results show that UniSUF’s architecture effectively prevents attacks, such as secret exposure and replay attacks while ensuring that software updates proceed in the correct sequence and terminate as expected. This demonstrates that UniSUF’s architecture, assumptions, and security requirements can be formally satisfied, providing strong assurances for its security in real-world deployments (\Cref{sec:Conclusions}).

An important clarification concerns the notion of computational cost and scalability. The effort reported in this paper concerns the formal specification and verification of requirements, rather than the execution of update protocols in deployed vehicles. Consequently, any computational overhead arises exclusively at \emph{design time}, when symbolic verification tools such as ProVerif analyze the model and proof obligations. These costs depend on model size and tool runtime, but are incurred only once during the verification process. They do not translate into runtime overhead for the vehicle platform or its backend infrastructure. Moreover, symbolic verification techniques, including ProVerif, are conceptually scalable in the sense that they reason over unbounded sessions and adversaries. Our approach, therefore, provides strong security assurances without introducing performance or scalability penalties in production systems.

To ensure the reproducibility of our results and to encourage further development, we pledge to release our solution as open-source upon acceptance of the paper, as detailed in Appendix A.

\input{related_works}

%% file: goals.tex
\begin{itemize}
    \item[\textbf{Confidentiality}:] To hinder eavesdropping.\begin{itemize}
        \item[\textbf{G1}:] Ensure software confidentiality during the software update process.
        \item[\textbf{G2}:] UniSUF session keys can only be viewed in decrypted format by authorized software components.
    \end{itemize}
    \item[\textbf{Integrity and Authenticity}:] To hinder spoofing and tampering.
    \begin{itemize}
        \item[\textbf{G3}:] The software is authentic against a  
        certificate and remains unchanged during the update process.
        \item[\textbf{G4}:] Only authentic resources are processed.
    \end{itemize}
    \item[\textbf{Freshness}:] To hinder replay attacks. 
    \begin{itemize}
        \item[\textbf{G5}:] An adversary should be unable to revert a vehicle's software to a previously installed version.

        \item[\textbf{G6}:] The system creates unique software distribution files per software update. Each such file can only be used for a designated vehicle.
    \end{itemize}
    \item[\textbf{Order}:] To hinder vulnerabilities that take advantage of running the process in an unintended order.
    \begin{itemize}
        \item[\textbf{G7}:] The software update process should follow the correct order.
    \end{itemize}
    \item[\textbf{Liveness}:] Hinder DoS attacks.
    \begin{itemize}
        \item[\textbf{G8}:] The software update process should eventually terminate, regardless of success or error.
    \end{itemize}
\end{itemize}

%% file: related_works.tex
\section{Related Work}
Analyzing the security of a complex system requires consideration of potential threats.
\citet{resilientShield} proposes a security and resilience framework, i.e., Resilient Shield, utilizing a security enhancement methodology \citep{8278159} along with mitigation mechanisms from \citep{9230288} in light of an analysis of attacks targeting vehicles. 
Based on the identified attacks, the authors establish security goals and specify the directives required to achieve them.
\citet{UniSUFAnalysis} details specifications for meeting security goals with an emphasis on vehicle software updates.
The authors provide a detailed threat analysis of UniSUF for different threat scenarios, aligned with security goals~\citep{resilientShield}, and detailed requirements for the UniSUF architecture~\citep{UniSUF}.

\subsection{Formal Verification Tools}

There are formal verification tools that can either assist or automate proofs.
For instance, we can express systems and their properties as logical formulas and use theorem provers such as Coq~\citep{coq} and Isabelle~\citep{isabelle} to either interactively or automatically prove the specified properties.
However, formulating the entire implementation of complex systems is tedious and error-prone.
Instead, it can be more intuitive to verify a system based on a formal model using model checkers such as \textsc{Spin}~\citep{spin}, \textsc{Uppaal}~\citep{uppaal}, and TLA+~\citep{tla+}.
Model checkers take a formula and a model and then verify whether the formula holds within the model.
The appropriate model checker can simplify the process of deriving a system model.
For example, \textsc{Uppaal} is a model-checking environment that provides both a graphical interface and a modeling language to model real-time systems~\citep{uppaal}, which makes it well suited for time-critical systems.

For our purposes, we also need to model the adversary, to be
explicitly defined in general model checkers. In contrast, cryptographic protocol verifiers, such as ProVerif~\citep{proverif-manual}, CryptoVerif~\citep{cryptoverif}, and Tamarin Prover~\citep{tamarin}, are designed to focus on security and implicitly incorporate the adversary model.
For instance, these cryptographic protocol verifiers verify their models under the assumption of the Dolev-Yao adversary model~\citep{Dolev-Yao}.
Therefore, cryptographic protocol verifiers are more appropriate for security properties as they eliminate the need to model an adversary.

\citet{junlangwang} provides protocols to attest that the manufacturer has approved vehicle hardware. 
These protocols enable the replacement of old components with new ones and the attestation of all vehicle components during start-up.
Wang \citet{junlangwang} used formal methods to prove the correctness of his protocols. 
He defined the system model, assumptions, and requirements and used ProVerif~\citep{proverif-manual} to formally verify that his protocols fulfill the specified requirements given the system model and assumptions.
The work by Wang has been an inspiration for our own research in this area.

\citet{5gFormalVerification} used Tamarin Prover to find weaknesses in the Authentication and Key Agreement protocol used by 5G. Tamarin Prover has also been used to analyze WiFi Protected Access 2~\citep{wpa2FormalVerification} and Transport Layer Security 1.3 (TLS 1.3)~\citep{tls1.3FormalVerification}. 
In \cite{tlsProVerif}, an analysis of TLS 1.3 was performed with ProVerif~\cite{proverif-manual}, where they also looked at a privacy extension for TLS 1.3 called Encrypted Client Hello.

\subsection{Formal Verification of Automotive and IoT Software Update Protocols}

Several formal methods have been applied to analyze the security of software update mechanisms in automotive and IoT systems. A significant body of work has focused on \emph{Uptane}, which is a widely adopted framework for secure automotive software updates, considered by many in the industry as a de facto standard. For example, Kirk et al.~\cite{DBLP:journals/jlap/KirkNBSW23} develop a Communicating Sequential Processes (CSP) model of the Uptane protocol together with an attacker inspired by the Dolev--Yao model~\citep{Dolev-Yao}, and use the Failures--Divergences Refinement (FDR4) checker to derive exhaustive traces representing potential security violations. These symbolic traces are then translated into executable test cases and run against a reference Uptane implementation on a hardware testbed. Their analysis validated Uptane’s defenses against the modeled threats, but also demonstrated that certain attacks, such as freeze or spoof, become feasible when specific defenses (e.g., expiration checks) are omitted in practice. Kirk et al. illustrate how CSP-based refinement and test generation can reveal both specification-level exposures and implementation-dependent weaknesses.

In a complementary study,~\cite{DBLP:conf/raid/LorchLTC24} present a comprehensive automated verification of Uptane by combining the Kind~2 infinite-state model checker with the Tamarin cryptographic protocol verifier in an eager combination. Unlike prior Uptane analyses that faced state-space explosion or termination issues, their workflow achieves termination in most instances while covering fine-grained message structures. This analysis rediscovered all five previously known vulnerabilities and identified six new ones, all of which were acknowledged by the Uptane standards body. Their work illustrates how combining model checking with cryptographic reasoning can provide both scalability and coverage under diverse adversary capabilities, including multiple key-compromise scenarios.

Another formal analysis of Uptane is provided by~\cite{boureanu2023uptane}, who used the Tamarin prover to model version~2.0 of the protocol. Building on the threat model in the Uptane~2.0 standard, this work introduced a hierarchical set of attacker tiers, including compromise of primary or secondary ECUs and of individual repositories. 
The authors validated a set of security and privacy requirements that go beyond those specified in the standard, including properties such as agreement and temporal correctness.
They identified several flaws that were responsibly disclosed to the Uptane Alliance. While the verification required substantial computational resources, the study demonstrated that symbolic analysis of Uptane~2.0 at scale is feasible and can directly inform improvements to the evolving standard.

Earlier,~\cite{mahmood2020ota} introduced the first model-based security testing approach for automotive over-the-air~(OTA) updates, targeting the Uptane reference implementation. Their framework combined attack trees for threat modeling with an automated tool that generates and executes corresponding test cases. In a proof-of-concept demonstration, they showed a simulated attack compromising Uptane repositories with malicious firmware. While this approach provides systematic coverage of known threats, its reliance on attack trees limits its ability to capture zero-day or composite attacks, motivating later exhaustive analyses based on formal models.

Not all formal verification studies target Uptane; some address other software update protocols or different aspects of the OTA software update process. \cite{DBLP:conf/icissp/PonsardD21}, for instance, apply the Tamarin prover to an IoT-oriented firmware update scheme called \emph{UpKit}. Their work focuses on proving selected properties, most notably firmware integrity and the freshness of update requests, under a Dolev-Yao adversary. By modeling the cryptographic operations and message flows in Tamarin’s rewrite-rule framework, they verify that UpKit meets its targeted requirements, though the analysis does not cover the full range of software update properties.

Other researchers have looked at implementation-level assurance. \cite{mukherjee2021trusted} propose an Uptane-based OTA update solution deployed entirely inside a Trusted Execution Environment (TEE) on commercial off-the-shelf embedded hardware (ARM TrustZone). Using SAW, the Software Analysis Workbench, they verify code-level security properties of the Uptane client within OP-TEE and demonstrate the approach on a Raspberry Pi~3B. Their threat model assumes that the normal-world OS and network may be compromised, while the TEE and server remain trusted. This work highlights how hardware-enforced isolation can be combined with formal software analysis to strengthen OTA software update implementations.

An earlier foundational effort by~\cite{6093061} introduced a holistic formal methodology for secure OTA updates in vehicles. Their approach extended the AVATAR SysML-based framework to incorporate both security and safety requirements, and integrated formal verification into a model-based development flow. Specifically, AVATAR models of OTA update protocols were translated to ProVerif for symbolic analysis under a Dolev-Yao adversary, while safety aspects were verified using UPPAAL. This allowed properties such as secure authentication, confidentiality, and data integrity of updates to be formally checked already at the design stage. The emphasis was on combining early requirements engineering with formal security and safety verification, thereby demonstrating as early as 2011 the feasibility of security-by-design for vehicular update systems.

Our work on UniSUF diverges from these studies in both goal and approach. Rather than assessing an existing protocol for hidden flaws, we focus on \emph{formally specifying a new multi-ECU update framework and verifying that it satisfies a comprehensive set of security requirements by design}. We encode UniSUF's update procedures in ProVerif and prove, under a standard Dolev--Yao adversary model, that critical properties such as authenticity of the update source, integrity of firmware payloads, confidentiality of sensitive data, and freshness of update commands hold. Ordering is ensured through replay-protection and monotonicity checks, while liveness is argued at the design level rather than mechanically proved. Our verification approach is requirement-driven: each high-level requirement is formalized as a ProVerif query, and the framework is decomposed into interrelated sub-protocols to mirror the modular structure of  the UniSUF framework. This decomposition facilitates scalability and clarity. In positioning, UniSUF extends the requirements-oriented perspective already seen in early frameworks such as AVATAR, while differing from vulnerability-focused analyses of Uptane and its implementations (e.g., attack-tree testing, CSP/FDR model-based testing, or hybrid model checking with cryptographic reasoning). Taken together with Boureanu’s symbolic proof of Uptane~2.0 and Mahmood’s attack-tree-based framework, these efforts trace the evolution of formal verification for software update protocols, from semi-formal test generation, through symbolic and hybrid analyses, to requirements-driven verification by design for new multi-ECU frameworks.

%% file: preliminaries.tex
\section{Preliminaries}
\label{ch:preliminaries}
We provide our definitions, assumptions, and requirements.

\input{system_setting}
\input{threat_model}
\input{cryptographic_primitives_notations_assumptions}
\input{update_rounds}
\input{problem_definition}

%% file: system_setting.tex
\subsection{System Settings}
\label{sec:system-settings} 

The system consists of computing entities that interact through communication channels. 
We assume the system is synchronous and that all entities can access universal time.
Every entity has a state, including its variables and all messages in its incoming communication channels. 
The entities update their states by taking atomic steps. 
Each step performs an internal computation that takes one time unit. 
These steps can also receive or send messages.
An unbounded sequence of atomic steps, $X$, denotes an execution.
For a given entity $E$, an execution of $E$ is a subsequence of $X$ from which all steps not taken by $E$ are omitted. 
Each of our studied problems is solved by a distributed algorithm.
The system entities collectively execute an algorithm by individually running a sequence of tasks. 
Each problem is divided into the sub-problems we analyse in \Cref{ch:subproblems}.
The last task in each sequence is the \emph{halt} task.

%% file: threat_model.tex
\subsection{Threat Model}
\label{sec:threat-model}

Based on the Dolev-Yao model \citep{Dolev-Yao}, the adversary has complete control over the communication between entities. In addition, message interception, injection, and modification are also possible. 
The adversary may also delay the delivery of the message by a bounded time $\eta$; therefore, communication channels are assumed to be reliable but without guarantees of FIFO ordering.
This is derived from \citet{junlangwang}.

%% file: cryptographic_primitives_notations_assumptions.tex
\subsection{Cryptographic Primitives, Notations, and Assumptions}
\label{sec:crypto-primitives}

We assume access to the standard cryptographic primitives in~\Cref{tab:notations}.
We emphasize the requirement for an authenticated symmetric encryption scheme, such as AES-GCM, which ensures that the encrypted data remain confidential and are authenticated to verify the sender~\citep[Sec. 1]{rfcAesGcm}. 
This is necessary because some secrets must be authenticated, such as inputs to the Trusted Execution Environment~\citep[Tab. 1]{UniSUFAnalysis}).
To simplify our model, we define a certificate as valid if it is signed by the root certificate. For simplicity, we omit the explicit notation of these signatures in our cryptographic descriptions, assuming that all valid certificates are implicitly signed. UniSUF operates under the assumption of a trusted root certificate \citep{UniSUF, UniSUFAnalysis}. Consequently, we assume that this root certificate is securely pre-installed in the vehicle.
 
UniSUF assumes secure and reliable communication between entities.
This can be achieved, for example, by using SSH~\citep{rfcSsh} or mutual TLS~\citep{rfcTls, rfcmTls}.
Therefore, adversaries cannot eavesdrop, tamper, or replay messages, with the exception of one link (see~\Cref{sec:stp-dcp-17}) for which UniSUF uses reliable non-FIFO communication without security guarantees.
UniSUF uses cryptographic materials (see~\Cref{sec:crypto-mat}), such as symmetric keys and cryptographic signatures. 

\begin{longtblr}[
    caption = {Notations used in the communication schemes. Inspired by \citet[Table 1]{junlangwang}.},
    label = {tab:notations}
]{
    colspec = {|p{0.46\textwidth}|X|},
    width = \columnwidth,
    rowhead = 1,
    hlines,
    row{even} = {gray9},
    row{1} = {olive9}, 
    column{0}={3.25cm}
}
    Notation & Description \\
    
    $ Obj_{Key} $ & Symmetric key of type $Obj$. \\
    
    $ E_{Cert}$ & Entity $E$'s certificate is an asymmetric key pair. The key pair's public key is signed by the root certificate and only $E$ knows the private key. The private key is omitted when $E_{Cert}$ is included in a data structure or message. \\
    
    $ E_{\SK}$ & The private key belonging to $E_{Cert}$. \\
    
    $ E_{\PK}$ & The public key belonging to $E_{Cert}$. \\
    
    $ AsymEnc(Message,\, E_{\PK}) $ & Asymmetrically encrypts the given $Message$ with $E_{\PK}$, thereby creating $CipherText$. \\
    
    $ AsymDec(CipherText,\, E_{\SK}) $ & Asymmetrically decrypts the given $CipherText$ with the private key $E_{\SK}$, thereby recovering $Message$. \\
    
    $ SymEnc(Message,\, Obj_{Key}) $ & Symmetrically encrypts the given $Message$ with the symmetric key $Key$, thereby creating $CipherText$. \\
    
    $ SymDec(CipherText,\, Obj_{Key}) $ & Symmetrically decrypts the given $CipherText$ with the symmetric key $Obj_{Key}$, thereby recovering $Message$. 
    \\
    
    $ AuthSymEnc(Message,\, Obj_{Key}) $ & Encrypts the $Message$, similar to $SymEnc(Message, Obj_{Key})$, but will also include an authentication tag that hinders $Message$ from being changed, and validates if the $Obj_{Key}$ is used to encrypt $Message$. \\
    
    $ AuthSymDec(CipherText,\, Obj_{Key}) $ & Decrypts the $CipherText$, similar to $SymDec(CipherText, Obj_{Key})$, but any change to the encrypted message or any message encrypted by $Obj_{Key}' \neq Obj_{Key}$ will be detected. \\
    
    $ Hash(Message) $ & Creates a hash $H$ of $Message$, such that $Message$ cannot be retrieved from $H$. \\
    
    $ Sign(H,\, E_{\SK}) $ & Creates a signature $S$ of the hash $H$ by encrypting $H$ with $E_{\SK}$, such that decrypting $S$ with $E_{\PK}$ returns $H$, i.e., $AsymDec(S, E_{\PK}) = H$. \\
    
    $ [Message]_E $ & Represents data that is signed by $E_{\SK}$. It is shorthand for $Message\;||\; Sign(Hash(Message),$ $\,E_{\SK})$. \\
    
    $ Create_{Obj}(args) $ & A function that creates an object of type \textit{Obj}. Optional arguments $args$ can also be included. The creation details may vary with different values of $Obj$ and $args$. \\
    
    $ \Req(Obj) $ & Creates a flag for requesting an item or functionality of type $Obj$. Such flags are used in messages to model the various requests sent between entities in UniSUF (see \Cref{ch:subproblems}). \\
    
    $ \Suc(Obj) $ & Creates a flag stating that item of type $Obj$ was successfully initiated. Such flags are used in messages to model success statuses sent between entities in UniSUF (see \Cref{ch:subproblems}). \\
    
    $ i_1\;||\ldots||\;i_n $ & Represents concatenation of multiple items, more specifically from $i_1$ to $i_n$. When the three dots operator ($\ldots$) notation is used at the end of the concatenation, e.g. $i\;||\;\ldots$, it indicates that additional but unspecified items are being concatenated after $i$. Note that an item can be any data. \\
    
    $ (i_1,\, \dots,\, i_i,\, \dots) = I$ & Items $i_1$ and up to $i_i$ are extracted from the set of items $I$. The optional dots at the end signify that more items left in $I$ are ignored. \\
\end{longtblr}

%% file: update_rounds.tex
\subsection{Update Rounds}
\label{sec:update-rounds}

Where applicable, cryptographic materials are assigned to individual vehicle identification numbers (VIN), $v_{id}$, and must be distributed within a specified deadline, $t_e$ (expiration time)~\citep{UniSUFAnalysis, kstrandberg}. 
The mapping is secured by appending the $v_{id}$ and $t_e$ to the cryptographic material and signing the resulting data.
We use the pair $(v_{id}, t_e)$ to refer to the software \emph{update rounds} in UniSUF. 
We choose the term rounds, rather than sessions, to avoid confusion with the term sessions used, e.g., for SSH~\citep[Sec. 2]{sshSurvey}.
When no $v_{id}$ can be specified, we omit the VIN from our update round notation and use only $t_e$.

For a given problem and its algorithm, an execution of an algorithm is denoted as an update round execution. 
We assume that message transmissions include an update round identifier. 
Therefore, entities can learn about new update rounds and associate each execution of their task sequences with an update round.
We denote this execution as an entity's execution of an update round, which is a subsequence of the update round execution.
We also assume that all entities have a persistent log of all received messages. Each message is identified by the tuple $(r, d)$, where $r$ is the update round identifier and $d$ is the cryptographic material in the message. All entities drop any message that is already in the persistent log.

%% file: problem_definition.tex
\subsection{Problem Definition}
\label{sec:problem-definition}
We present our requirements for UniSUF, using the goals derived in \Cref{sec:problem-description}. 

\Crefrange{req:confidential-secrets}{req:termination} specify UniSUF at the system level. 
The \Cref{req:confidential-secrets,req:integrity-of-cryptographic-materials,req:inter-round-uniqueness,req:integrity-of-handling-events} depend on requirements specified for each UniSUF sub-problem. These requirements are presented in \Cref{ch:subproblems}. 
Namely, one derives the specifications of UniSUF sub-problems by specifying the sub-problems' \emph{set of secrets} ($S$), the \emph{set of cryptographic materials} ($\mathcal{D}$), the \emph{set of procedures} ($\{\ell_1,\ \ell_2,\ \dots\}$) that handles cryptographic materials, and the ordering constraints on the procedure invocations, which we call the \emph{handling partial order} ($\mathcal{P}(\ell)$).

\begin{metaRequirement}[Confidential Secrets]
Let $X$ be an update round execution and $S$ be the \emph{set of secrets} in $X$, which we specify per sub-problem (see \Cref{ch:subproblems}).
There is no $s_i \in S$ such that an adversary $A$ can obtain $s_i$ during $X$.
    
\label{req:confidential-secrets}
\end{metaRequirement}

\Cref{req:integrity-of-cryptographic-materials,req:inter-round-uniqueness,req:intra-round-uniqueness} 
consider a set of cryptographic materials $\mathcal{D}$, which we specify for each sub-problem in \Cref{ch:subproblems}. Each element in $\mathcal{D}$ is a pair; the first element is the cryptographic material itself, and the second element is the material's designated origin entity. As certificates are pre-existing cryptographic materials, we state that the origin entity of each certificate is the root CA. 

\Cref{req:integrity-of-cryptographic-materials} requires that the adversary must not manipulate the cryptographic materials in $\mathcal{D}$. 

\begin{metaRequirement}[Integrity of Cryptographic Materials]
UniSUF only uses cryptographic materials created by their designated origin entity, which we specify in \Cref{ch:subproblems}, and is not modified by any other entity.
\label{req:integrity-of-cryptographic-materials}
\end{metaRequirement}

\Cref{req:inter-round-uniqueness} considers procedures that handle cryptographic materials, such as material production, sending, receiving, validation, and software installation.
To safeguard the vehicle during the most critical part of the update process, the vehicle enters offline mode \citep{UniSUF, UniSUFAnalysis, kstrandberg}.
Additionally, the ECUs are normally locked for security reasons, but must be unlocked to install new software.
Therefore, we also classify ECU unlocking and vehicle offline mode activation as handling events.

Let $X$ be an execution of update round $r$, $d$ (documents) a subset of cryptographic materials, and $\ell$ a name of a procedure that handles $d$ in event $e(r,d,\ell) \in X$. In \Cref{ch:subproblems}, we list these procedures per task.
We state that $e(r,d,\ell)$ is a handling event in $X$. Note that $d \subseteq \mathcal{D}$.

\Cref{req:inter-round-uniqueness} specifies that cryptographic materials coupled with their handling procedures and associated with a specific update round are processed only during that update round; i.e., replays of round unique cryptographic materials between update rounds are not allowed. 

\begin{metaRequirement}[Inter-Round Uniqueness]
    \label{req:inter-round-uniqueness}
    Let $e(r,d,\ell)$ and $e(r',d',\ell')$ be handling events during update round executions $X$ and $X'$, respectively. 
    Suppose $(d,\ell)=(d',\ell')$.
    It holds that $r=r'$.
\end{metaRequirement}

\Cref{req:intra-round-uniqueness} specifies that cryptographic materials coupled to a specific update round are only processed once per procedure during that update round.
In other words, replays of materials in an update round are not allowed. 
\begin{metaRequirement}[Intra-Round Uniqueness]
\label{req:intra-round-uniqueness}
    Let $X$ be an update round execution and $e(r,d,\ell) \in X$ be a handling event. 
    No event $e'(r,d,\ell) \text{ exists in} \ X$.
\end{metaRequirement}

\Cref{req:integrity-of-handling-events} specifies that procedures are executed in an order that follows UniSUF's specification.

\begin{metaRequirement}[Integrity of Handling Events]
\label{req:integrity-of-handling-events}
    Let $X$ be an update round execution. 
    The occurrences of handling events, $e(r,d,\ell) \in X$, must follow a handling partial order $\mathcal{P}(\ell)$, which depends only on $\ell$, where $\mathcal{P}(\ell)$ is specified per task (see \Cref{ch:subproblems}). 
\end{metaRequirement}

\Cref{req:termination} prevents non-termination of update rounds, given our system assumptions. The termination is considered timely if all UniSUF entities take the halt task before the update round expires; if not, it is considered late.

\begin{metaRequirement}[Termination]
    \label{req:termination}
    All executions of update rounds must terminate. 
\end{metaRequirement}

\noindent Using the goals specified in \Cref{sec:problem-description}, we discuss how the requirements satisfy the goals. 
Firstly, \textbf{G1} and \textbf{G2} are covered by \textbf{Confidential Secrets}, as we ensure that both the software and the cryptography keys are part of the set of secrets $S$. 
Note that \textbf{G2} is only partially covered because  ProVerif~\citep{proverif-manual} can only prove that the adversary is unable to learn the secrets and not that each secret is only available to a certain set of entities.
Secondly, \textbf{G3} and \textbf{G4} are fulfilled by \textbf{Integrity of Cryptographic  Materials} because the software and other materials produced by the origin entities are never modified throughout the update process.
Thirdly, \textbf{Inter-Round Uniqueness} and \textbf{Intra-Round Uniqueness} together satisfy \textbf{G5} because any update round that processes the installation of software cannot be replayed to execute previous versions. \textbf{G6} is fulfilled because any VUUP can only exist in a single update round, and such an update round is directly coupled to a vehicle via the VIN (see \Cref{sec:system-settings}).
Fourthly, \textbf{G7} is fulfilled by \textbf{Integrity of Handling Events}, and lastly, \textbf{G8} is achieved by \textbf{Termination} since this requires that all executions terminate, legitimate or illegitimately.

%% file: architecture.tex
\section{UniSUF Architecture}
\label{sec:UniSUF Architecture}
In this section, we detail the UniSUF architecture and its functionalities.

\input{materials}
\input{system_entities}
\input{modelling_unisuf}

%% file: materials.tex
\subsection{Cryptographic Materials}
\label{sec:crypto-mat}

As detailed in \Cref{tab:crypto-materials} and previously mentioned in~\Cref{sec:crypto-primitives}, UniSUF uses different cryptographic materials. The signing process is shown in \Cref{fig:signing}.

\begin{figure}[ht]
    \centering
    \includegraphics[width=\columnwidth]{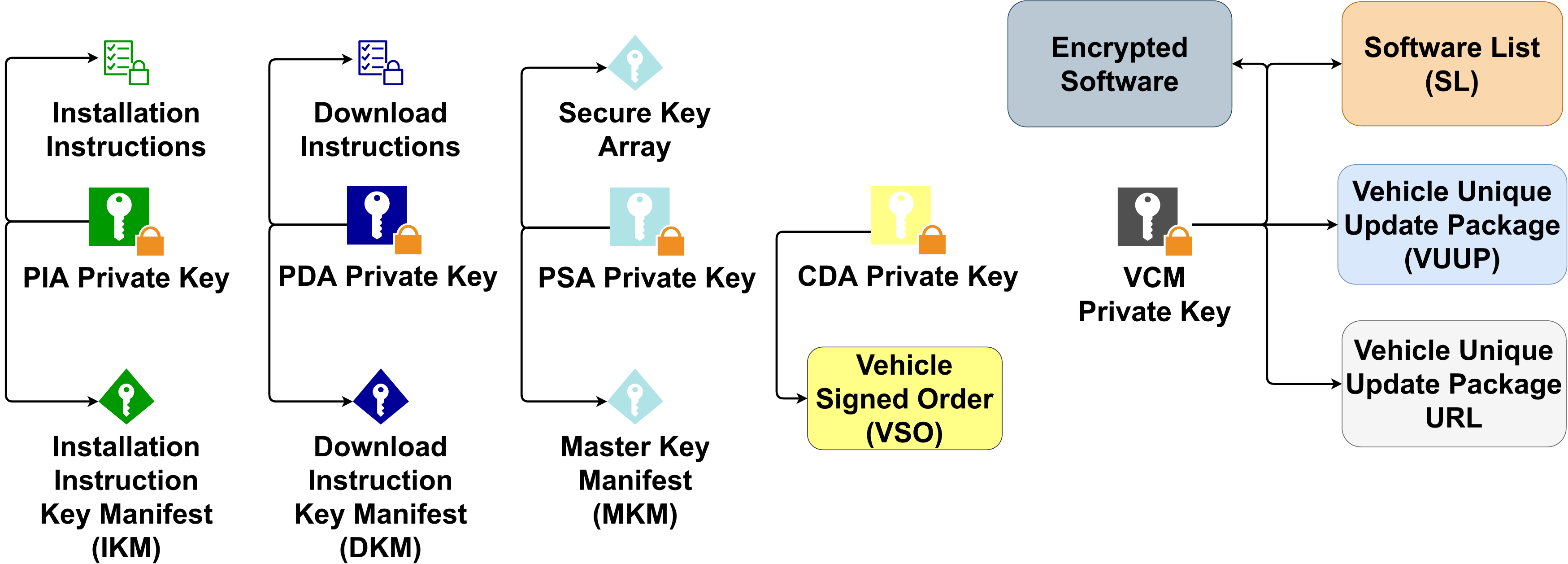}
    \caption{Different cryptographic materials in UniSUF are shown, each with its respective key. Figure derived from \citet[Fig. 3]{UniSUF}.}
    \label{fig:signing}
\end{figure}

\begin{longtblr}[
    caption = {UniSUF cryptographic materials, sorted in alphabetical order.},
    label = {tab:crypto-materials}
]{
    colspec = {|p{0.34\textwidth}|X|},
    width = \columnwidth,
    rowhead = 1,
    hlines,
    row{even} = {gray9},
    row{1} = {olive9}, 
    column{0}={3.25cm}
}
   Cryptographic Material & Description \\    
   Certificate Package & $ PDA_{Cert} \;||\; PIA_{Cert} $ \\
    
    \DI & $ Create_{\DI}([\SL]_{\VCM}) $ \\
    
    Download Instruction Key Manifest (DKM) & $ AsymEnc(\DKM_{Key},\,\Vehc_{\PK})$
    $\newline\phantom|||\; \DKM_{Policy} $ \\
    
    \II & $ Create_{\II}( $
    $ \newline\phantom{xd}[\SL]_{\VCM} \,||\, [\SKA]_{\PSA} $
    $\newline\phantom{xd}||\; [\MKM]_{\PSA} \,||\, $
    $ \newline\phantom{xd}||\; \PSA_{Cert}) $ \\
    Installation Instruction Key Manifest (IKM) & $ AsymEnc(\IKM_{Key},\,\Vehc_{\PK})$
    $\newline\phantom{xd}||\; \IKM_{Policy} $ \\
    
    Master Key Manifest (MKM) & $ \MKM_{\SA} \;||\; \MKM_{\SW} \;||\; \dots $ \\
    
    $ \MKM_{\SA} $ & $ AsymEnc(\MKM_{\SA_{Key}},$
    $ \newline\phantom{xd}\Vehc_{\PK}) $
    $ \phantom{xd}|| \;\MKM_{\SA_{Policy}} $ \\
	
	$ \MKM_{\SW}$ & $ AsymEnc(\MKM_{\SW_{Key}},\,\Vehc_{\PK}) $
    $ \phantom{xd}|| \;\MKM_{\SW_{Policy}} $ \\

    Secure Key Array (SKA) & $ \SKA_{\SA} \;|| \; \SKA_{\SW} \;|| \;\dots $\\
    
    $ \SKA_{\SA}$ & $ AuthSymEnc($
    $ \newline\phantom{xd}\SA_{Key_1},\,\MKM_{\SA_{Key}}) $ 
    $ \newline\phantom{xd}||\dots $ 
    $ \newline\phantom{xd}||\;AuthSymEnc( $
    $ \newline\phantom{xd} \SA_{Key_n},\,\MKM_{\SA_{Key}}) $ \\
   
    $ \SKA_{\SW}$ & $AuthSymEnc(\SW_{Key_1},$
    $\newline\phantom{xd} \MKM_{\SW_{Key}}) $ 
    $ \newline\phantom{xd}||\;\dots $ 
    $ \newline\phantom{xd}||\;AuthSymEnc(\SW_{Key_n},$
    $\newline\phantom{xd} \MKM_{\SW_{Key}}) $\\

    \SW & $ Version \;||\; Content $\\
    $ \SW_{Encased}$ & $[SymEnc([\SW]_{Supplier}, $
    $\newline\phantom{xd} \SW_{Key})]_{\VCM} $\\
  
  Vehicle Unique Update Package (VUUP) & $ \VCM_{Cert} \;||\; [\VUUP_{Content}]_{\VCM} $\\
  
  $ \VUUP_{Content}$ & $CertificatePackage $
  $ \newline\phantom{xd}||\; [SymEnc(\DI,$
  $\newline\phantom{xdxdxd}\DKM_{Key})]_{\PDA} $
  $ \newline\phantom{xd}||\; [\DKM]_{\PDA} $
  $ \newline\phantom{xd}||\; [SymEnc(\II,$
  $\newline\phantom{xdxdxd}\IKM_{Key})]_{\PIA} $
  $ \newline\phantom{xd}||\; [\IKM]_{\PIA} $\\

\end{longtblr}

UniSUF uses Vehicle Unique Update Packages (VUUP) to install vehicle updates. 
A VUUP is an update package produced by UniSUF for a specific vehicle. 
It contains all the necessary cryptographic keys, certificates, and instructions for the software update.
The internal structure of a VUUP file is shown in~\Cref{fig:vuup}.
The VUUP does not contain the actual software files; instead, \DI are included to specify where software files can be downloaded.
The specific implementation of \DI is unspecified but can be seen as URLs to the software update files. 
Note that the \DI is encrypted by a unique session key coupled to its update round. 
This key is retrieved from the Download Instruction Key Manifest (\DKM).

As shown in ~\Cref{fig:encryption}, a key manifest consists of a symmetric session key that has been asymmetrically encrypted, accompanied by a policy defining the key's usage (see \DKM, \IKM, and \MKM). 
Note that the UniSUF term session is equivalent to update rounds (see \Cref{sec:update-rounds}). 
In \Cref{tab:crypto-materials}, none of the key manifests are signed, to remain consistent with the notation used by~\citet{UniSUF}. However, when the key manifests are transmitted during tasks (see \Cref{ch:subproblems}), all key manifests are signed with the certificates according to \Cref{fig:signing}.

Additionally, a VUUP includes \II encrypted by the session key from the Installation Instruction Key Manifest (\IKM).
The \II contains the diagnostic instructions for installing software, the Master Key Manifest (\MKM), and the Secure Key Array (\SKA), as shown in~\Cref{tab:crypto-materials}.
To guide task-specific requirements of Confidential Secrets (see~\Cref{req:confidential-secrets}), we define the notations $\MKM_{Key}$ and $\SKA_{Key}$ to refer to all keys in \MKM and \SKA respectively. 

Thus, the DKM and IKM session keys are packaged into key manifests (i.e., a master key plus a policy). In contrast, the MKM can contain multiple master keys, whereas the DKM and IKM each contain only one.

\citet{kstrandberg} defines encased software as software that undergoes a multi-layered protection process. Initially, the software is signed by the software supplier to ensure its integrity and authenticity. Next, as part of UniSUF, it is encrypted to ensure confidentiality. Finally, the encrypted software is secured with an additional signature, ensuring that the encrypted package can be validated to avoid initiating the decryption process in case of validation failure. The word encased should not be mixed up with the term encapsulated used by \citet{UniSUF} to represent materials contained in the VUUP file. 

\begin{figure}[ht]
    \centering
    \includegraphics[width=\columnwidth]{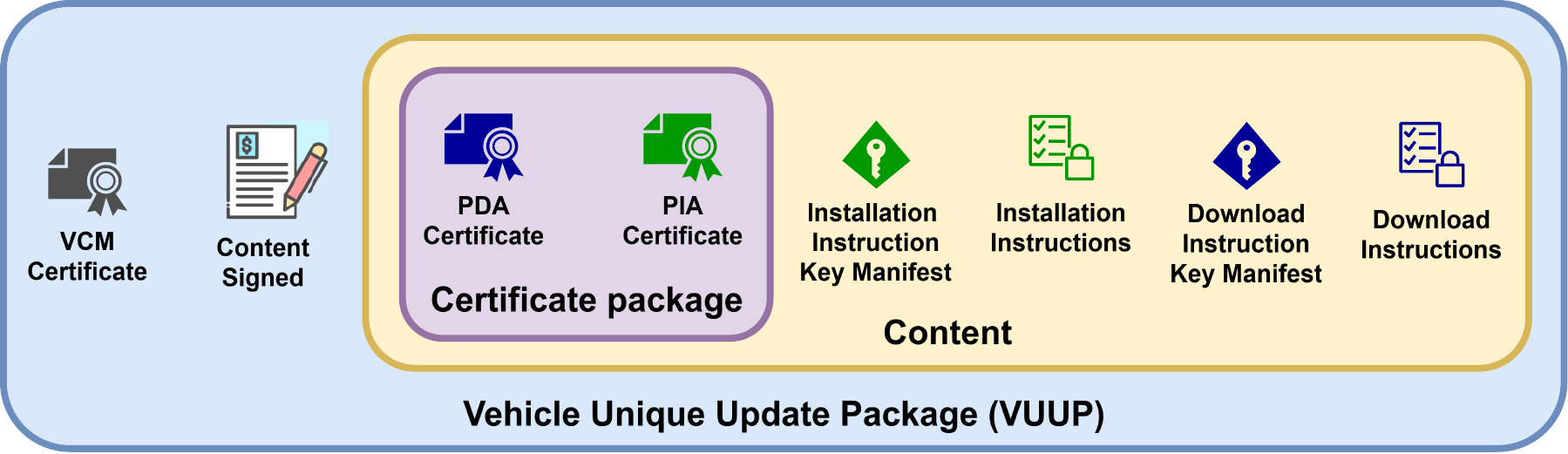}
    \caption{Internal structure of a VUUP file. The blue items are used in the download process, while the green ones are used for the installation process. Note that the VUUP content has been signed by $\VCM_{Cert}$. The figure is derived from \citet[Fig. 3]{UniSUF}.}
    \label{fig:vuup}
\end{figure}

\begin{figure}[ht]
    \centering
    \includegraphics[width=\columnwidth]{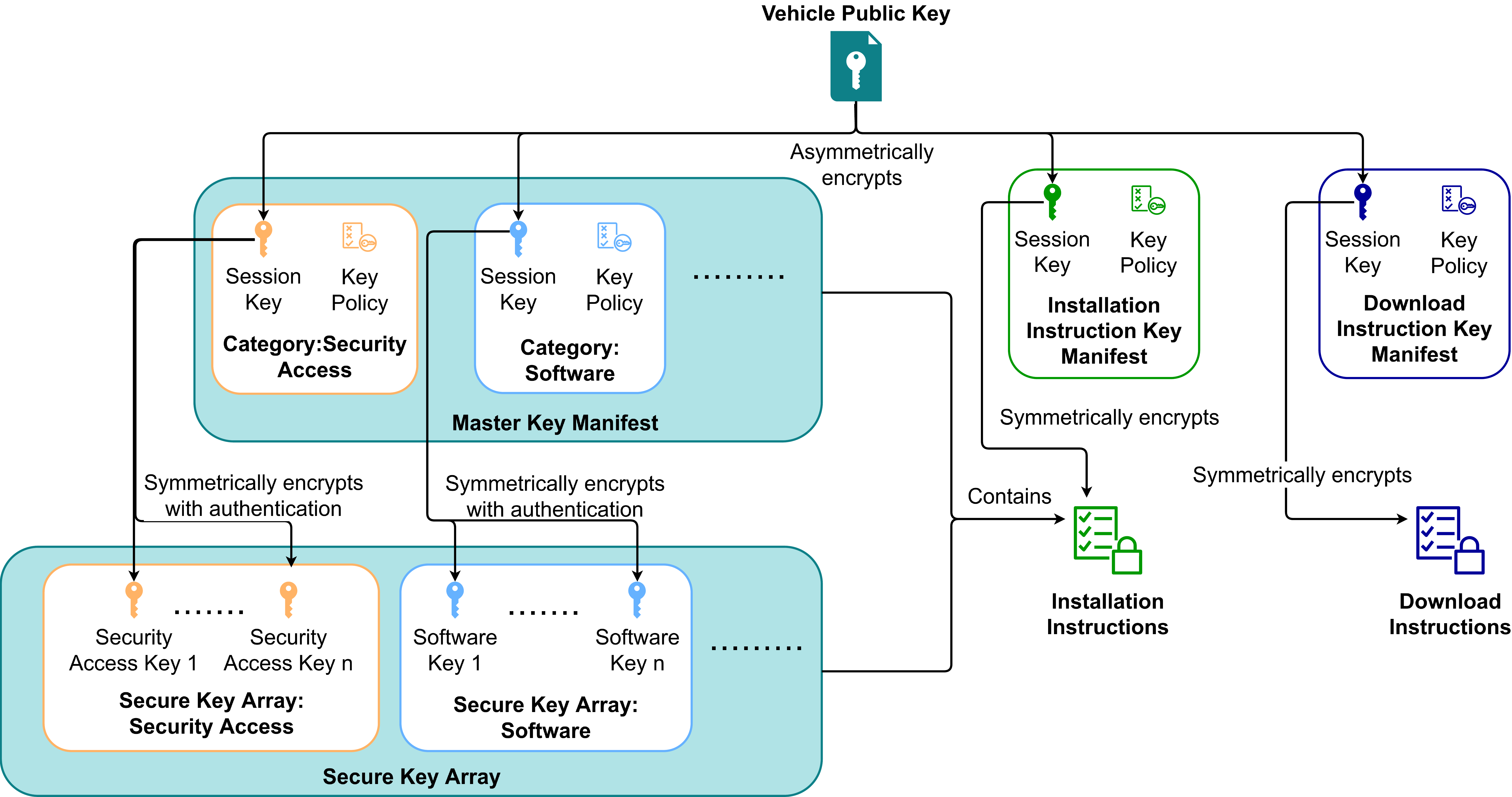}
    \caption{Different cryptographic materials in UniSUF and how they are encrypted. Figure derived from \citet[Fig. 2]{UniSUF}.}
    \label{fig:encryption}
\end{figure}

UniSUF uses the term category to classify master keys based on purpose, such as decrypting software files, or unlocking ECUs to enable update capabilities. 

The Vehicle Signed Order (VSO) is a readout per vehicle that contains detailed information about the vehicle, such as the onboard software versions.
UniSUF uses this information and the latest available software versions to construct a software list containing the software files and configurations for a specific and vehicle unique software update.
\DI and \II are then created based on the software list.

The software itself is located in external sources and not present in the actual VUUP file. 
The software files are signed by the software suppliers and associated with version numbers to prevent installations of older software versions \citep{kstrandberg}. 
UniSUF validates the supplier signature, further encrypts the software, and appends another signature.
Finally, the signed encrypted software is uploaded to the software repository.

%% file: system_entities.tex
\subsection{System Entities}
\label{sec:entities}
As shown in~\Cref{fig:entity_diagram}, UniSUF consists of three entities: Producer, Consumer, and the Software Repository~\citep[Sec. 4]{UniSUFAnalysis}. Additionally, UniSUF interacts with external entities, such as Software Suppliers and ECUs~\citep{kstrandberg}. UniSUF uses redundant entities and interacts with multiple Software Suppliers and vehicles with multiple ECUs ~\citep{UniSUF, UniSUFAnalysis}. 
However, for our proof, we consider a simplified system in which each vehicle has exactly one Consumer and one ECU. Additionally, all vehicles communicate with exactly one Producer and one Software Repository, and there exists only one Software Supplier.

\begin{figure}[ht]
    \centering
    \includegraphics[width=\columnwidth]{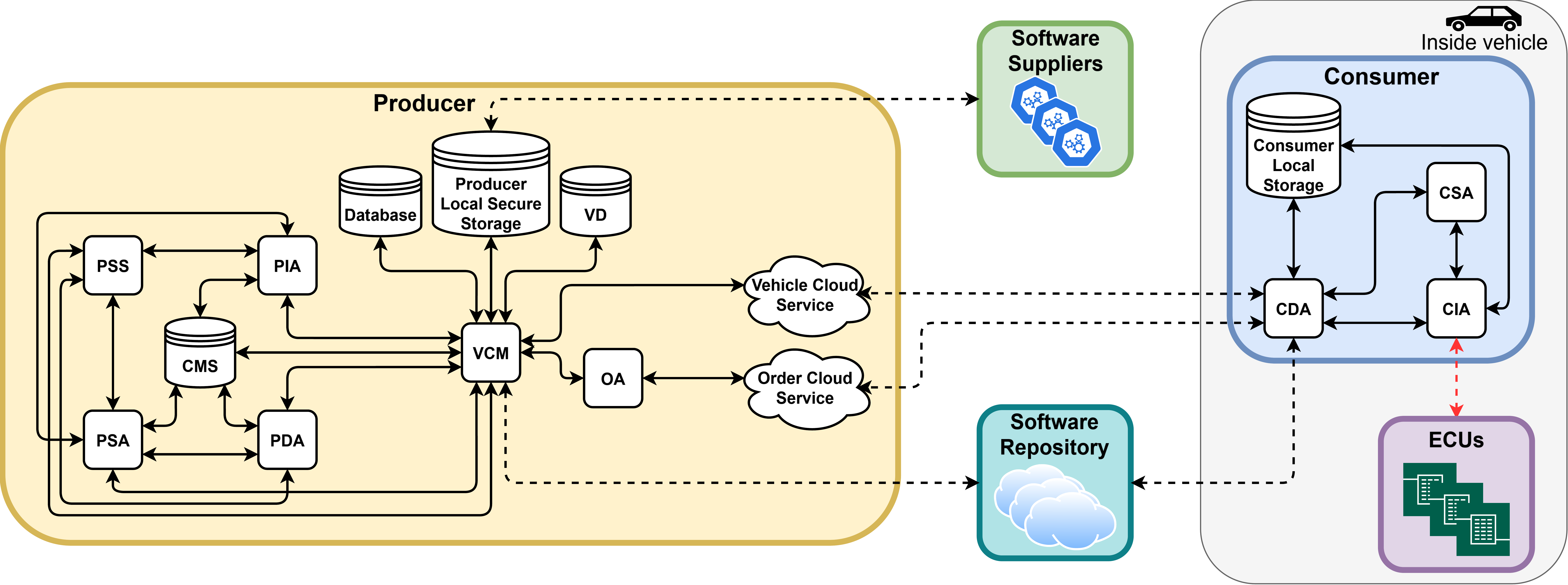}
    \caption{Diagram of all the entities and their communication in UniSUF. Dotted arrows denote communication channels between entities from different modules. The black arrows denote secure communications, while the sole red arrow denotes an insecure communication link.}
    \label{fig:entity_diagram}
\end{figure}

\subsubsection{Software Repository}
The UniSUF Software Repository is an entity that represents multiple distributed repositories \citep{UniSUF}, mainly responsible for software storage, where each software file is associated with a specific download URL. However, in offline cases, the software can also be stored on Network-Attached Storage (NAS) or a USB stick ~\cite {kstrandberg}. 

\subsubsection{Producer}
\begin{figure}[ht]
    \centering
    \includegraphics[width=\columnwidth]{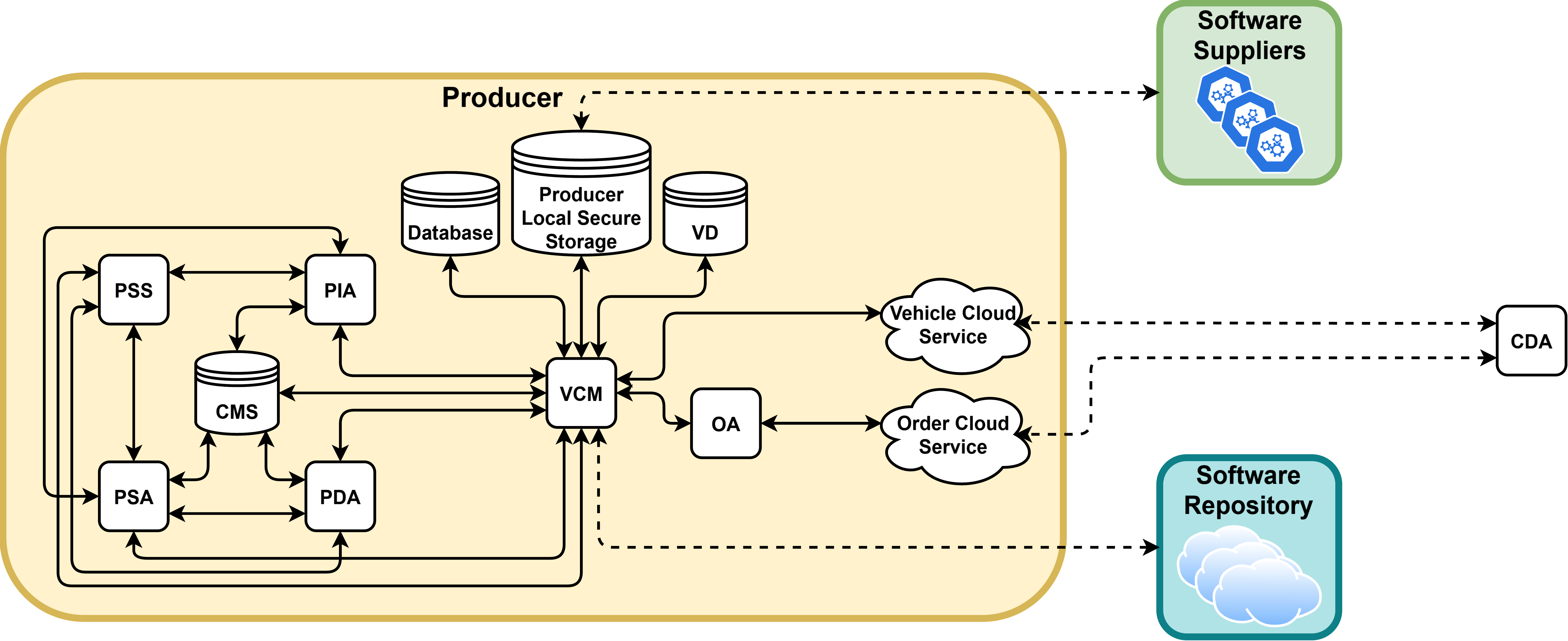}
    \caption{The communication flow between the Producer entities. Dotted arrows denote communication channels between a Producer entity and a non-Producer entity.}
    \label{fig:producer_flow}
\end{figure}

\citet{UniSUF} define the Producer as a collection of different sub-entities responsible for producing and securing software update packages.  
The Producer is also responsible for disseminating software to different storage repositories. 
As shown in~\Cref{fig:producer_flow}, UniSUF has the following 12 producer entities (cf. \cite{UniSUF}, Table 1).

\begin{itemize}
    \item \textbf{Producer Local Secure Storage} -- Stores software files received from software suppliers.
    
    \item \textbf{Version Control Manager (VCM)} -- Coordinates the producer entities and finalizes the creation of the VUUP file for a specific vehicle.

    \item \textbf{Producer Signing Service (PSS)} -- Produces signatures on behalf of other entities.
    
    \item \textbf{Cryptographic Material Storage (CMS)} -- Securely stores cryptographic materials, such as keys for decrypting software or unlocking ECUs, as well as certificates. 
  
    \item \textbf{Producer Security Agent (PSA)} -- Generates session keys and retrieves additional keys from the \CMS, such as keys for unlocking ECUs, performing privileged diagnostic requests, and decrypting software. These additional keys are further encrypted with session keys, which are, in turn, encrypted using a vehicle-specific public certificate.
    
    \item \textbf{Database} -- Store URLs to software files located in software repositories.

    \item \textbf{Order Cloud Service} -- Stores the Vehicle Signed Order (VSO) in a queue and URLs to VUUP files.
    
    \item \textbf{Order Agent (OA)} -- Verifies the validity of incoming VSOs and starts the updating process by forwarding the request to the Version Control Manager (\VCM). 
    
    \item \textbf{Producer Download Agent (PDA)} -- Creates download instructions from the software list received from \VCM. Later in the process, the download instructions are encrypted, signed, and sent to \VCM.
    
    \item \textbf{Producer Installation Agent (PIA)} -- Creates installation instructions from the software list received from \VCM. The installation instructions are bundled with cryptographic material, encrypted, signed, and sent to \VCM.
    
    \item \textbf{VIN Database (VD)} -- Stores data about unique vehicles and software versions. 

    \item \textbf{Vehicle Cloud Service} -- Stores the VUUP files and \VCM certificates that can be downloaded by the Consumer via a VUUP URL.
    
\end{itemize}
 
\subsubsection{Consumer}
\label{sec:consumer}
\begin{figure}[ht]
    \centering
    \includegraphics[width=0.8\columnwidth]{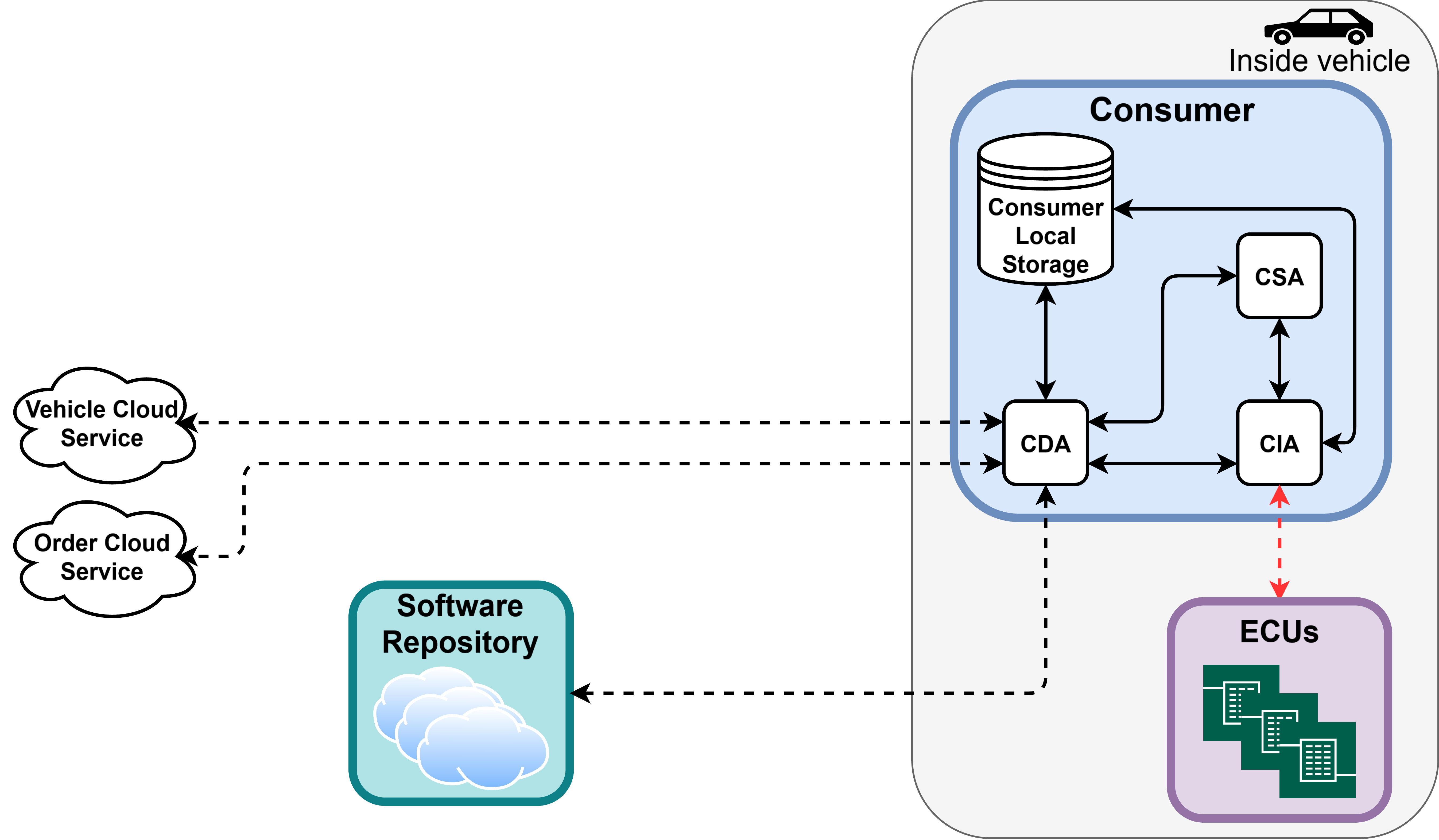}
    \caption{The communication flow between the Consumer entities is denoted by the solid arrows. Dotted arrows denote communication between a Consumer and a non-Consumer entity. The red arrow denotes an insecure communication channel. All other channels are secure.}
    \label{fig:consumer_flow}
\end{figure}

The Consumer is the Producer's counterpart and is responsible for decapsulating the VUUP and the processing of the download and installation instructions~\citep[Sec. 4.2]{UniSUF}. 
UniSUF has the following four Consumer entities as shown in~\Cref{fig:consumer_flow}.

\begin{itemize}
    \item \textbf{Consumer Local Storage} -- Stores signed VUUP, and signed and encrypted software files.

    \item \textbf{Consumer Download Agent (CDA)} -- Executes the download instructions and retrieves software from software repositories.  
    
    \item \textbf{Consumer Security Agent (CSA)} -- Provides a trusted execution environment where, e.g., decryption can occur in an isolated and secure space.
    
    \item \textbf{Consumer Installation Agent (CIA)} -- Executes the installation instructions and then streams the decrypted software to the unlocked ECUs with the help of \CSA.

\end{itemize}

\subsubsection{Software Suppliers} 

Software suppliers create and deliver software to the \PA for the installation in different vehicles~\citep[Sec. 3.3]{UniSUF}. 
The software suppliers also sign their software to provide authenticity.
We assume a simplified model that considers a single software supplier supplying software to a single ECU (see \Cref{sec:ecus}). 

\subsubsection{The Electronic Control Unit}
\label{sec:ecus}
An Electronic Control Unit (ECU) is a vehicle computer responsible for various tasks, from simple signal processing to more advanced functionality, for instance, an infotainment system running various applications.
For our simplified model, we assume a system with only one ECU.
The ECU is first unlocked and put in programming mode by using security access and a secret key~\cite[Sec. 4.2]{UniSUF}, to allow the ECU to receive and install software with Unified Diagnostic Services (UDS)~\citep{kstrandberg, IEEEref:ISO14229}. 

\subsubsection{Adversary}
As shown in \Cref{fig:architecture-with-adversary}, the adversary is based on the Dolev-Yao model, assuming an adversary with access to the communication channels. 

The adversary is actively present on all communication links in the system.
However, most of these links, except the one depicted in red, are secure and reliable communication channels; that is, the adversary cannot interfere, according to Dolev-Yao.
The one exception is the link between the Consumer and the ECUs, which is a reliable but not a secure communication channel, where the adversary can potentially read, modify, delete, or insert messages.

\begin{figure}[ht]
    \centering
    \includegraphics[width=\columnwidth]{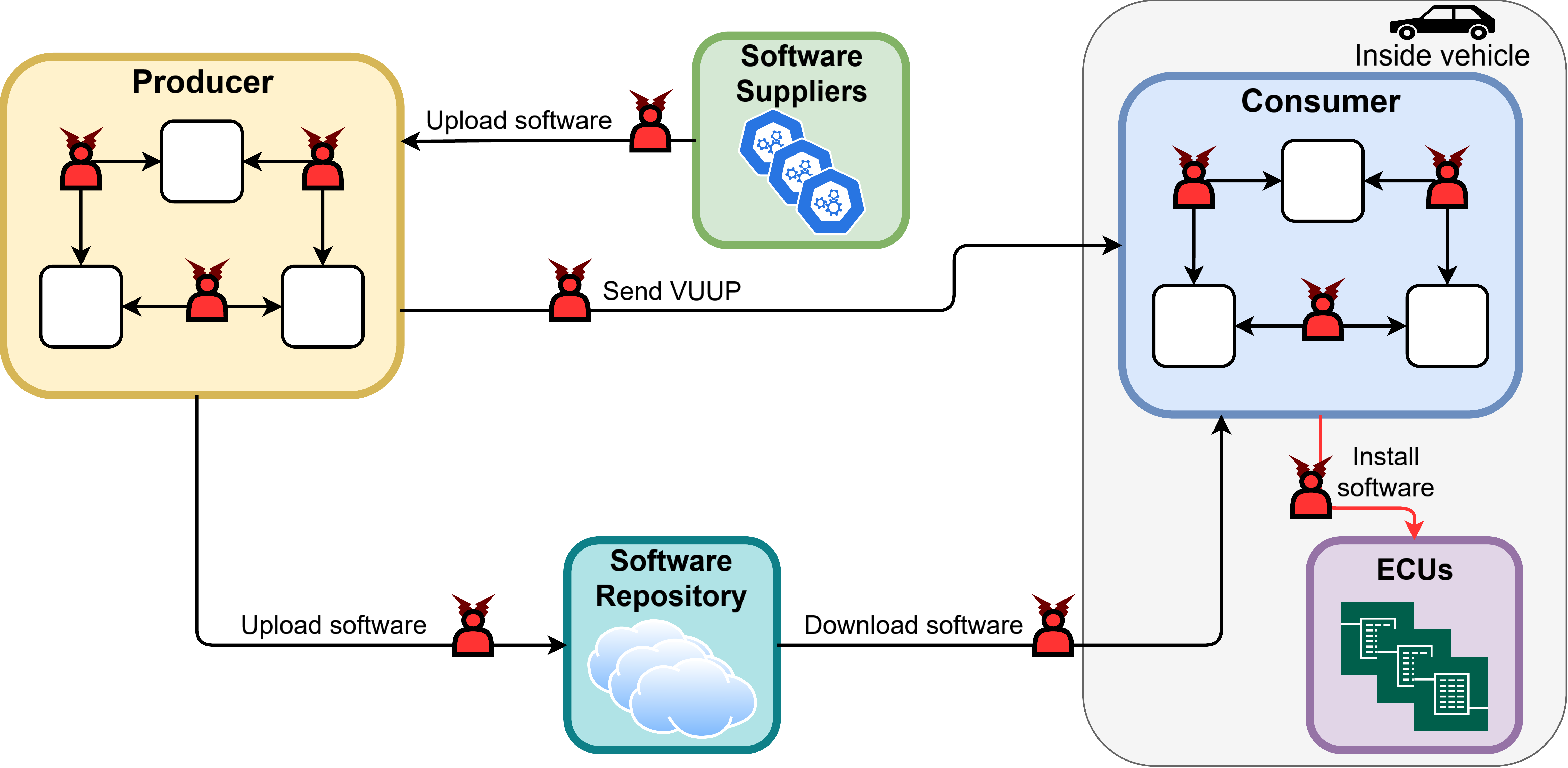}
    \caption{The illustration depicts the presence of adversaries on the high-level UniSUF architecture shown in \cref{fig:highlevel-arch}, where the adversaries are represented by the red devils.}
    \label{fig:architecture-with-adversary}
\end{figure}

%% file: modelling_unisuf.tex
\subsection{Modelling UniSUF}
\label{sec:modelling-unisuf}
This section explains the modeling of entities and their respective executions in UniSUF.

\subsubsection{Execution of System Entities}
Each system entity runs a sequence of tasks for a given problem.
The entity running the initial user-invoked task is called the \textit{initiator}; illustrated in \Cref{fig:intitiator_entity}. 
All other entities are listeners (see \Cref{fig:listening_entity}) since their first task is the listening task. This task listens for an initiation message specified for each listener.
The listening task invokes the next task in
sequence and a new listening task, enabling concurrent executions of the task sequence. We divide listeners into two categories: a passive listener that will wait for an initiation message and an active listener that will request an initiation message.  

\begin{figure}[!ht]
    \centering
    \includegraphics[width=\columnwidth]{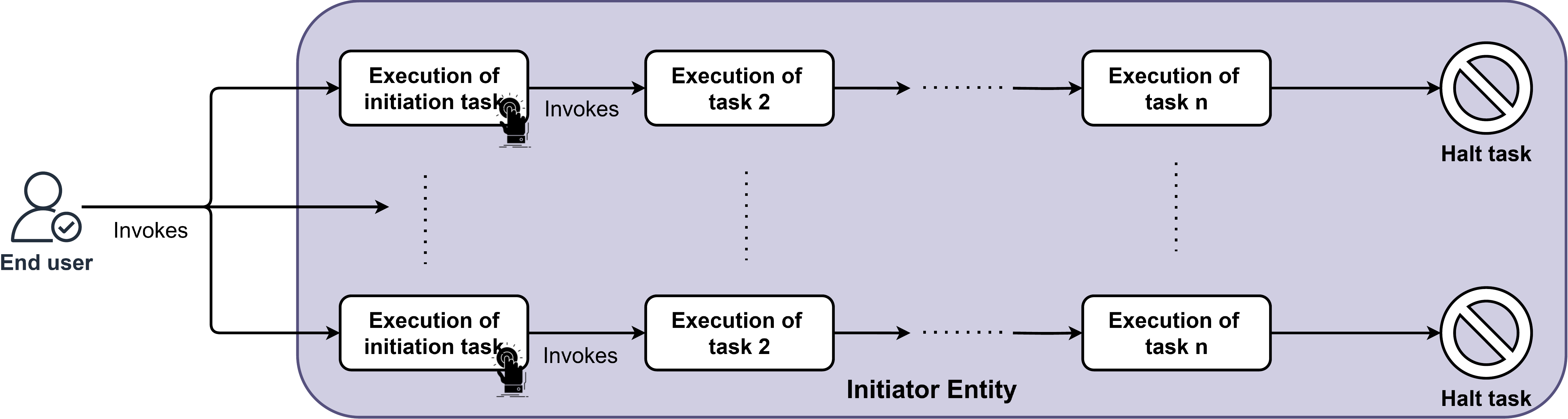}
    \caption{The initial user-invoked task begins an execution of the initiator entity's tasks. The initiation tasks are marked with the symbol of a hand pressing a button.}
    \label{fig:intitiator_entity}
\end{figure}

\begin{figure}[htbp]
    \centering
    \includegraphics[width=\columnwidth]{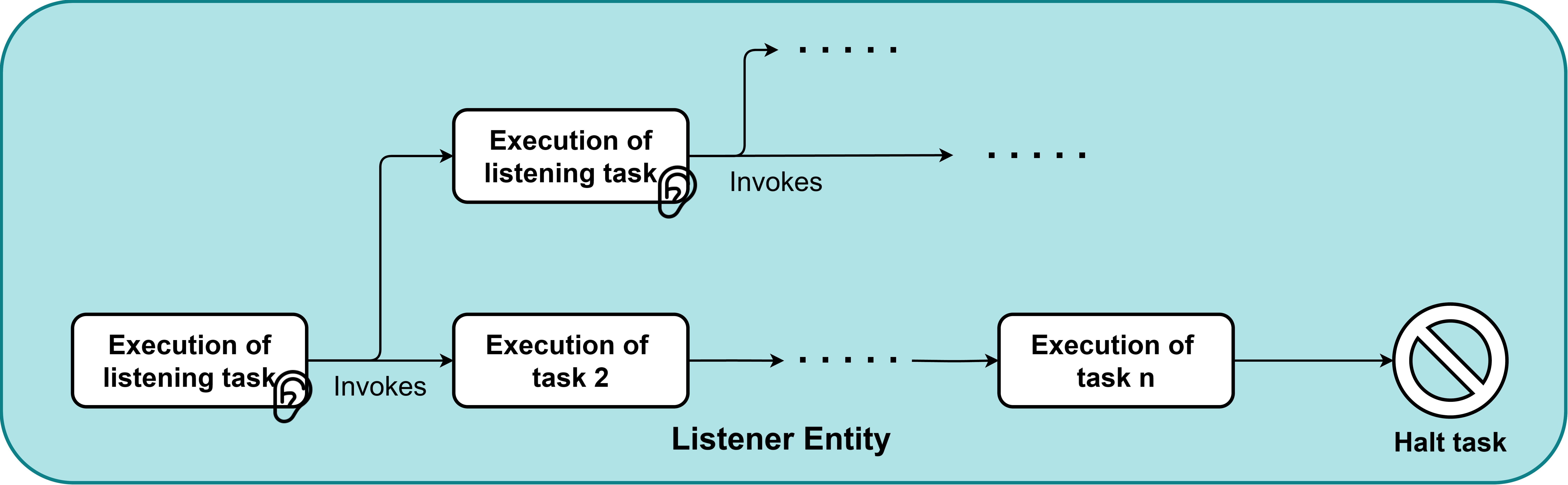}
    \caption{The listening task invokes a copy of itself to continue listening on a new execution. It also begins an execution of the listening entity's tasks. The listening tasks are marked with a symbol of an ear.}
    \label{fig:listening_entity}
\end{figure}

\subsubsection{Lifecycle of Update Rounds}
\label{sec:lifecycle-of-update-rounds}
An update round execution begins when the round identifier is created during the initiation task. Throughout this round, listener entities start executing once their listening task receives a message containing the round identifier. 

As shown in \Cref{fig:initiator_entity_context,fig:listening_entity_context}, an entity maintains a context for each update round.
This context contains the update round identifier and the cryptographic materials (see \Cref{sec:crypto-mat}) used throughout the execution of the round. The context is passed along the task sequence and can be updated by each task. 

\begin{figure}[!ht]
    \centering
    \includegraphics[width=\columnwidth]{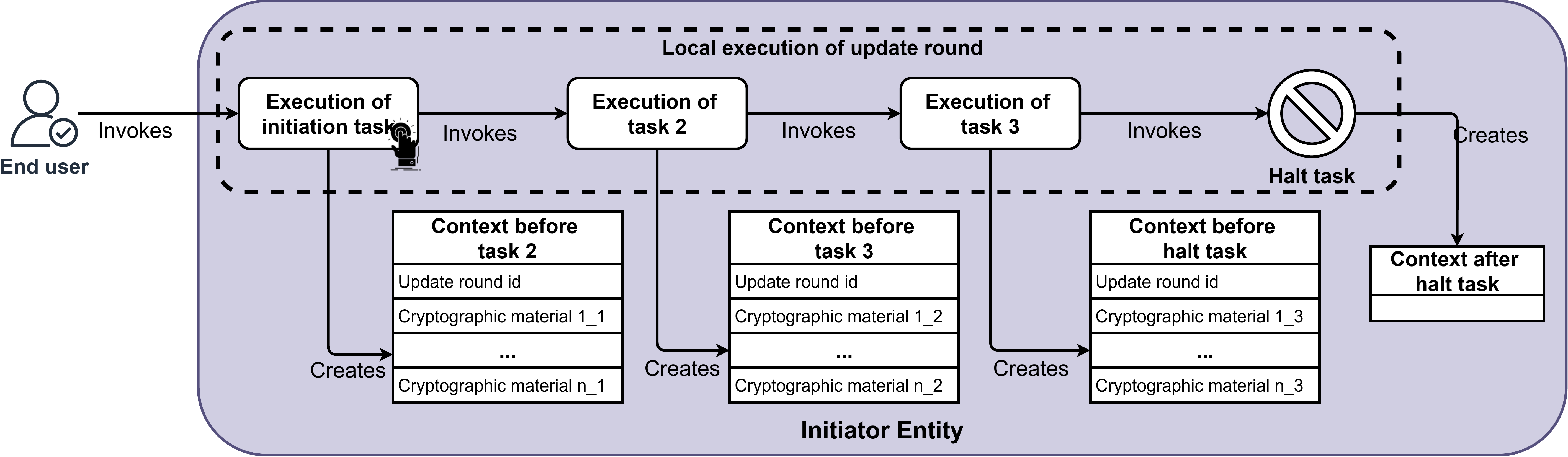}
    \caption{The figure illustrates an example of executing an initiator entity's task sequence. For each task, we also show what the context can contain.}
    \label{fig:initiator_entity_context}
\end{figure}

\begin{figure}[htbp]
    \centering
    \includegraphics[width=\columnwidth]{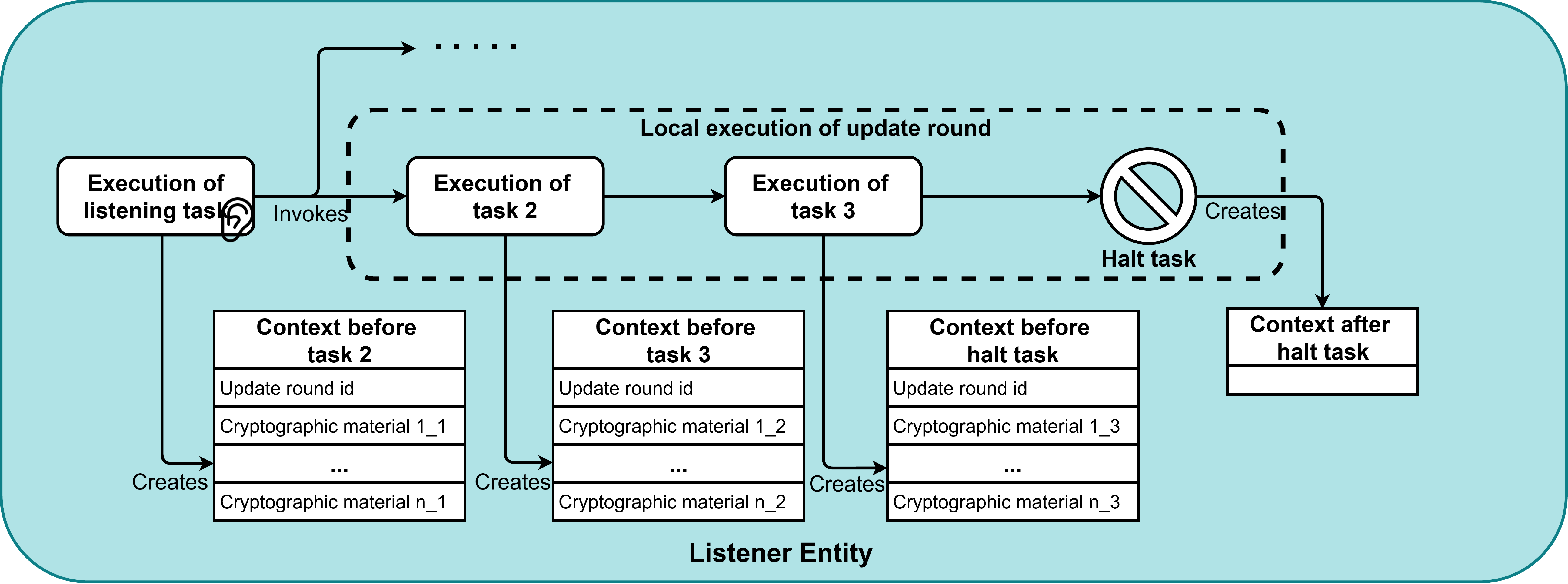}
    \caption{The figure illustrates an example of an execution of a listening entity's task sequence. For each task, we also show what the context can contain.}
    \label{fig:listening_entity_context}
\end{figure}

Each update round identifier encodes its expiration time 
and when such expiration occurs, all entities \textit{halt} their local execution of the update round,
by removing the context associated with the round and ignoring any further messages related to that round.
The update round is terminated once all entities have halted their execution of the update round. 

\subsubsection{Passing Contexts Across Segments of Task Sequences}
\label{sec:passing-contexts-task-sequences}

We divide the main problems into sub-problems, where each sub-problem uses a subset of the system entities.
In \Cref{fig:problem_division}, \emph{sub-problem 1} uses the \emph{initiator} and \emph{listener entity 1}. \Cref{fig:problem_division} shows that an entity's task sequence can be segmented so that each segment belongs to a single sub-problem. For example, \emph{listener entity 1}'s \emph{task 3} and \emph{halt task} form a segment that belongs to \emph{sub-problem 2}.

\begin{figure}[htbp]
    \centering
    \includegraphics[width=\columnwidth]{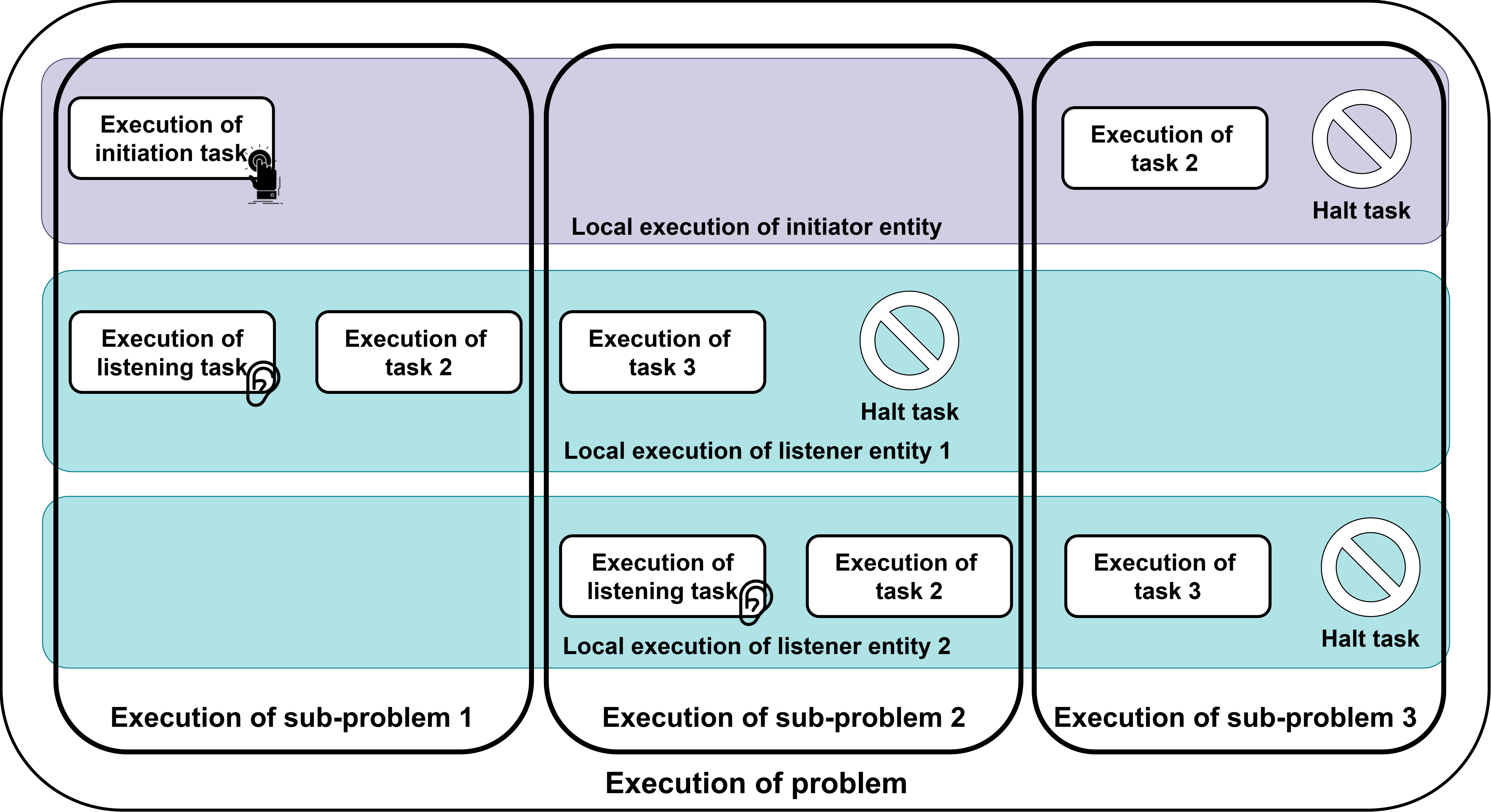}
    \caption{The execution of an example problem. This problem considers three entities, whereas one is the initiator entity. The problem is divided into three sub-problems.}
    \label{fig:problem_division}
\end{figure}

When an entity starts working on a sub-problem for an update round, it possibly already has an update round context. For instance, if the entity has previously participated in other sub-problems for the same update round, i.e., it has previously run either the initiation task or listening task for the same update round. 

\Cref{fig:problem_division_context} shows that \emph{listening entity 2} has a context at the start of the sub-problem since it starts on \emph{task 2}, i.e., it has previously run the \emph{listening task}. The other listener, \emph{listener entity 1}, does not start with any context, since it first needs to run its listening task at the start of the sub-problem.

Once a segment's execution finishes, the last context of this segment is passed to the next segment as the starting context.
Therefore, we identify which entities have previously executed task segments in the same update round for a given sub-problem.
For these entities, we specify the cryptographic materials present in their starting contexts (see \Cref{ch:subproblems}).

\begin{figure}[htbp]
    \centering
    \includegraphics[width=\columnwidth]{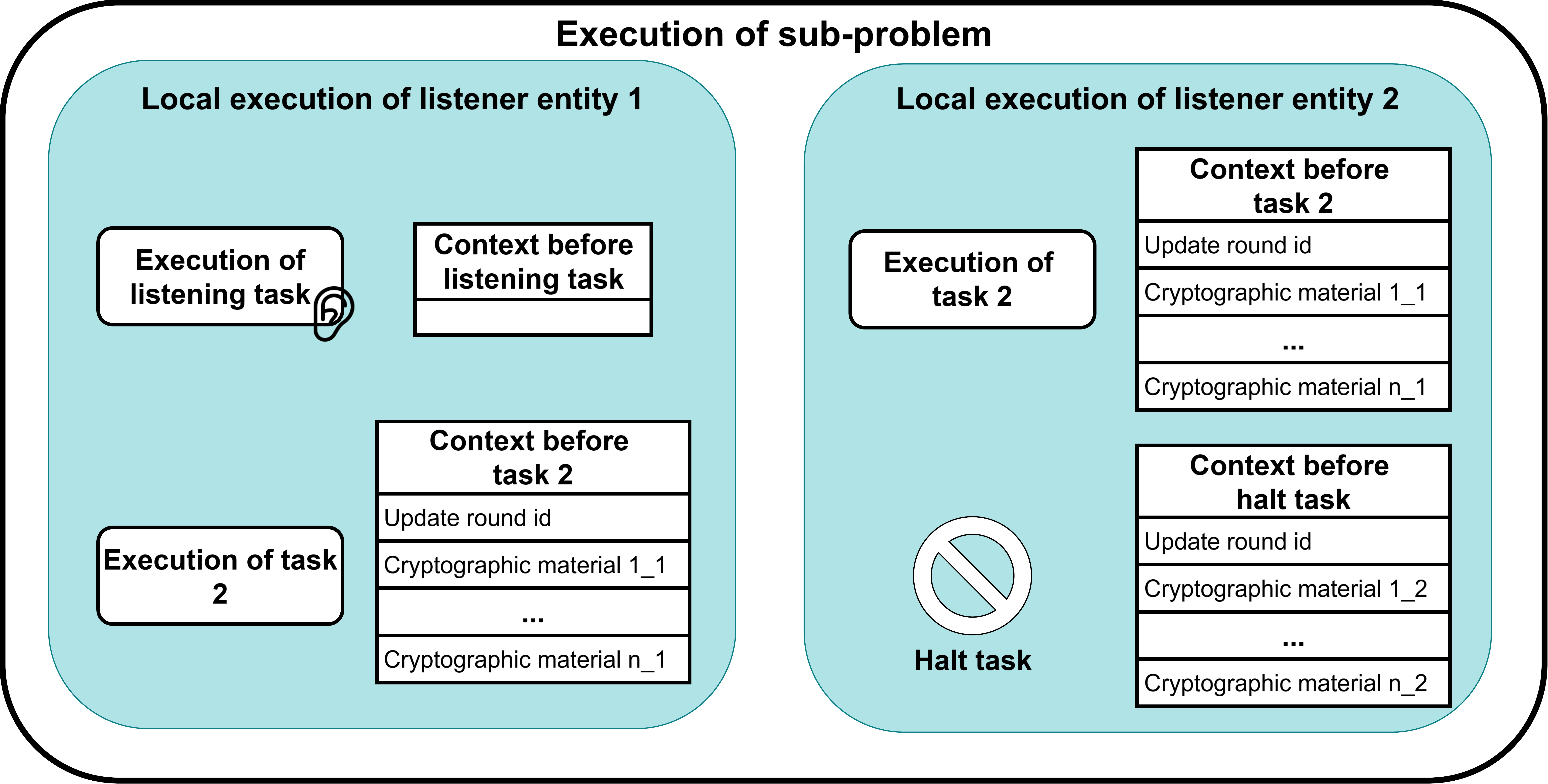}
    \caption{Example of two entities solving a sub-problem together. The figure also shows how each task in the entities is connected to a context.}
    \label{fig:problem_division_context}
\end{figure}

\subsubsection{Well-Known Addresses of UniSUF Entities}

We assume that all consumer entities, as defined in \Cref{sec:consumer}, are aware of the specific Vehicle Identification Number (VIN) of the vehicle they occupy. Each vehicle has a pre-stored public vehicle certificate containing metadata, including VIN-related information~\cite {kstrandberg}.

Additionally, we assume that all entities in UniSUF know each other's addresses, e.g., through existing protocols, such as DNSSEC, enabling entities to securely obtain IP addresses from domain names~\cite {dnssecAnalysis}. 

%% file: subproblems.tex
\section{Sub-Problems}
\label{ch:subproblems}

From the UniSUF specifications~\citep{UniSUF}, we analyse two problems: the \emph{software preparation} and the \emph{software update}.
The software preparation process (cf. \Cref{sec:preparation}) involves safeguarding software received from suppliers, ensuring the software's confidentiality and authenticity for it to be securely incorporated into future vehicle updates.
Additionally, the software update is further divided into the \emph{encapsulation} and the \emph{decapsulation} stages (cf. \Cref{sec:encapsulation,sec:decapsulation}) emphasizing securely updating vehicular software and configurations.
\citet{UniSUF} specifies an additional stage after decapsulation: the post-state.
The post-state encompasses installation reports and logs, potentially affecting upcoming software updates. However, we do not consider the post-state in our simplified model. 

Consequently, the main tasks in UniSUF are dissected into sub-problems, based on the steps provided by \citet{UniSUF}, where each sub-problem, has a description, a diagram depicting its algorithm, a communication scheme, and assumptions and requirements.

\input{preparation}
\input{encapsulation}
\input{decapsulation}

%% file: preparation.tex
\subsection{Preparation}
\label{sec:preparation}
An overview of the entities involved in the preparation stage and their communication links can be seen in \Cref{fig:preparation}. 

In the preparation stage, software supplier files are processed before being used for software updates \citet{UniSUF}. We have divided this stage into two sub-problems: the \textit{Secure Software Files} and the \textit{Upload Software Files}. The focus of the first sub-problem is the encrypting and signing of the software, whereas the second sub-problem handles the software upload and the creation of software URLs. 

The preparation stage manages both software files applicable to multiple vehicles and unique files for specific vehicles \citet{kstrandberg}.
In our model, we assume the case when software is being prepared for multiple vehicles and we additionally assume that the update round is solely identified by the expiration time $t_e$ (see \Cref{sec:update-rounds}). 

\begin{figure}[ht]
    \centering
    \includegraphics[width=\columnwidth]{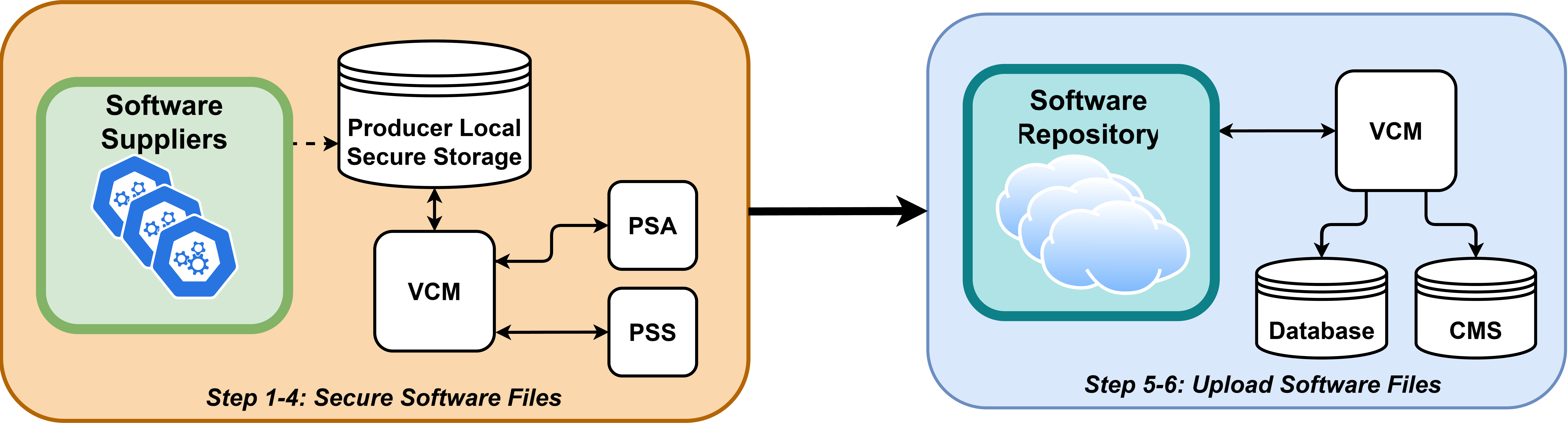}
    \caption{An overview derived from \cite[Sec. 3.3]{UniSUF} for the communication between all entities involved in the preparation stage further divided into two sub-problems.}
    \label{fig:preparation}
\end{figure}

\subsubsection{Step 1--4: Secure Software Files}
\label{sec:prp-1-4}
The initial phase of the update process focuses on securing software files (see \Cref{fig:secure_software_files,tab:secure_software_files}). The software supplier signs the software files before sending them to local storage on the producer side. \VCM validates the signature and encrypts the software using a symmetric key obtained from \PSA; this key is referred to as $Software_{Key}$. The encrypted software is then signed with the \VCM certificate to finalize the $\SW_{Encapsualted}$ assembly.

\begin{figure}[ht]
    \centering
    \includegraphics[width=\columnwidth]{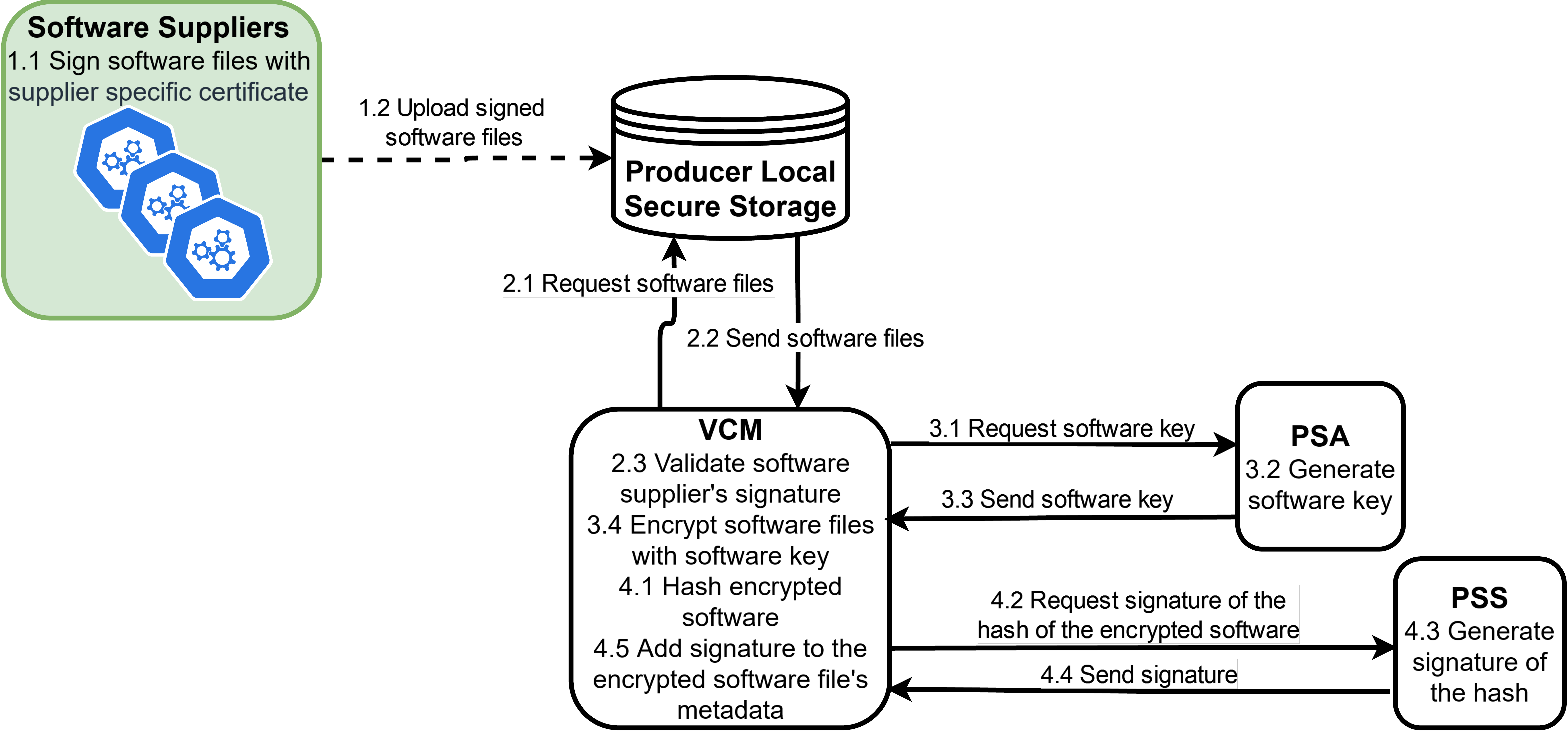}
    \caption{Diagram of the sub-problem \textit{Secure Software Files}.}
    \label{fig:secure_software_files}
\end{figure}

\begin{table}[ht]
    \centering
    \begin{tabular}{lll}
    \midrule
    (1.2)& $ \Prod \to \PLS $ & $ [\SW]_{Supplier} $ \\
    \\
    \hdashline
    \\
    (2.1)& $ \VCM \to \PLS $ & $ \Req(\SW) $ \\
    
    (2.2)& $ \PLS \to \VCM $ & $ [\SW]_{Supplier} $ \\

    (3.1)& $ \VCM \to \PSA $ & $ \Req(\SW_{Key})$ \\
    
    (3.2)& $ \PSA \to \VCM $ & $ \SW_{Key} $ \\
    
    (4.2)& $ \VCM \to \PSS $ & $ SoftwareHash$ \\
    && $SoftwareHash$ = $Hash($ \\
    && $ SymEnc([\SW]_{Supplier},$ \\
    && $\newline\phantom{xd}\SW_{Key})) $ \\
    
    (4.4)& $ \PSS \to \VCM $ & $SignedSoftwareHash$\\
    && $SignedSoftwareHash$ = \\
    && $\newline\phantom{xd}Sign(SoftwareHash,$\\
    && $\newline\phantom{xdxd} \VCM_{\SK}) $\\
    
    \midrule
    \end{tabular}
    \caption{Communication scheme for the sub-problem \textit{Secure Software Files}. The dashed line represents parallel processes.}
    \label{tab:secure_software_files}
\end{table}

\begin{center}
    \begin{align*}
        S = & \{\SW, \ \SW_{Key}, \ Supplier_{\SK}, \ VCM_{\SK}\} \\
        \mathcal{D} = & \{(\SW, \ \Prod), \ (\SW_{Key}, \ \PSA), \  (SoftwareHash, \ \VCM), \\
        & \ (SignedSoftwareHash, \ \PSS)\} \\
    \end{align*}
\end{center}

To derive the specific sub-problem requirements from the system requirements (see \Cref{sec:problem-definition}), we define $S$, $\mathcal{D}$ and the partial order $\mathcal{P}(\ell) = \ell_i < \ell_{i+1}$ for $1 \leq i \leq 4$ with the following labels:
\begin{labelenumerate}
    \item The software suppliers upload the software to \PLS.
    \item \PSA generates the software key.
    \item \VCM generates a hash of the software.
    \item \PSS generates a signature for the software.
    \item \VCM assembles $Software_{Encased}$.
\end{labelenumerate}

\subsubsection{Step 5--6: Upload Software Files} 
Once the software file has been encased, it is uploaded into a software repository. The repository then generates a URL to the uploaded encased software and returns this URL to \VCM (see steps 5.1 and 5.2 in \Cref{fig:upload_software_files,tab:upload_software_files}). Finally, the \VCM stores the software URL in the VIN database and stores the $Software_{Key}$ in \CMS; used for the encryption of software files. 

Note that \VCM has a starting context because its initiation task is in a previous sub-problem (see \Cref{sec:prp-1-4}). \VCM's starting context contains $\SW_{Encased}$ and $\SW_{Key}$.

\begin{figure}[ht]
    \centering
    \includegraphics[width=0.6\columnwidth]{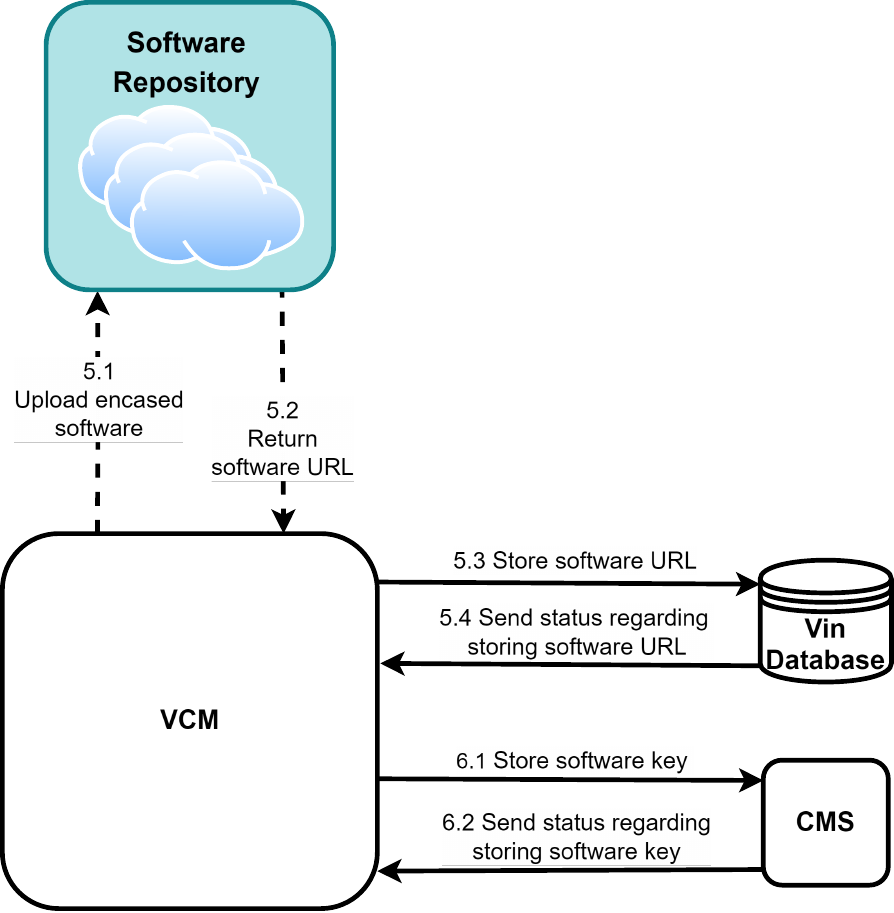}
    \caption{Diagram of the sub-problem \textit{Upload Software Files}.}
    \label{fig:upload_software_files}
\end{figure}

\begin{table}[ht]
    \centering
    \begin{tabular}{lll}
    \midrule

    (5.1)& $ \VCM \to \SR $   & $ \SW_{Encased} $ \\
    
    (5.2)& $ \SR \to \VCM $   & $ \SW_{URL}     $ \\
    
    (5.3)& $  \VCM \to VIN Database $  & $ \SW_{URL}     $ \\
    (5.4)& $  VIN Database \to \VCM$  & $ \Suc(\SW_{URL})     $ \\
    
    (6.1)& $  \VCM \to \CMS $   & $ \SW_{Key}     $ \\
    (6.2)& $ \CMS \to \VCM$   & $ \Suc(\SW_{Key})     $ \\

    \midrule
    \end{tabular}
    \caption{Communication scheme for the sub-problem \textit{Upload Software Files}.}
    \label{tab:upload_software_files}
 \end{table}

\begin{center}
    \begin{align*}
        S = & \{\SW, \ \SW_{Key}, \ Supplier_{\SK}, \ VCM_{\SK}\} \\ 
        \mathcal{D} = & \{(Software_{Encased}, \ \VCM), \ (\SW_{Key}, \ \PSA), \\
        &  \ (Software_{URL}, \ \SR)\} \\
    \end{align*}
\end{center}

\sloppy To derive the specific sub-problem requirements from the system requirements (see \Cref{sec:problem-definition}), we define $S$, $\mathcal{D}$ and the partial order $\mathcal{P}(\ell) = \ell_i < \ell_{i+1}$ for $1 \leq i \leq 3$ with the following labels:
\begin{labelenumerate}
    \item \SR receives the $Software_{Encased}$.
    \item VIN Database stores the $Software_{URL}$.
    \item \CMS stores the $Software_{key}$.
    \item \VCM receives status of $Software_{key}$ being stored in \CMS.
\end{labelenumerate}

%% file: encapsulation.tex
\subsection{Encapsulation}
\label{sec:encapsulation}
The encapsulation stage starts when the \OA has received an order request from \CDA, whereafter producer entities work collaboratively to produce a VUUP \citep[Sec. 4.1]{UniSUF}.
\Cref{fig:encapsulation} shows a high-level flow diagram of the different sub-problems involved in the encapsulation stage.
In steps 1--2, the encapsulation stage starts by producing a \VSO, which is later processed to a \SL in step 3.
Furthermore, in steps 4--7, the \SL is sent to \PDA, \PIA, and \PSA to generate necessary materials, for instance, download and installation instructions.
Instructions and other materials are included in a VUUP file (step 8), whereafter the encapsulation stage finalizes by notifying the consumer that updates are available (cf. steps 9--11).

\begin{figure}[ht]
    \centering
    \includegraphics[width=\columnwidth]{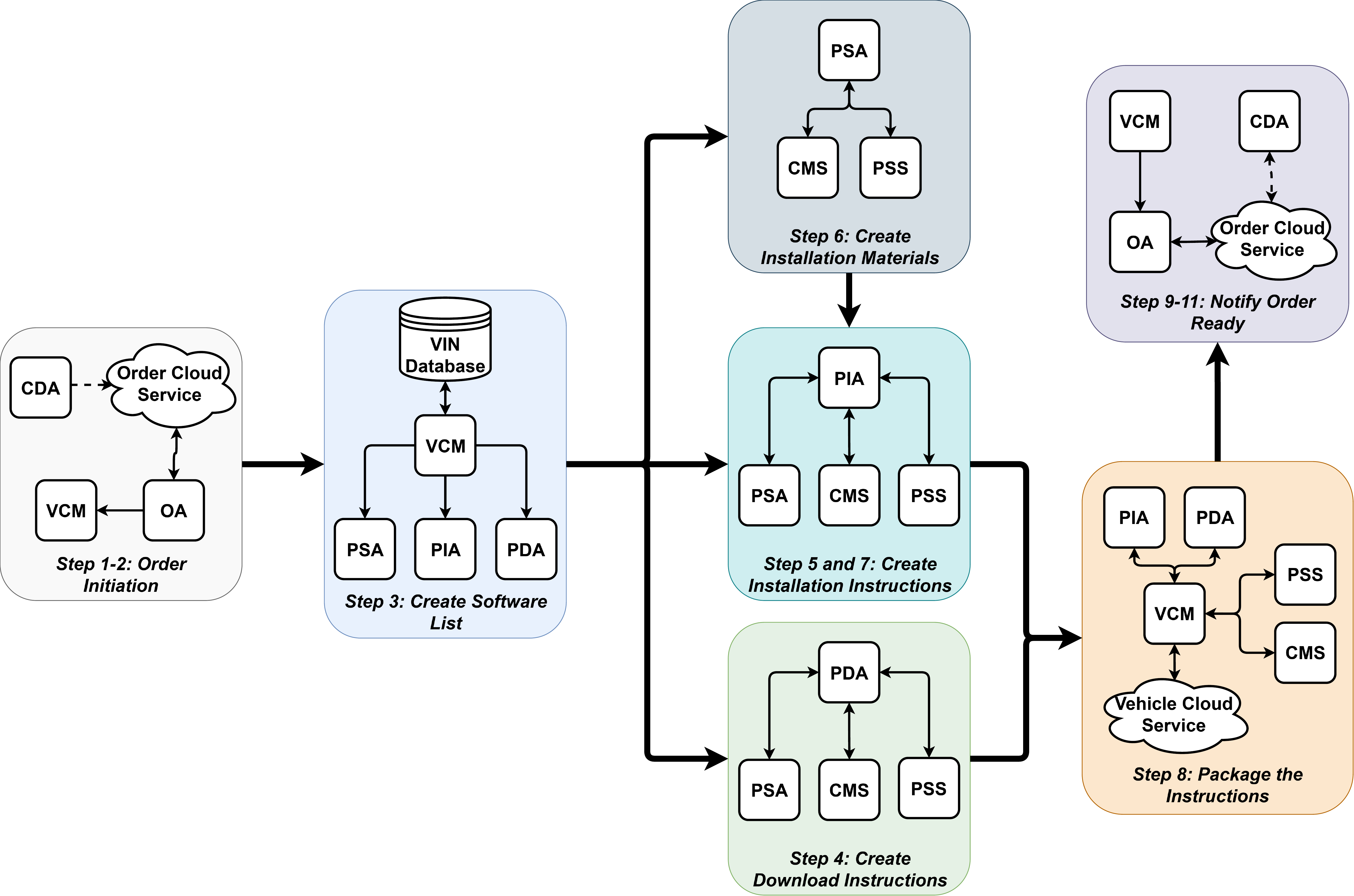}
    \caption{Overview of communications between all entities that are involved in the encapsulation stage.}
    \label{fig:encapsulation}
\end{figure}

\subsubsection{Step 1--2: Order Initiation}
\label{sec:stp-ecp-1-2}
The encapsulation process begins when the \CDA requests an update by sending a signed order (denoted as \VSO) to the \OCS (see steps 1.1--1.3 in \Cref{fig:order_initiation}). 
In step 1.4, \OCS stores the order in a queue, whereafter the \OA attempts to fetch an order from this queue.
If an order is available, \OCS sends a \VSO to \OA, which verifies the signature and initiates \VCM using the \VSO (steps 2.1--2.4).

\begin{figure}[ht]
    \centering
    \includegraphics[width=0.8\columnwidth]{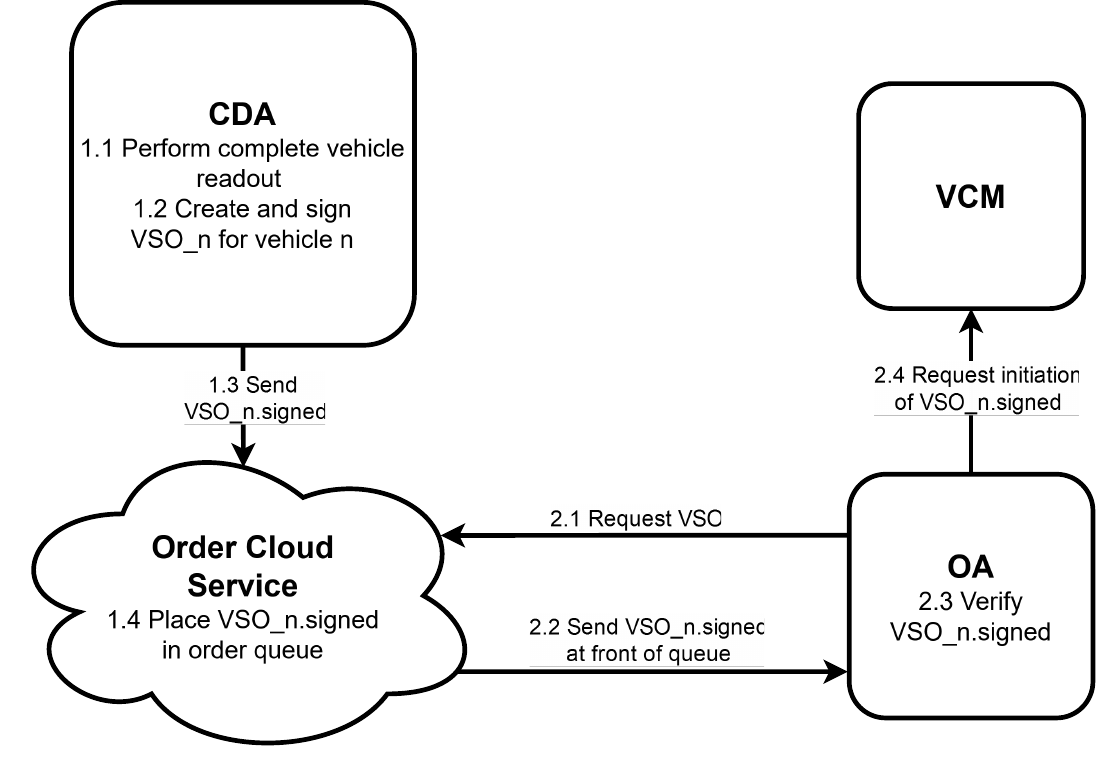}
    \caption{Diagram of the sub-problem \textit{Order Initiation}.}
    \label{fig:order_initiation}
\end{figure}

\begin{table}[ht]
    \centering
    \begin{tabular}{lll}
    \midrule
    (1.3)& $ \CDA \to \OCS $ & $ [\VSO]_{\CA} $ \\
    \\
    \hdashline
    \\
    (2.1) & $ \OA \to \OCS $ & $ \Req(\VSO)$ \\
    (2.2) & $ \OCS \to \OA $ & $ [\VSO]_{\CA} $ \\
    (2.4) & $ \OA \to \VCM $ & $ [\VSO]_{\CA} $ \\
    \midrule
    \end{tabular}
    \caption{Communication scheme for the sub-problem \textit{Order Initiation}. }
    \label{tab:order_initiation}
\end{table}

To derive the specific sub-problem requirements from the system requirements (see \Cref{sec:problem-definition}), we define $S = \{\CDA_{\SK}\}$, $\mathcal{D} =\{(\VSO,\ \CDA)\}$, and the partial order $\mathcal{P}(\ell) = \ell_i < \ell_{i+1}$ for $1 \leq i \leq 4$ with the following labels:
\begin{labelenumerate}
    \item \CDA generates a signed \VSO.
    \item \OCS stores the signed \VSO in its queue.
    \item \OA pulls signed \VSO from \OCS.
    \item \OA initiates \VCM with the signed \VCM.
    \item \VCM sends status on initialisation.
\end{labelenumerate}

\subsubsection{Step 3: Create Software List}
\label{sec:stp-ecp-3}
Given a \VSO, the \VCM will produce a software list (see step 3.1--3.5 in \Cref{fig:create_software_list}). 
The software list contains the software to be installed in the vehicle. 
After creating the software list, it is transmitted to the \PIA, \PSA, and \PDA for further processing (step 3.6).
The software list is used to generate the download and installation instructions.

The entity \VCM has a starting context because its listening task has been included in the previous sub-problem \textit{Order Initiation} (see \Cref{sec:stp-ecp-1-2}).
The starting context of \VCM contains the signed \VSO.

\begin{figure}[ht]
    \centering
    \includegraphics[width=0.8\columnwidth]{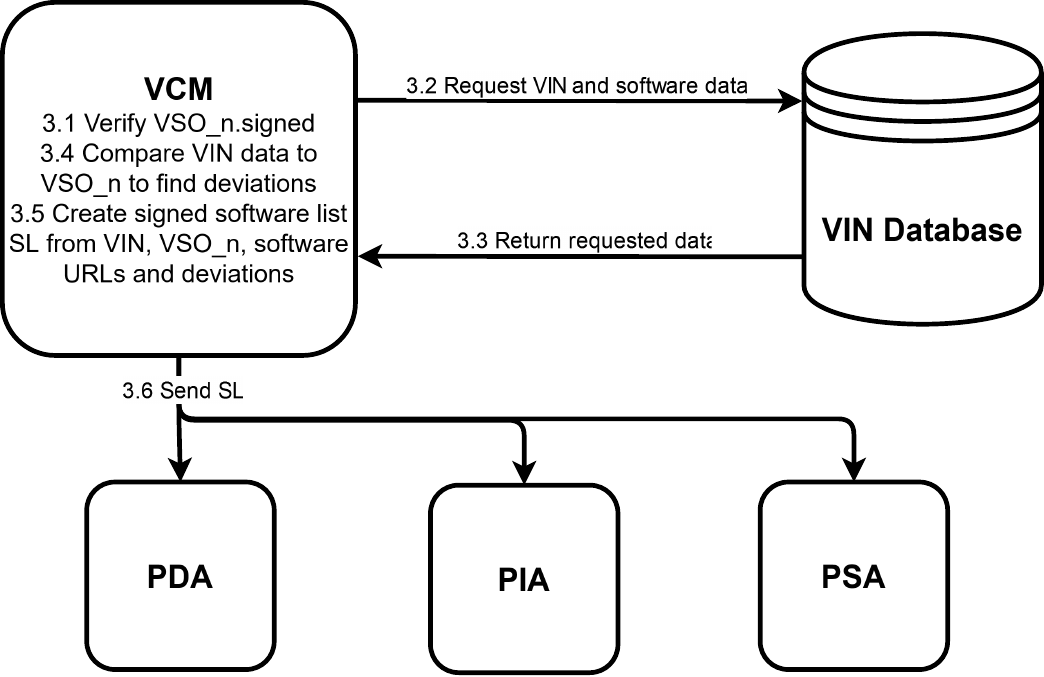}
    \caption{Diagram of the sub-problem \textit{Create Software List}.}
    \label{fig:create_software_list}
\end{figure}

\begin{table}[ht]
    \centering
    \begin{tabular}{lll}
    \midrule
    (3.2) & $ \VCM \to \VD $ & $ \VIN $ \\ 
        & & $ (\VIN, ...) = \VSO $ \\
    (3.3) & $ \VD \to \VCM $ & $ \VIN_{Data}\ \|\ \SW_{Versions} $ \\
    (3.6) & $ \VCM \to \PDA $ & $ [\SL]_{\VCM}$ \\
        & $ \VCM \to \PSA $ & $ [\SL]_{\VCM}$ \\
        & $ \VCM \to \PIA $ & $ [\SL]_{\VCM}$ \\
        & & $ \SL = Create_{\SL}( $ \\
        & & $\quad \VIN_{Data}\ \|\ \SW_{Versions}) $ \\
    \midrule
    \end{tabular}
    \caption{Communication scheme for the sub-problem \textit{Create Software List}. }
    \label{tab:create_software_list}
\end{table}

\begin{center}
\begin{align*}
S = & \{ \CDA_{\SK},\ \VCM_{\SK} \}\\
\mathcal{D} = & \{ (\VSO,\ \CDA),\ (\SL,\ \VCM),\ (\VIN_{Data},\ \VD), \\
& (\SW_{Versions},\ \VD)\}\\
\end{align*}
\end{center}

To derive the specific sub-problem requirements from the system requirements (see \Cref{sec:problem-definition}), we define $S$ and $ \mathcal{D}$ as seen above, and the partial order $P(\ell) = \{\ell_1 < \ell_2,\  \ell_2 < \ell_3,\ \ell_2 < \ell_4,\ \ell_2 < \ell_5,\ \ell_3 < \ell_6,\ \ell_4 < \ell_6,\ \ell_5 < \ell_6\}$. 
In other words, $\ell_1$ happens before $\ell_2$ happens, which then precedes $\ell_3$, $\ell_4$, and $\ell_5$ occurring in parallel, and finally $\ell_6$ happens last.
\begin{labelenumerate}
    \item \VCM obtains the most up-to-date software versions and vehicle data from \VD.
    \item \VCM creates a signed \SL.
    \item \VCM sends the signed \SL to \PDA.
    \item \VCM sends the signed \SL to \PIA.
    \item \VCM sends the signed \SL to \PSA.
    \item \VCM receives status of \SL being sent to \PDA, \PIA  and \PSA.
\end{labelenumerate}

\subsubsection{Step 4: Create Download Instructions}
\label{sec:stp-ecp-4}

Based on the \SL received from the \VCM in the previous sub-problem \textit{Create Software List} (see \Cref{sec:stp-ecp-3}), \PDA creates the \DI (see steps 4.1--4.2).
The \DI is encrypted using a generated session key.
The session key is then used to produce the key manifest \DKM (steps 4.3--4.10). 
This sub-problem finishes by signing the \DI and \DKM (steps 4.12--4.16).

The entities \PDA and \PSA have starting contexts because their listening tasks have been included in a previous sub-problem (see sub-problem \textit{Create Software List}).
The starting context of \PDA contains the signed \SL and $\VCM_{Cert}$.

\begin{figure}[ht]
    \centering
    \includegraphics[width=\columnwidth]{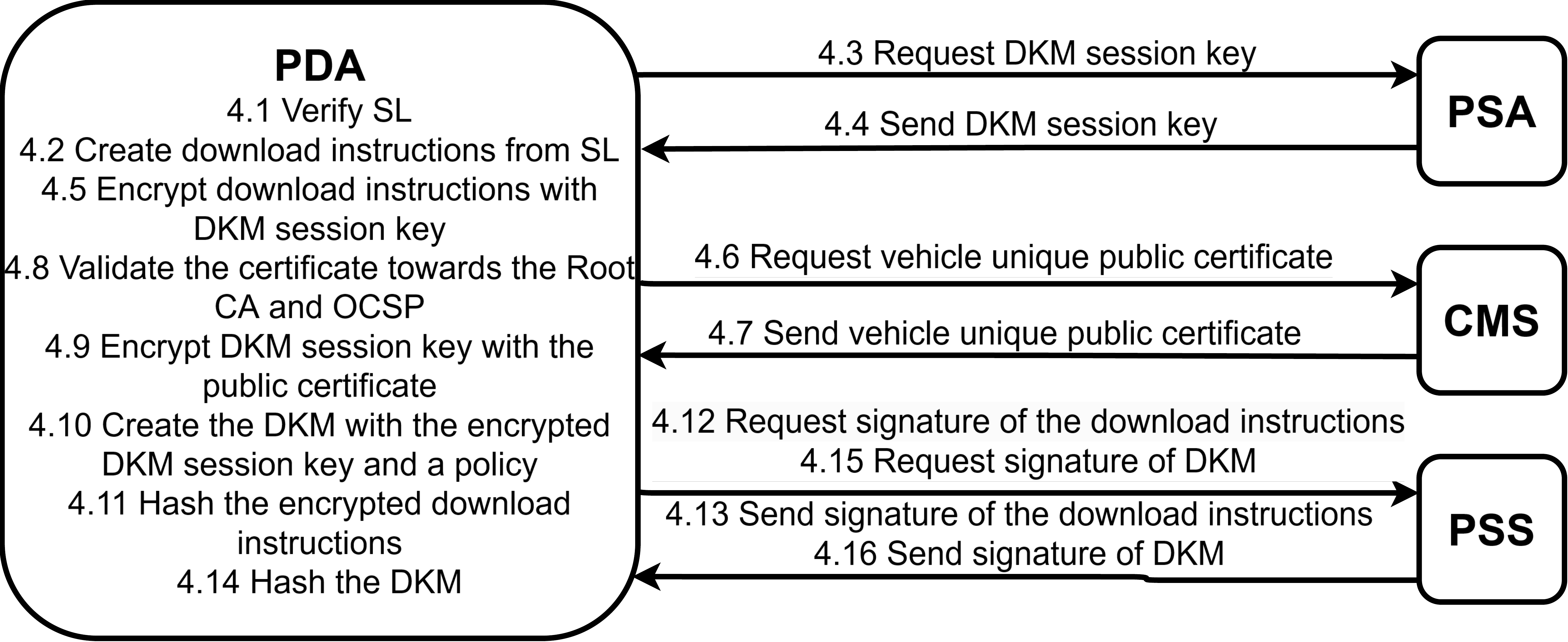}
    \caption{Diagram of the sub-problem \textit{Create Download Instructions}.}
    \label{fig:create_download_instructions}
\end{figure}

\begin{table}[ht]
    \makebox[\textwidth][c]{
    \begin{tabular}{lll}
    \midrule
    (4.3)& $ \PDA \to \PSA $ & $ \Req(\DKM_{Key}) $\\ 
    (4.4)& $ \PSA \to \PDA $ & $ \DKM_{Key} $ \\
    (4.6)& $ \PDA \to \CMS $ & $ \Req(\Vehc_{Cert}) $\\
    (4.7)& $ \CMS \to \PDA $ & $ \Vehc_{Cert} $ \\
    (4.12)& $ \PDA \to \PSS $ & $  Hash( $\\
    & & $\quad SymEnc( $\\ 
    & & $\qquad \DI,\ $\\
    & & $\qquad \DKM_{Key})) $\\
    (4.13)& $ \PSS \to \PDA $ & $ Sign( $\\
    & & $\quad Hash( $ \\ 
    & & $\qquad SymEnc( $ \\
    & & $\qquad\quad \DI,\ $\\
    & & $\qquad\quad \DKM_{Key})),$ \\
          & & $\quad \PDA_{\SK}) $\\
    (4.15)& $ \PDA \to \PSS $ & $  Hash(\DKM) $\\
    (4.16)& $ \PSS \to \PDA $ & $ Sign(Hash(\DKM),\ \PDA_{\SK}) $\\
    \midrule
    \end{tabular}
    }
    \caption{Communication scheme for the sub-problem \textit{Create Download Instructions}. }
    \label{tab:create_download_instructions}
\end{table}

\begin{center}
\begin{align*}
S = & \{ \SL,\ \DI,\  \DKM_{Key},\ Root_{\SK},\\
    &\PDA_{\SK},\VCM_{\SK} \}\\
\mathcal{D} = & \{ (\SL,\ \VCM),\ (\DI,\ \PDA),\\
    & (\DKM_{Key},\ \PSA),\ (\Vehc_{Cert},\ Root),\ (\PDA_{Cert},\ Root),\\
    & (\DKM,\ \PDA) \}
\end{align*}
\end{center}

To derive the specific sub-problem requirements from the system requirements (see \Cref{sec:problem-definition}), we define the sets $S$ and $\mathcal{D}$ as seen above, and the partial order $\mathcal{P}(\ell) = \ell_i < \ell_{i+1}$ for $1 \leq i \leq 10$ with the following labels:

\begin{labelenumerate}
    \item \PDA verifies the \SL
    \item \PDA creates the \DI
    \item \PSA generates $\DKM_{Key}$
    \item \PSA sends $\DKM_{Key}$ to \PDA
    \item \CMS sends $\Vehc_{Cert}$ to \PDA
    \item \PDA generates \DKM
    \item \PSS generates a signature for the \DI
    \item \PSS sends the signed encrypted \DI to \PDA
    \item \PSS generates signature for \DKM 
    \item \PSS sends the signed encrypted \DKM to \PDA
\end{labelenumerate}

\subsubsection{Step 6: Generate Installation Materials}
\label{sec:stp-ecp-6}

In parallel to the creation of \DI and \II (see \Cref{sec:stp-ecp-4,sec:stp-ecp-5-7}), the \PSA generates and secures cryptographic materials necessary for the installation of software updates, e.g., \SKA and \MKM. 
For instance, ECU keys for the unlocking of ECUs, to gain extended privileges.

The entities \PSA, \PSS, and \CMS have starting contexts because their listening tasks have been included in previous sub-problems (see sub-problem \textit{Create Download Instructions}).
\PSA's starting context contains the signed \SL and $\VCM_{Cert}$ and \CMS's starting context contains the $\Vehc_{Cert}$.

\begin{figure}[ht]
    \centering
    \includegraphics[width=\columnwidth]{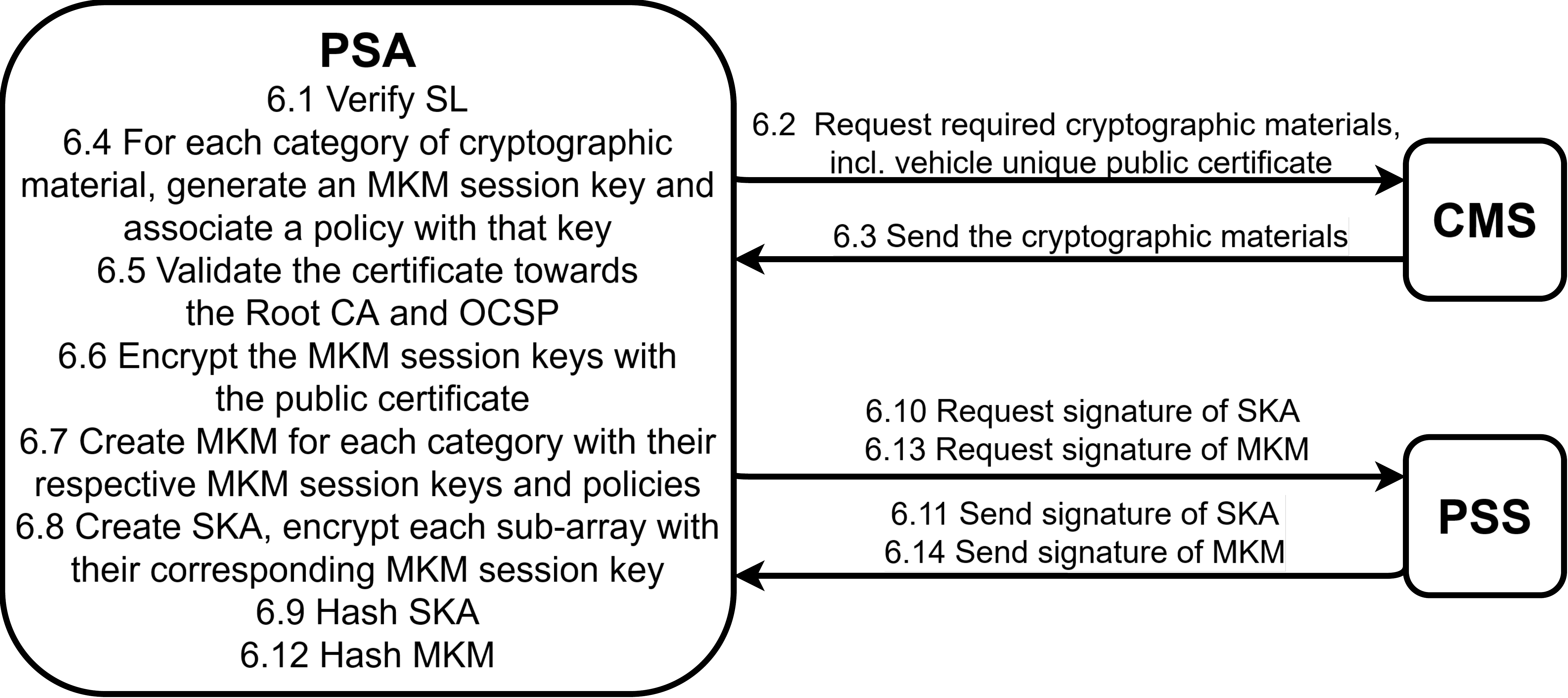}
    \caption{Diagram of the sub-problem \textit{Generate Installation Materials}.}
    \label{fig:generate_installation_materials}
\end{figure}

\begin{table}[ht]
    \centering
    \begin{tabular}{lll}
    \midrule
    (6.2)& $ \PSA \to \CMS $ & $ \SL $\\
    (6.3)& $ \CMS \to \PSA $ & $ \Vehc_{Cert}\ \|\ \SKA_{Key} $ \\ 
    (6.10)& $ \PSA \to \PSS $ & $ Hash(\SKA)$ \\
    (6.11)& $ \PSS \to \PSA $ & $ Sign(Hash(\SKA), \PSA_{\SK}) $\\
    (6.13)& $ \PSA \to \PSS $ & $ Hash(\MKM)$ \\
    (6.14)& $ \PSS \to \PSA $ & $ Sign(Hash(\MKM), \PSA_{\SK}) $\\
    \midrule
    \end{tabular}
    \caption{Communication scheme for the sub-problem \textit{Generate Installation Materials}. }
    \label{tab:generate_installation_materials}
\end{table}

\begin{center}
\begin{align*}
S = & \{ \SL,\ \MKM_{Key},\ \SKA_{Key},\ \ Root_{\SK},\ \PSA_{\SK},\\
    & \VCM_{\SK} \} \\
\mathcal{D} = & \{ (\SL,\ \VCM),\ (\Vehc_{Cert},\ Root),\ (\PSA_{Cert},\ Root),\\
    & (\MKM,\ \PSA),\ (\SKA,\ \PSA) \}
\end{align*}
\end{center}

To derive the specific sub-problem requirements from the system requirements (see \Cref{sec:problem-definition}), we define the sets $S$ and $\mathcal{D}$ as seen above, and the partial order $\mathcal{P}(\ell) = \ell_i < \ell_{i+1}$ for $1 \leq i \leq 9$ with the following labels:
\begin{labelenumerate}
    \item \PSA verifies the \SL
    \item \CMS sends cryptographic material to \PSA
    \item \PSA generates $\MKM_{ECU_{Key}}$ and $\MKM_{SW_{Key}}$
    \item \PSA generates \MKM
    \item \PSA generates \SKA
    \item \PSS generates a signature for the \SKA
    \item \PSS sends the signed \SKA to \PSA
    \item \PSS generates signature for \MKM 
    \item \PSS sends the signed \MKM to \PSA
\end{labelenumerate}

\subsubsection{Step 5 and 7: Create Installation Instructions}
\label{sec:stp-ecp-5-7}
The \PIA receives \SL from \VCM and creates \II based on the \SL (see steps 5.1--5.2 in \Cref{fig:create_installation_instructions}). 
Similarly to sub-problem \textit{Create Download Installation}, the \II are encrypted with a generated session key,  which gives rise to the \IKM (steps 7.1--7.11). 
The \II is also appended with the materials generated in the sub-problem \textit{Generate Installation Materials} (steps 7.1--7.3).

The entities \PIA, \PSA, \PSS, and \CMS have starting contexts because their listening tasks are included in previous sub-problems (see sub-problems \textit{Create Software List} and \textit{Create Download Instructions}).    
\PIA's starting context contains the signed \SL and $\VCM_{Cert}$, \PSA's starting context contains the signed \SKA and signed \MKM, and \CMS's starting context contains the $\Vehc_{Cert}$.

\begin{figure}[ht]
    \centering
    \includegraphics[width=\columnwidth]{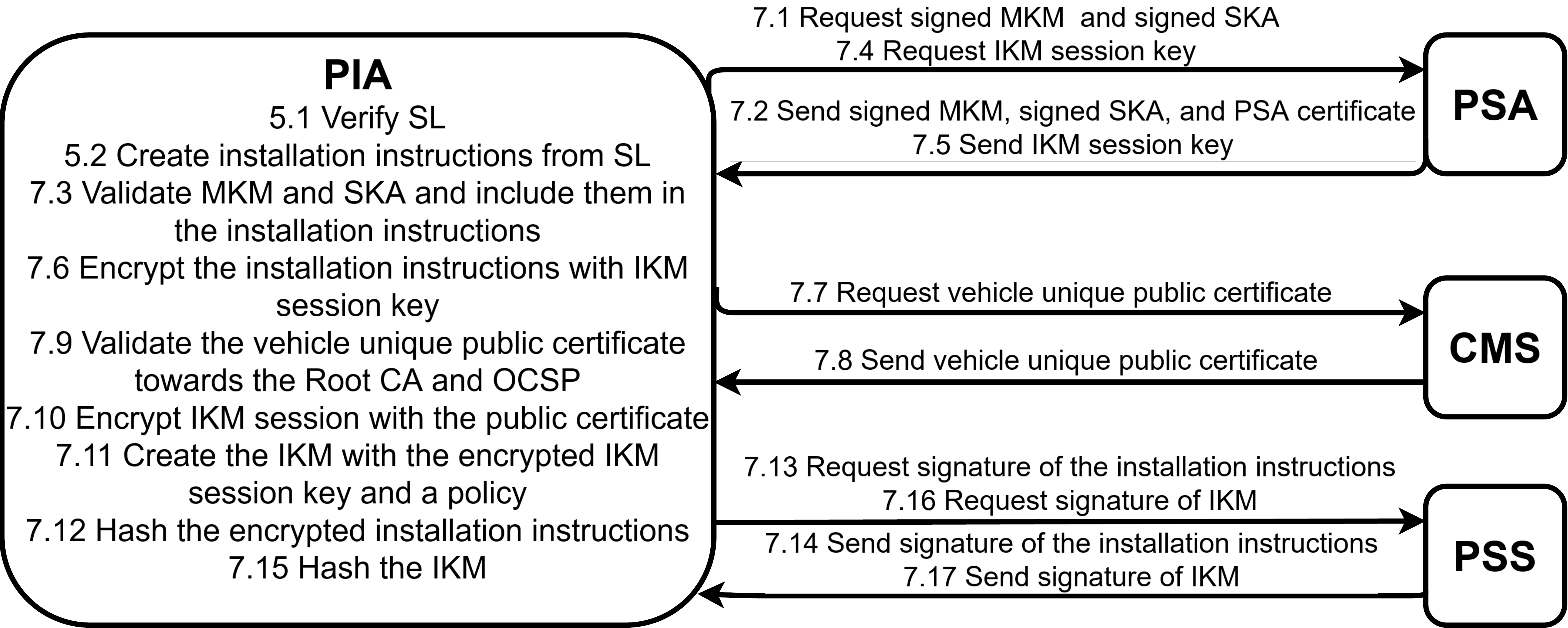}
    \caption{Diagram of the sub-problem \textit{Create Installation Instructions}.}
    \label{fig:create_installation_instructions}
\end{figure}

\begin{table}[ht]
    \centering
    \makebox[\textwidth][c]{
    \begin{tabular}{lll}
    \midrule
    (7.1)& $ \PIA \to \PSA $ & $ \Req(\MKM)\ \|\ \Req(\SKA) $ \\
    (7.2)& $ \PSA \to \PIA $ & $ [\MKM]_{\PSA}\ \|\  [\SKA]_{\PSA}\ \|\  \PSA_{Cert} $ \\
    (7.4)& $ \PIA \to \PSA $ & $ \Req(\IKM_{Key}) $ \\
    (7.5)& $ \PSA \to \PIA $ & $ \IKM_{Key} $ \\
    (7.7)& $ \PIA \to \CMS $ & $ \Req(\Vehc_{Cert}) $\\
    (7.8)& $ \CMS \to \PIA $ & $ \Vehc_{Cert} $ \\
    (7.13)& $ \PIA \to \PSS $ & $  Hash( $ \\
    & & $\quad SymEnc( $ \\
    & & $\qquad \II, $ \\
    & & $\qquad \IKM_{Key})) $\\
    (7.14)& $\PSS \to \PIA $ &  $ Sign( $\\
    & & $\quad Hash( $ \\ 
    & & $\qquad SymEnc( $ \\
    & & $\qquad\quad \II, $ \\
    & & $\qquad\quad \IKM_{Key})), $ \\ 
    & & $\quad \PIA_{\SK}) $\\
    (7.16)& $ \PIA \to \PSS $ & $ Hash(\IKM) $\\
    (7.17)& $ \PSS \to \PIA $ & $ Sign(Hash(\IKM), \PIA_{\SK}) $\\
    
    \midrule
    \end{tabular}
    }
    \caption{Communication scheme for the sub-problem \textit{Create Installation Instructions}. }
    \label{tab:create_installation_instructions}
\end{table}

\begin{center}
\begin{align*}
S = & \{ \SL,\ \II,\  \IKM_{Key},\ Root_{\SK},\\
    & \PIA_{\SK},\ \VCM_{\SK} \} \\
\mathcal{D} = & \{ (\SL,\ \VCM),\ (\II,\ \PIA),\\
    & (\IKM_{Key},\ \PSA),\ (\Vehc_{Cert},\ Root),\ (\PIA_{Cert},\ Root),\\
    & (\MKM,\ \PSA),\ (\SKA,\ \PSA) \}
\end{align*}
\end{center}

To derive the specific sub-problem requirements from the system requirements (see \Cref{sec:problem-definition}), we define the sets $S$ and $\mathcal{D}$ as seen above, and the partial order $\mathcal{P}(\ell) = \ell_i < \ell_{i+1}$ for $1 \leq i \leq 11$ with the following labels:
\begin{labelenumerate}
    \item \PIA verifies the \SL
    \item \PIA creates the \II
    \item \PSA sends $[\MKM]_{\PSA}$, and $[\SKA]_{\PSA}$, and $\PSA_{Cert}$
    \item \PSA generates $\IKM_{Key}$
    \item \PSA sends $\IKM_{Key}$ to \PIA
    \item \CMS sends $\Vehc_{Cert}$ to \PIA
    \item \PIA generates \IKM
    \item \PSS generates signature for the \II
    \item \PSS sends the signed encrypted \II to \PIA
    \item \PSS generates signature for \IKM 
    \item \PSS sends signed encrypted \IKM to \PIA
\end{labelenumerate}

\subsubsection{Step 8: Package the Instructions}
\label{sec:stp-ecp-8}

Once the required materials for software update have been created, the materials need to be inserted into the \VUUP file.
The procedure starts with \PDA and \PIA sending data to \VCM (see steps 8.1--8.2 in \Cref{fig:package_the_instructions}).
Before data is packaged into \VUUP content, signatures are validated to ensure authenticity (steps 8.3--8.8).
The \VUUP file is signed by the \PSS and later uploaded to the cloud (steps 8.9--8.14).

The entities \PDA, \PIA, \VCM, \CMS, and \PSS have a starting context because their listening tasks are included in previous sub-problems \textit{Order Initiation} (see \Cref{sec:stp-ecp-3,sec:stp-ecp-4,sec:stp-ecp-5-7,sec:stp-ecp-6}).
\PDA's starting context contains the encrypted and signed \DI and the signed \DKM, and \PIA's starting context contains the encrypted and signed \II and the signed \IKM.

\begin{figure}[ht]
    \centering
    \includegraphics[width=\columnwidth]{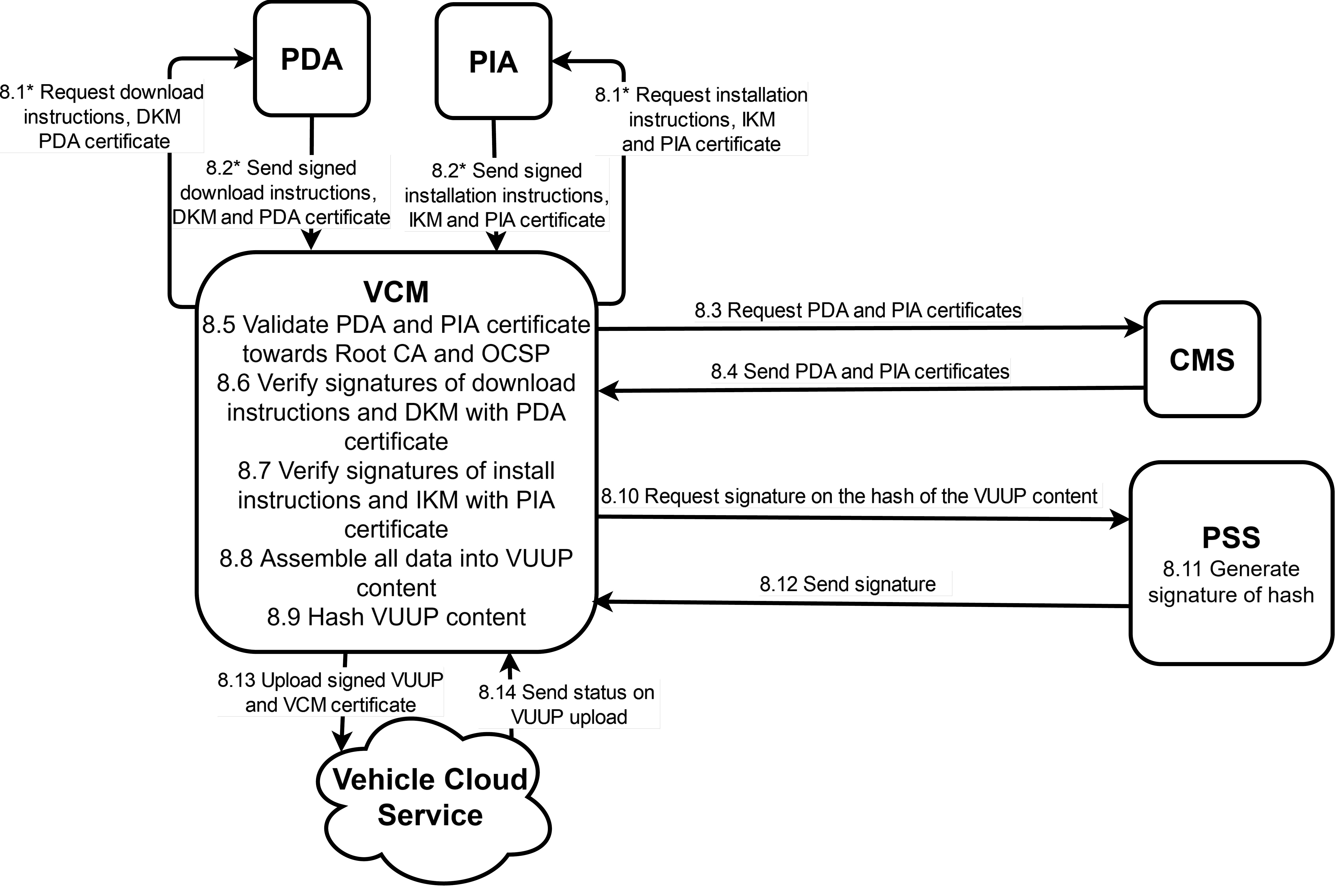}
    \caption{Diagram of the sub-problem \textit{Package the Instructions}. Steps marked with * can occur in parallel.}
    \label{fig:package_the_instructions}
\end{figure}

\begin{table}[!ht]
    \makebox[\textwidth][c]{
    \begin{tabular}{lll}
    \midrule
    (8.1)& $ \VCM \to \PDA $ & $ \Req(\DI) $ \\
         & & $\quad \|\ \Req(\DKM)\ $ \\
         & & $\quad \|\  \Req(\PDA_{Cert})  $ \\
         & $\VCM \to \PIA $ & $ \Req(\II)$ \\
         & & $ \quad \|\ \Req(\IKM)\ $\\
         & & $ \quad \|\ \Req(\PIA_{Cert})  $ \\
    (8.2)& $ \PDA \to \VCM $ & $ [SymEnc(\DI, $ \\
        & & $\quad \DKM_{Key})]_{\PDA} $ \\
        && $\quad \|\ [\DKM]_{\PDA}\ \|\ \PDA_{Cert}  $ \\
        & $ \PIA \to \VCM $ & $ [SymEnc(\II, $ \\
        & & $\quad \IKM_{Key})]_{\PIA} $\\
        && $\quad \|\ [\IKM]_{\PIA}\  \|\ \PIA_{Cert} $ \\
        
    (8.3)& $\VCM \to \CMS $ & $\Req(\PDA_{Cert}) $ \\
    & & $\quad \|\ \Req(\PIA_{Cert})$\\
    
    (8.4)& $\CMS \to \VCM $ & $\PDA_{Cert}\ \|\ \PIA_{Cert}$ \\
    
    (8.10)& $ \VCM \to \PSS $ &  $ Hash(\VUUP_{Content}) $ \\

    (8.12)& $ \PSS \to \VCM $ & $ Sign(Hash(\VUUP_{Content}), $\\ 
    & & $ \quad \VCM_{\SK}) $ \\

    (8.13)&  $ \VCM \to VCS $ & $ \VUUP $ \\
    (8.14)& $ VCS \to \VCM $ & $ Success(VUUP) $ \\

    \midrule
    \end{tabular}
    } 
    \caption{Communication scheme for the sub-problem \textit{Package the Instructions}, where $VCS$ is the Vehicle Cloud Service.}
    \label{tab:package_the_instructions}
\end{table}

\begin{center}
\begin{align*}
S = & \{ Vehicle_{\SK},\ Root_{\SK},\ \PDA_{\SK},\ \PIA_{\SK},\\
    & \PSA_{\SK},\ \VCM_{\SK},\ \DKM_{Key},\ \IKM_{Key}\} \\
\mathcal{D} = & \{ (\VUUP_{URL},\ \VCS),\ (\VUUP,\ \VCM),\\
    & (\DI,\ \PDA),\ (\II,\ \PIA),\\
    & (\PDA_{Cert},\ Root),\ (\PIA_{Cert},\ Root),\ (\VCM_{Cert},\ Root) \}
\end{align*}
\end{center}

For any execution of this sub-problem, the entities \PDA, \PIA, \VCM, \CMS, and \PSS are aware of which update round it belongs to since these entities have been in the same update round in sub-problems \textit{Order Initiation} and \textit{Create Download Instructions}.
To derive the specific sub-problem requirements from the system requirements (see \Cref{sec:problem-definition}), we define $S$ and $\mathcal{D}$ as seen above, and the partial order $\mathcal{P}(\ell) = \{\ell_1 < \ell_3,\ \ell_2 < \ell_3,\ \ell_3 < \ell_4,\ \ell_4 < \ell_5,\ \ell_5 < \ell_6,\ \ell_6 < \ell_7 \}$. 
In other words, $\ell_1$ and $\ell_2$ happen parallel initially, followed sequentially by the remaining handling events, with $\ell_7$ happening last.
\begin{labelenumerate}
    \item \VCM receives \DI and \DKM  from \PDA.
    \item \VCM receives \II and \IKM from \PIA.
    \item \VCM retrieves \PDA and \PIA certificates from \CMS.
    \item \VCM assembles the \VUUP.
    \item \PSS signs the \VUUP. 
    \item \VCM uploads the signed \VUUP along with the \VCM certificate to \VCS.
    \item \VCS receives a status of the signed \VUUP and \VCM certificate being uploaded to \VCS.
\end{labelenumerate}

\subsubsection{Step 9--11: Notify Order Ready}
\label{sec:stp-ecp-9-11}

When the \VUUP has been created, the \OA is notified that the order is ready for download (see step 9 in \Cref{fig:notify_order_ready}).
The notification consists of a signed URL to the \VUUP file.
The signature of this URL is validated by the \OA before it is uploaded to the \OCS (steps 10.1 and 10.2).

In step 11 of this sub-problem, \CDA will pull the status from \OCS to see if any update is available.
We assume updates are always available since if no updates are available, there will be an illegitimate termination of the update process.

The entities \CDA, \OA, \OCS, and \VCM have starting contexts because their listening tasks are included in a previous sub-problem (see sub-problem \textit{Order Initiation}).
\VCM's starting context contains the $VUPP_{URL}$.

\begin{figure}[ht]
    \centering
    \includegraphics[width=\columnwidth]{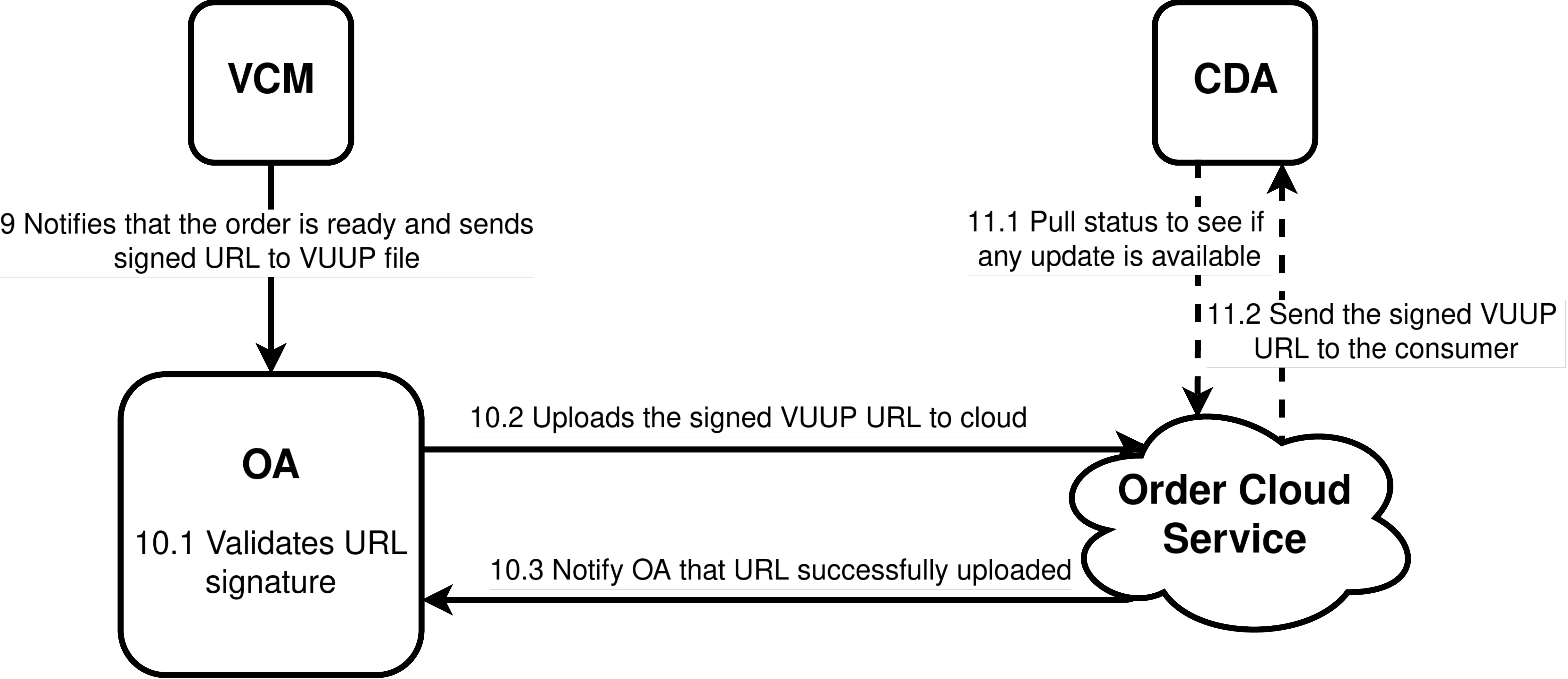}
    \caption{Diagram of the sub-problem \textit{Notify Order Ready}.}
    \label{fig:notify_order_ready}
\end{figure}

\begin{table}[ht]
    \centering
    \begin{tabular}{lll}
    \midrule
    (9)& $ \VCM \to \OA $ & $ [\VUUP_{URL}]_{\VCM} $ \\
    (10.2)& $ \OA \to OCS $ & $ [\VUUP_{URL}]_{\VCM} $ \\
    (10.3)& $ OCS \to \OA $ & $Success(\VUUP_{URL})$ \\
    (11.1)& $ \CDA \to OCS $ & $Request(\VUUP_{URL})$ \\
    (11.2)& $ OCS \to \CDA $ & $[\VUUP_{URL}]_{\VCM}\ \|\ \VCM_{Cert}$ \\
    \midrule
    \end{tabular}
    \caption{Communication scheme for the sub-problem \textit{Notify Order Ready}, where $OCS$ is the Order Cloud Service.}
    \label{tab:notify_order_ready}
 \end{table}

To derive the specific sub-problem requirements from the system requirements (see \Cref{sec:problem-definition}), we define $S = \{ \VCM_{\SK} \}$,
$\mathcal{D} = \{ (\VUUP_{URL},\ \VCS),\  (\VCM_{Cert},\ Root) \}$, and the partial order $\mathcal{P}(\ell) = \ell_i < \ell_{i+1}$ for $1 \leq i \leq 4$ with the following labels:
\begin{labelenumerate}
    \item \VCM sends $[\VUUP_{URL}]_{\VCM}$ to \OA
    \item \OA sends  $[\VUUP_{URL}]_{\VCM}$ to the \OCS
    \item \CDA pulls status from the \OCS
    \item The \OCS sends $[\VUUP_{URL}]_{\VCM}$ to \CDA
\end{labelenumerate}

%% file: decapsulation.tex
\subsection{Decapsulation}
\label{sec:decapsulation}
We summarise the decapsulation stage defined by \citet[Sec. 4.2]{UniSUF} as follows: \CDA retrieves a $\VUUP_{URL}$ from \OCS, further used to retrieve the \VUUP from \VCS. Once \CDA has validated the \VUUP, \CDA uses \CSA for the decryption and plain text retrieval of the \DI. The \DI\ is used to download software from the \SR. After that, \CIA will use \CSA to decrypt and for the plain text retrieval of the \II, and with the help of \CSA, install software to the ECUs.

\begin{figure}[ht]
    \centering
    \includegraphics[width=\columnwidth]{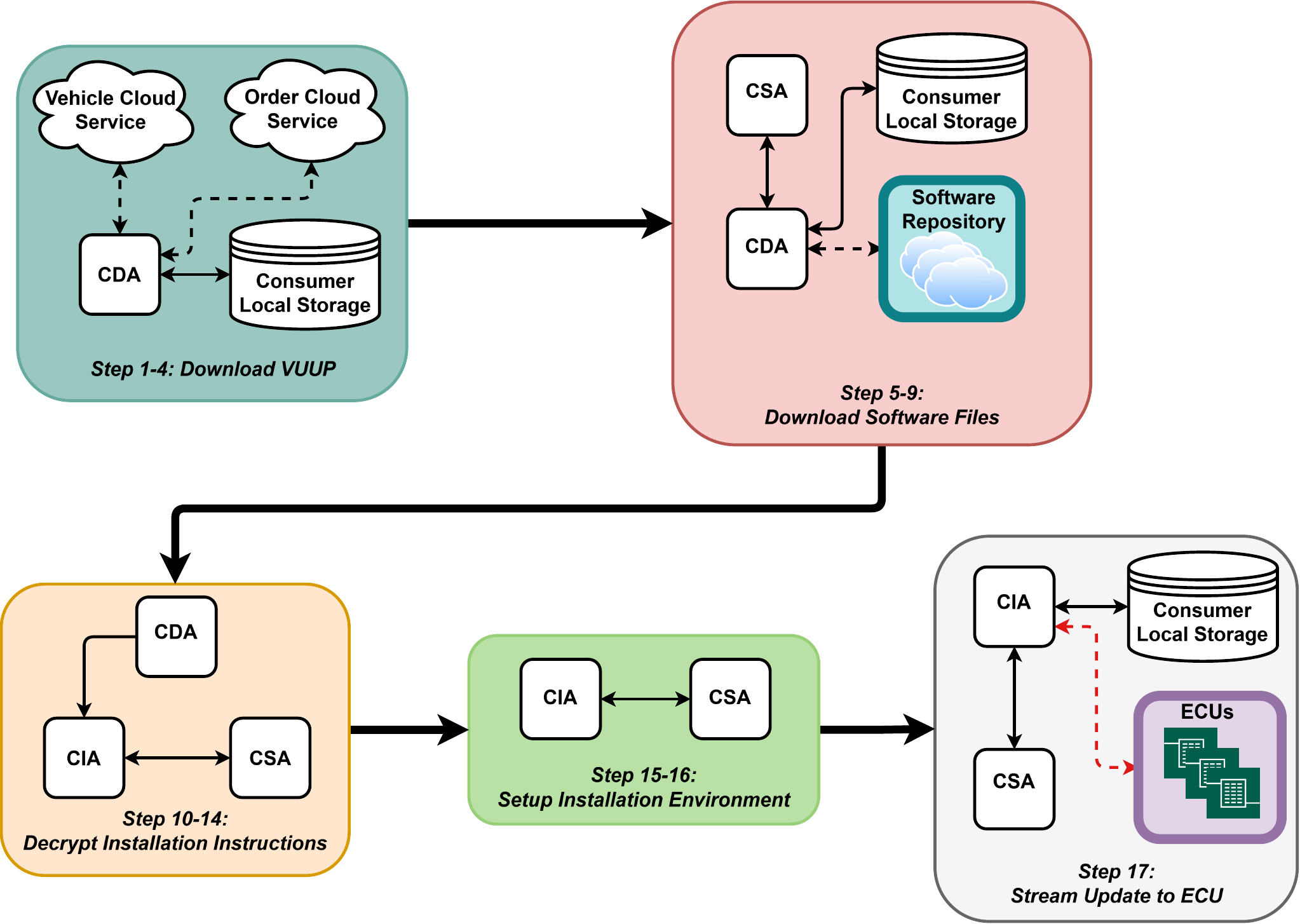}
    \caption{Using the specifications of the decapsulation stage provided by \citet[Sec. 4.2]{UniSUF}, we present the following division of sub-problems.}
    \label{fig:master-decap}
\end{figure}

We have divided the decapsulation stage into five sub-problems. A summary of the division can be seen in \Cref{fig:master-decap}. In \citet[Fig. 4]{UniSUF}, a more detailed version of the entire decapsulation process is presented. 

\subsubsection{Step 1--4: Download VUUP}
\label{sec:stp-dcp-1-4}
Once the encapsulation is done (see \Cref{sec:encapsulation}), \CDA retrieves the $\VUUP_{URL}$, and uses the URL to obtain a valid \VUUP (see step 2.2 in \Cref{fig:download_vuup,tab:download_vuup}). \CDA also guarantees that the content of the VUUP is valid. For this sub-problem, \CDA has a starting context because its initiation task is in a previous sub-problem (see \Cref{sec:stp-ecp-1-2}). \OCS and the \VCS also have starting contexts, which respectively contain the $\VUUP_{URL}$ and \VUUP for the update round (see \Cref{sec:stp-ecp-9-11,sec:stp-ecp-8}). Note that the first steps of this sub-problem, steps 1.1 and 1.2 in \Cref{fig:download_vuup,tab:download_vuup}, are the same as the last steps for the previous sub-problem (see \Cref{sec:stp-ecp-9-11}).

\begin{figure}[ht]
    \centering
    \includegraphics[width=\columnwidth]{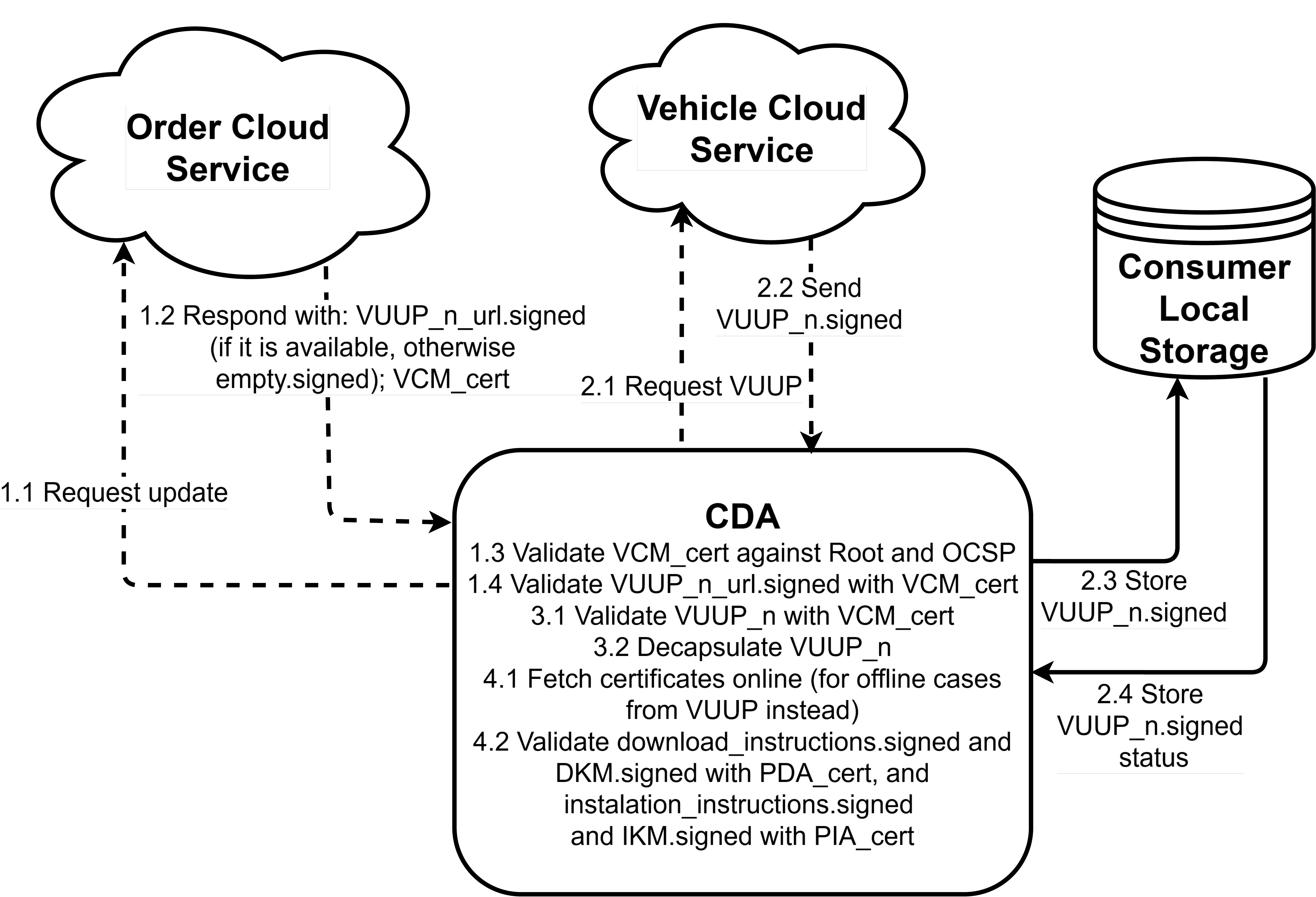}
    \caption{Diagram of the sub-problem \textit{Download VUUP}.}
    \label{fig:download_vuup}
\end{figure}

\begin{table}[ht]
    \centering
    \begin{tabular}{lll}
    \midrule

    (1.1)& $ \CDA \to \OCS $        & $ \Req(VUUP_{URL}) $ \\
    (1.2)& $ \OCS \to \CDA $        & $ [\VUUP_{URL}]_{\VCM}\ \|\ \VCM_{Cert} $ \\
    (2.1)& $ \CDA \to \VCS $        & $ VUUP_{URL}
    $\\
    (2.2)& $ \VCS \to \CDA $        & $ \VUUP$ \\
    (2.3)& $ \CDA \to \CLS $        & $\VUUP$ \\
    (2.4)& $ \CLS \to \CDA $ 
        & $\Suc(VUUP)$ \\
    \midrule
    \end{tabular}
    \caption{Communication scheme for the sub-problem \textit{Download Software Files}.}
    \label{tab:download_vuup}
\end{table}

\begin{center}
    \begin{align*}
        S = & \{\VUUP_{URL}, \ \DI,\ \II,\  \\
        & \DKM_{Key}, \IKM_{Key},\ \MKM_{Key}, \ \SKA_{Key}, \ \VCM_{\SK}, \   \\
        & \ Vehicle_{\SK}, \ \PDA_{\SK},\PIA_{\SK}, \ \PSA_{\SK}, \\
        & Root_{\SK}\} \\
        \mathcal{D} = & \{(\VUUP_{URL}, \ \VCS), (\VUUP, \ \VCM), \ \\
        & \ (\DI, \ \PDA), (\II , \ \PIA), \    \\
        &  (\DKM, \ \PDA),\ (\IKM, \ \PIA), \ (\VCM_{Cert}, \ Root), \ (\PDA_{Cert}, \ Root),\\
        &\ (\PIA_{Cert}, \ Root), \ (Root_{Cert}, \ Root) \}
    \end{align*}
\end{center}

To derive the specific sub-problem requirements from the system requirements (see \Cref{sec:problem-definition}), we define 
$S$, $\mathcal{D}$ and the partial order $\mathcal{P}(\ell) = \ell_i < \ell_{i+1}$ for $1 \leq i \leq 4$ with the following labels:
\begin{labelenumerate}
    \item \OCS has sent the signed $\VUUP_{URL}$.
    \item \CDA has validated the $\VUUP_{URL}$.
    \item \VCS has sent the signed \VUUP.
    \item \CLS has stored the signed \VUUP.
    \item \CDA has validated \DI, \DKM, \II and \IKM.
\end{labelenumerate}

\subsubsection{Step 5--9: Download Software Files}
\label{sec:stp-dcp-5-9}
\CDA requests \CSA to associate master keys in \DKM to specific trusted applications within the Trusted Execution Environment (TEE). This is followed by the decryption of the \DI\ on behalf of \CDA. Then \CDA uses the download instructions to retrieve the software from \SR (see step 9.1 in \Cref{fig:download_software_files,tab:download_software_files}). 
For this sub-problem, \CDA and \CLS have starting contexts because their respective initiation and listening tasks are included in previous sub-problems (see \Cref{sec:stp-ecp-1-2,sec:stp-dcp-1-4}). \CDA's starting context contains the signed \DKM, $\PDA_{Cert}$, $\VCM_{Cert}$, and the signed and encrypted \DI (see \Cref{sec:stp-dcp-1-4}).

\begin{figure}[ht]
    \centering
    \includegraphics[width=\columnwidth]{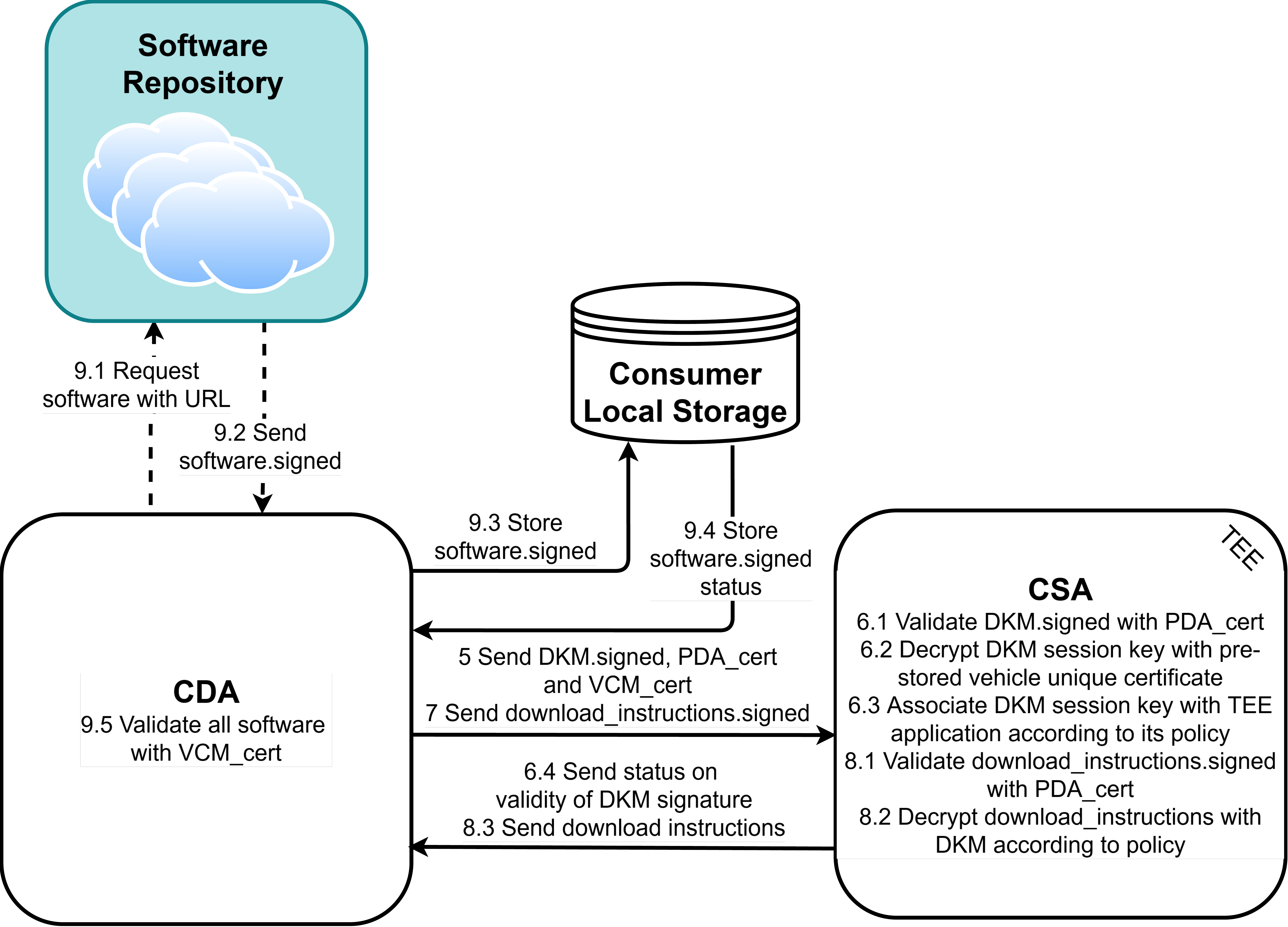}
    \caption{Diagram of the sub-problem \textit{Download Software Files}.}
    \label{fig:download_software_files}
\end{figure}

\begin{table}[ht]
    \centering
    \makebox[\textwidth][c]{
    \begin{tabular}{lll}
    \midrule
    (5)& $ \CDA \to \CSA$           & $ [\DKM]_{\PDA}\ \|\  \PDA_{Cert}\ $ \\
    & & $\quad \|\ \VCM_{Cert}$ \\
    (6.4)& $ \CSA \to \CDA$         & $\Suc(DKM)$ \\
    (7)& $ \CDA \to \CSA$           & $ [SymEnc( $ \\ 
    & & $\quad \DI,$ \\
    & & $\quad \DKM_{Key})]_{\PDA}$ \\
    (8.3)& $ \CSA \to \CDA$         & $ \DI $ \\
    (9.1)& $ \CDA \to \SR$          & $ \SW_{URL} $ \\
    (9.2)&$ \SR \to \CDA$          & $ \SW_{Encased} $ \\
    (9.3)&$ \CDA \to \CLS$         & $ 
    \SW_{Encased} $ \\
    (9.4)& $\CLS \to \CDA$
        & $\Suc(\SW)$\\
    \midrule
    \end{tabular}
    }
    \caption{Communication scheme for the sub-problem \textit{Download Software Files}. Note that the $\VCM_{Cert}$ is sent to \CSA in step 5, so that \CSA has access to the certificate in step 17 (see \Cref{sec:stp-dcp-17}) when it needs to validate software signatures \citep{kstrandberg}.}
    \label{tab:download_software_files}
\end{table}

\begin{center}
    \begin{align*}
        S = & \{\SW, \ \DI, \ \DKM_{Key}, \ \SKA_{\SW_{Key}}, \\
        & \VCM_{\SK}, \ Vehicle_{\SK}, \ \PDA_{\SK}, \ Supplier_{\SK}, \\
        & Root_{\SK}\} \\
        \mathcal{D} = & \{(\SW, \ Supplier), \ (\DI, \ \PDA), \\
        & (\DKM, \ \PDA), \ (\VCM_{Cert}, \ Root),\ (Vehicle_{Cert}, \ Root), \\
        & (\PDA_{Cert}, \ Root), \ (Root_{Cert}, \ Root)\}
    \end{align*}
\end{center}

To derive the specific sub-problem requirements from the system requirements (see \Cref{sec:problem-definition}), we define $S$, $\mathcal{D}$, and the partial order $\mathcal{P}(\ell) = \ell_i < \ell_{i+1}$ for $1 \leq i \leq 5$ with the following labels:
\begin{labelenumerate}
    \item \CSA has associated the \DKM.
    \item \CSA decrypts the \DI.
    \item \CDA has received the decrypted \DI.
    \item \SR has sent $\SW_{Encased}$.
    \item \CLS has stored $\SW_{Encased}$.
    \item \CDA has validated $\SW_{Encased}$.
\end{labelenumerate}

\subsubsection{Step 10--14: Decrypt Installation Instructions}
\label{sec:stp-dcp-10-14}
On behalf of \CIA, \CSA initiates the \IKM, followed by the decryption of \II\ (see steps 12.5 and 14.3 in \Cref{fig:decrypt_installation_instructions,tab:decrypt_installation_instructions}). For this sub-problem, \CDA and \CSA have starting contexts because their respective initiation and listening tasks are run in previous sub-problems (see ~\Cref{sec:stp-dcp-1-4,sec:stp-dcp-5-9}). \CDA's starting context contains $\PIA_{Cert}$, the signed \IKM, and the signed and encrypted \II\ (see \Cref{sec:stp-dcp-1-4}).  

\begin{figure}[ht]
    \centering
    \includegraphics[width=\columnwidth]{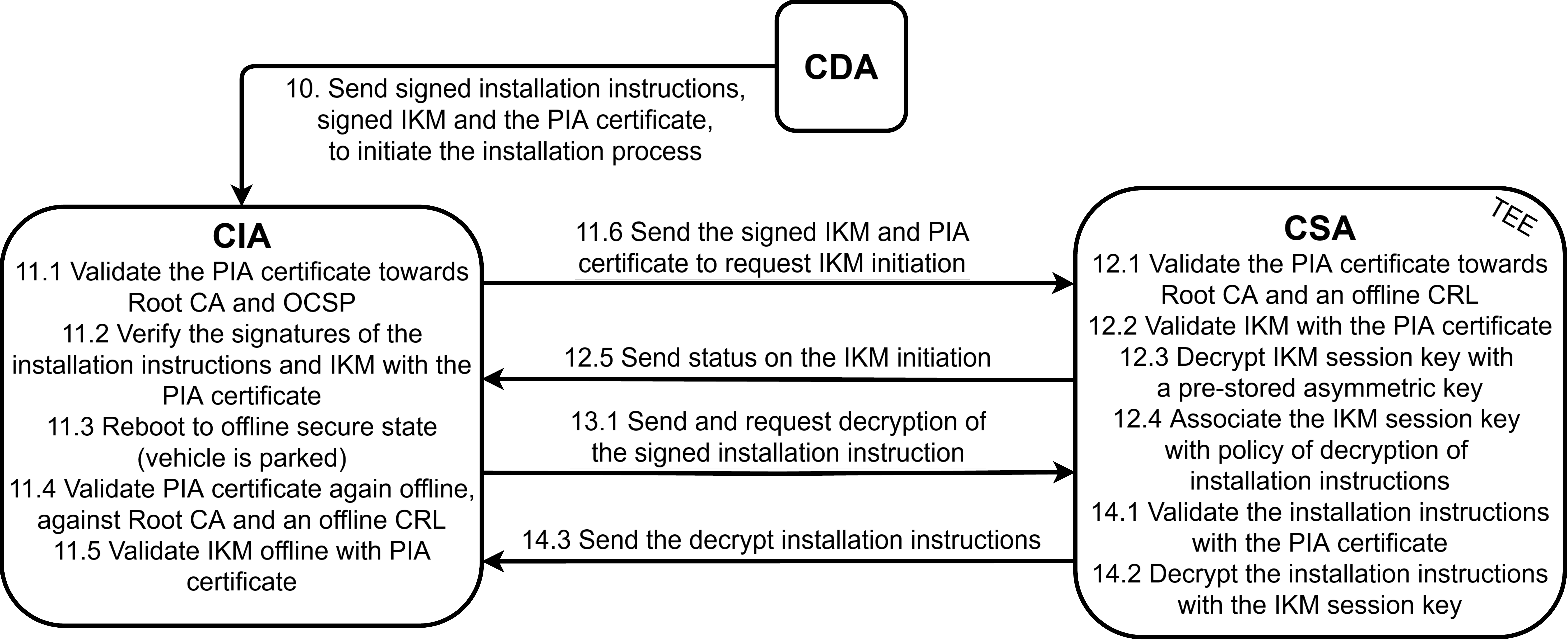}
    \caption{Diagram of the sub-problem \textit{Decrypt Installation Instructions}.}
    \label{fig:decrypt_installation_instructions}
\end{figure}

\begin{table}[ht]
    \centering
    \begin{tabular}{lll}
    \midrule
    
    (10)& $ \CDA \to \CIA $ & $ [SymEnc( $\\
    & & $\quad \II, $ \\
    & &$ \quad \IKM_{Key})]_{\PIA} $ \\
        && $\quad \|\ [\IKM]_{\PIA}\ \|\ \PIA_{Cert} $ \\
    (11.6)& $ \CIA \to \CSA $ & $ [\IKM]_{\PIA}\ \|\ \PIA_{Cert} $ \\
    (12.5)& $ \CSA \to \CIA $ & $ \Suc(\IKM) $ \\
    (13.1)& $ \CIA \to \CSA $ & $ [SymEnc( $ \\ 
    & & $\quad \II, $\\
    & & $\quad \IKM_{Key})]_{\PIA} $ \\
    (14.3)& $ \CSA \to \CIA $ & $ \II $ \\
        \midrule

    \end{tabular}
    \caption{Communication scheme for the sub-problem \textit{Setup Installation Environment}.}
    \label{tab:decrypt_installation_instructions}
\end{table}

\begin{center}
    \begin{align*}
        S = &\{\II, \ \IKM_{Key}, \
        \MKM_{Key}, \ 
        \SKA_{Key}, \\
        & \PIA_{\SK},\ \PSA_{\SK}, \ Root_{\SK}\} \\
        \mathcal{D} = & \{(\II, \ \PIA), \ (\IKM, \ \PIA), \ (\PIA_{Cert}, \ Root), \\
        & \ (Root_{Cert}, \ Root)\}
    \end{align*}
\end{center}

To derive the specific sub-problem requirements from the system requirements (see \Cref{sec:problem-definition}), we define $S$, $\mathcal{D}$ and the partial order $\mathcal{P}(\ell) = \ell_i < \ell_{i+1}$ for $1 \leq i \leq 3$ with the following labels:
\begin{labelenumerate}
    \item \CIA initiates offline mode.
    \item \CSA associates the \IKM session key with its policy.
    \item \CSA decrypts the \II.
    \item \CIA receives the decrypted \II.
\end{labelenumerate}

\subsubsection{Step 15--16: Setup Installation Environment}
\label{sec:stp-dcp-15-16}
On behalf of \CIA (see Step 15.4 in \Cref{fig:setup_installation_environment,tab:setup_installation_environment}), \CSA sets up an installation environment using the \MKM. Afterwards, \CSA is ready to decrypt software and unlock the ECUs (see \Cref{sec:stp-dcp-17}). For this sub-problem, \CSA and \CIA have starting contexts because their listening tasks are run in previous sub-problems (see \Cref{sec:stp-dcp-5-9,sec:stp-dcp-10-14} respectively). \CIA's starting context has access to a set of decrypted \II\ (see \Cref{sec:stp-dcp-10-14}).

\begin{figure}[ht]
    \centering
    \includegraphics[width=\columnwidth]{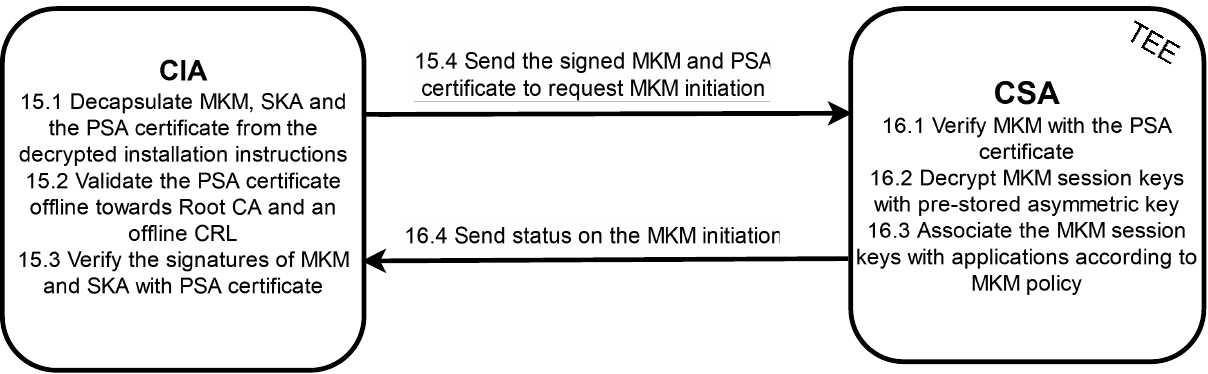}
    \caption{Diagram of the sub-problem \textit{Setup Installation Environment}.}
    \label{fig:setup_installation_environment}
\end{figure}

\begin{table}[ht]
    \centering
    \begin{tabular}{lll}
    \midrule
    (15.4)& $ \CIA \to \CSA $ & $ [\MKM]_{\PSA}\ \|\  \PSA_{Cert} $ \\
    (16.4)& $ \CSA \to \CIA $ & $ \Suc(\MKM) $ \\
    \midrule

    \end{tabular}
    \caption{Communication scheme for the sub-problem \textit{Setup Installation Environment}.}
    \label{tab:setup_installation_environment}
\end{table}

\begin{center}
    \begin{align*}
        S = & \{\II, \ \MKM_{Key}, \ 
        \SKA_{Key}, \ \PSA_{\SK}, \\
        &\ Root_{\SK} \} \\
        \mathcal{D} =& \{(\II, \ \PIA), \ (\SKA, \ \PSA), \ (\MKM, \PSA), 
         \\
         & \ (\PSA_{Cert}, Root), \ (Root_{Cert}, Root)\}
    \end{align*}
\end{center}

To derive the specific sub-problem requirements from the system requirements (see \Cref{sec:problem-definition}), we define $S$, $\mathcal{D}$ and the partial order $\mathcal{P}(\ell) = \ell_i < \ell_{i+1}$ for $1 \leq i \leq 2$ with the following labels:
\begin{labelenumerate}
    \item \CIA validates \MKM and \SKA.
    \item \CSA associated the \MKM keys with their policy.
    \item \CIA receives \MKM status from \CSA.
\end{labelenumerate}

\subsubsection{Step 17: Stream Update to ECU}
\label{sec:stp-dcp-17}
The goal of this sub-problem is to stream software updates to ECUs. This can mainly be divided into two processes; first, an ECU needs to be unlocked via a challenge-response schema, i.e., security access. The \CSA solves the unlocking on behalf of \CIA (see step 17.1 -- 17.9 in \Cref{fig:stream_update_to_ecu} and \Cref{tab:stream_update_to_ecu}). 
Second, the software is transmitted to the ECU from \CIA, after it has been decrypted by the \CSA. The software signature is validated, and depending on whether the software was successfully installed or not, a status message is sent to \CIA (see steps 17.12 -- 17.21). 
These two steps are then repeated to install different software on different ECUs \citep[Fig. 4]{UniSUF}. 
However, as mentioned in \Cref{sec:ecus} we only consider the update of one single ECU for one occasion.

For this sub-problem, \CLS, \CSA, and \CIA have starting contexts because their respective initiation and listening tasks are run in previous sub-problems (see \Cref{sec:stp-dcp-1-4,sec:stp-dcp-5-9,sec:stp-dcp-10-14}). \CLS's starting context contains $Software_{Encased}$. \CSA's starting context contains $\MKM_{\SA_{Key}}$, $\MKM_{\SW_{Key}}$ and $\VCM_{Cert}$. \CIA's starting context contains $\SKA_{\SA_{Key}}$, $\SKA_{\SW_{Key}}$, $\MKM_{\SA_{Key}}$ and $\MKM_{\SW_{Key}}$.

\begin{figure}[ht]
    \centering
    \includegraphics[width=\columnwidth]{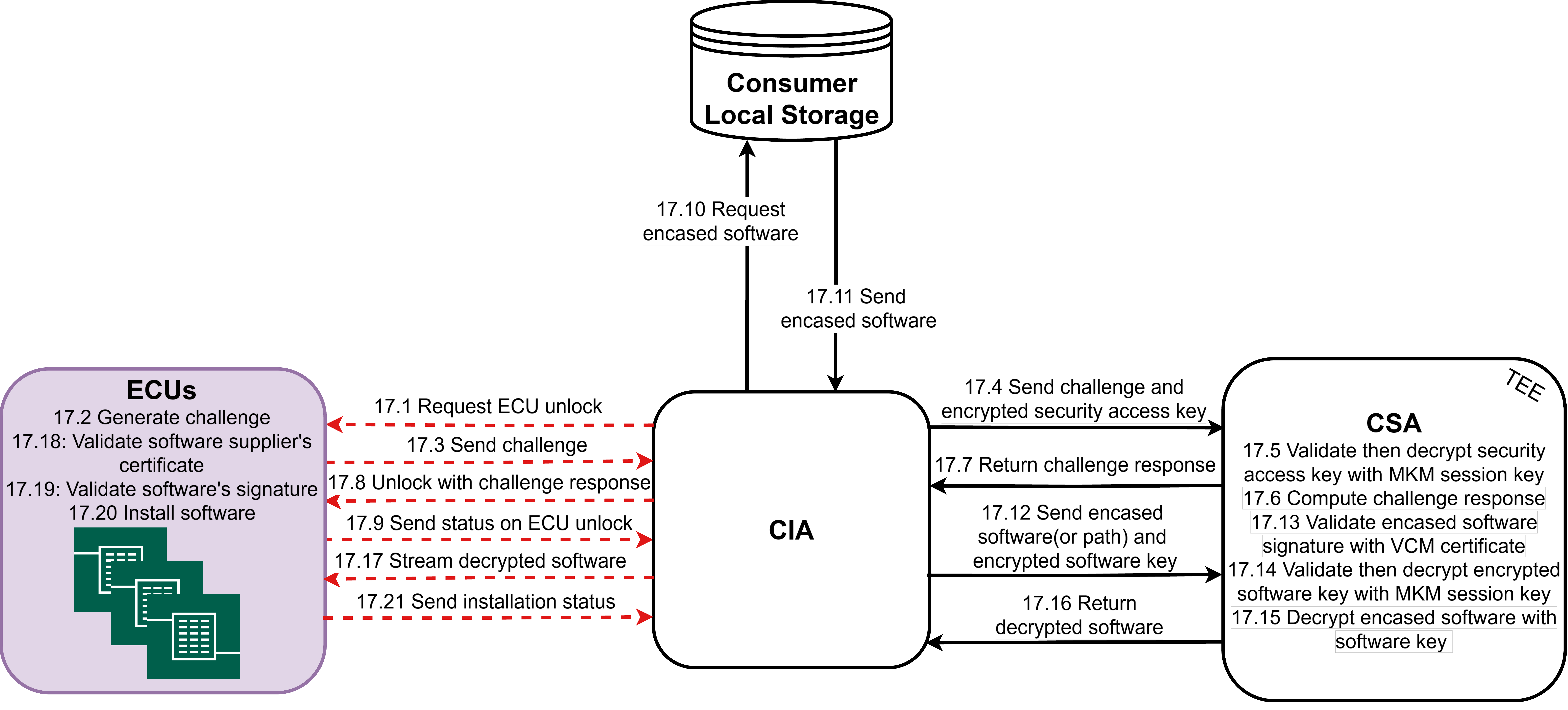}
    \caption{Diagram of the sub-problem \textit{Stream Update to ECU}.}
    \label{fig:stream_update_to_ecu}
\end{figure}

\begin{table}[ht]
    \makebox[\textwidth][c]{
    \begin{tabular}{lll}
    \midrule

    (17.1)& $  \CIA \to ECU $         & $ \Req(ECU) $ \\
    (17.3)& $  ECU \to \CIA $         & $ ECU_{Challenge} $ \\
    (17.4)& $  \CIA \to \CSA $         & $ ECU_{Challenge}\ $ \\
    & &  $\quad \|\ AuthSymEnc( $ \\ 
    & & $\qquad \SKA_{\SA_{Key}}, $ \\
    & & $\qquad MKM_{\SA_{Key}}) $ \\
    (17.7)& $  \CSA \to \CIA $         & $ ECU_{Challenge-Response}    $ \\
    (17.8)& $  \CIA \to ECU $         & $ ECU_{Challenge-Response}    $ \\
    (17.9)& $  ECU \to \CIA $         & $ ECU_{Unlocked} $ \\
    (17.10)& $ \CIA \to \CLS $          & $ \Req(Software)$ \\
    (17.11)& $ \CLS \to \CIA $         & $\SW_{Encased}$ \\
    (17.12)& $  \CIA \to \CSA $         & $ \SW_{Encased} $\\
    & & $\quad \|\ AuthSymEnc( $ \\
    & & $\qquad \SKA_{\SW_{Key}}, $ \\
    & & $\qquad \MKM_{\SW_{Key}}) $ \\
    (17.16)& $  \CSA \to \CIA $         & $ [\SW]_{Supplier} $ \\
    (17.17)& $  \CIA \to ECU $         & $ [\SW]_{Supplier} $ \\
    (17.21)& $  ECU \to \CIA $         & $ ECU_{Installation-Status} $ \\

    \midrule
    \end{tabular}
    }
    \caption{Communication scheme for the sub-problem \textit{Stream Update to ECU}. Note that this sub-problem can be repeated ~\citep[Fig.4]{UniSUF}, and therefore $\SKA_{\SA_{Key}}$, $\SKA_{\SW_{Key}}$ and $\SW$ might refer to different keys and software for different iterations. }
    \label{tab:stream_update_to_ecu}
 \end{table}

\begin{center}
    \begin{align*}
        S = & \{\SKA_{\SW_{Key}},\ \SKA_{\SA_{Key}}, \ \MKM_{\SW_{Key}}, \\
        & \MKM_{\SA_{Key}}, \ Vehicle_{\SK}, \ Supplier_{\SK}, \\
        & VCM_{\SK}, \ Root_{\SK}\} \\
        \mathcal{D} = & \{(\SW, \ \Prod), \ (\SKA_{\SW_{Key}}, \ \PSA), \\
        & (\SKA_{\SA_{Key}}, \ \PSA), \ (\MKM_{\SW_{Key}}, \ \PSA), \\
        & (\MKM_{\SA_{Key}}, \ \PSA), \  (Supplier_{Cert}, \ Root),\ (VCM_{Cert}, \ Root), \\
        & (Root_{Cert}, \ Root)\} \\
    \end{align*}
\end{center}

To derive the specific sub-problem requirements from the system requirements (see \Cref{sec:problem-definition}), we define $S$, $\mathcal{D}$ 
and the partial order $\mathcal{P}(\ell) = \{\ell_1 < \ell_2, < \ell_3 < \ell_4 < \ell_9, \ \ell_5 < \ell_6, < \ell_7 < \ell_8 < \ell_9$\}.
In other words, both sequences $\ell_{1-4}$ and $\ell_{5-8}$ happen before $\ell_9$ but can be executed concurrently.

\begin{labelenumerate}
    \item \ECU generates challenge.
    \item \CIA forwards the challenge to \CSA.
    \item \CSA responds to the challenge.
    \item \ECU accepts the challenge.
    \item \CLS sends $Software_{Encased}$ to \CIA.
    \item \CIA receives the $Software_{Encased}$ from \CLS.
    \item \CIA sends the $Software_{Encased}$ and $Software_{Key}$ to \CSA.
    \item \CSA decrypts software.
    \item \ECU installs software.
\end{labelenumerate}

%% file: methods.tex
\section{Methods}
\label{sec:methods}
We present the methods for simulating our assumptions and system settings (see \Cref{ch:preliminaries}) in ProVerif. 
Furthermore, we outline the methods used to model our requirements and provide a summary of the verification results. 
We also formally demonstrate our proofs for intra-round uniqueness and termination, i.e., \Cref{req:intra-round-uniqueness,req:termination}.

\input{simulation_of_cryptographic_primitives_in_proverif}
\input{representing_requirements_in_proverif}
\input{simulation_of_assumptions}
\input{libraries}
\input{correctness_proof_for_inner_round_uniqueness}
\input{correctness_proof_for_termination}

%% file: simulation_of_cryptographic_primitives_in_proverif.tex
\subsection{Simulation of Cryptographic Primitives in ProVerif}
\label{sec:methods-crypto-primitives}
In this section, the simulation of cryptographic primitives is described.

\subsubsection{Unauthenticated Symmetric Encryption}
\citet{proverif-manual} illustrates an example of unauthenticated symmetric encryption, as seen in~\Cref{lst:unauth-sym-enc}, by defining functions for symmetric encryption \verb|senc| and symmetric decryption \verb|sdec|. Both functions take arguments of a \verb|bitstring| (a built-in type) and a \verb|key|.
Note that \verb|key| is a user-defined type that represents symmetric keys.
The functions \verb|senc| and \verb|sdec| output the ciphertext and plaintext, respectively.
The equations on lines 4 and 5 describe the relationship between the functions \verb|senc| and \verb|sdec|, where \verb|m| is the message and \verb|k| is the key.
These equations ensure that whenever the algorithm decrypts some data,  \verb|sdec| outputs the original plaintext if and only if the same key was used during encryption and the ciphertext has not been modified.
Otherwise, it outputs some arbitrary data.

\begin{listing}[ht]
     \begin{minted}{text}
type key.
fun senc(bitstring, key): bitstring.
fun sdec(bitstring, key): bitstring.
equation forall m: bitstring, k :key; sdec(senc(m,k), k) = m.
equation forall m: bitstring, k :key; senc(sdec(m,k), k) = m.
    \end{minted}   
    \caption{Unauthenticated symmetric encryption \cite[Sec. 4.2.2]{proverif-manual}.}
    \label{lst:unauth-sym-enc}
\end{listing}

\subsubsection{Authenticated Symmetric Encryption}
The authenticated encryption in \Cref{lst:auth-sym-enc} ensures that the same key is used for encryption and decryption and that modified ciphertexts are detected.
In ProVerif, we model this by defining the decryption as a destructor through the reserved word \verb|reduc|~\citep[Sec. 3.1]{proverif-manual}. 
If the mentioned abnormalities are detected during the decryption, ProVerif blocks. Alternatively, the \verb|authSdec|: \mintinline[bgcolor=backcolour]{text}{let m = authSdec(c) in P else Q} syntax can be used, such that process \verb|P| is run if decryption succeeded and \verb|Q| is run on failure. If \mintinline{text}{else Q} is omitted, then the process \verb|P| terminates on failure.

To clarify, the behavior of the destructor is specified through reduction rules.
These can be expressed with \texttt{reduc forall}, which introduces universally quantified rewrite rules.
In particular, the rule in \cref{lst:auth-sym-enc} means: for all keys \texttt{k} and messages \texttt{m}, decrypting a ciphertext produced by \texttt{authSenc(m, k)} with the same key \texttt{k} yields the original message \texttt{m}.
In other words, \texttt{reduc forall} tells ProVerif how to symbolically simplify terms when the left-hand side pattern matches.
If no rule applies (e.g., due to a wrong key or a tampered ciphertext), the reduction does not occur, and the process either blocks or executes the else branch.

\begin{listing}[ht]
     \begin{minted}{text}
fun authSenc(bitstring, key): bitstring.
reduc forall m: bitstring, k: key; authSdec(authSenc(m, k), k) = m.
    \end{minted}   
    \caption{Authenticated symmetric encryption \cite[Sec. 3.1.2]{proverif-manual}.}
    \label{lst:auth-sym-enc}
\end{listing}

\subsubsection{Asymmetric Encryption}
Asymmetric encryption  (see \Cref{lst:asym-enc}) is similar to authenticated symmetric encryption. 
The main difference is in the consideration of key pairs. In ProVerif, the keypair is modeled such that the public key can be retrieved from the private key, but not the other way around~\citep{proverif-manual}. Also, a message decrypted with a private key must have been encrypted with the corresponding public key.

\begin{listing}[ht]
     \begin{minted}{text}
type skey.
type pkey.
fun pk(skey): pkey.
fun aenc(bitstring, pkey): bitstring.
reduc forall m: bitstring, k: skey; adec(aenc(m, pk(k)), k) = m.
    \end{minted}   
    \caption{Asymmetric encryption \cite[Sec. 3.1.2]{proverif-manual}.}
    \label{lst:asym-enc}
\end{listing}

\subsubsection{Hash Function}

The hash function (see~\Cref{lst:hash}) takes a \verb|bitstring| as input and outputs an arbitrary \verb|bitstring|, and has no associated destructors or equations~\citep[Sec. 4.2.5]{proverif-manual}. By excluding destructors and equations, the hash function resembles a random oracle model, making it impossible to reverse the hash to obtain the original value.

\begin{listing}[htbp]
     \begin{minted}{text}
fun hash(bitstring): bitstring. 
    \end{minted}   
    \caption{Hash function \citep[Sec. 4.2.5]{proverif-manual}.}
    \label{lst:hash}
\end{listing}

\subsubsection{Digital Signature}
\label{sec:dig-sign}

The modeling of digital signatures can be more complex in comparison to other cryptographic primitives.
\citet{proverif-manual} describe a method to simulate the signing function. 
It is modeled as a function that takes a message of type \verb|bitstring| and a signing private key of type \verb|sskey| as input and outputs the signed message as a \verb|bitstring| (see~\Cref{lst:proverif-digital-signatures}). 
A reducer is used to validate the authenticity of a signed message. The reducer is specified to pattern match on the correct private key, similar to how asymmetric encryption is simulated. 

\begin{listing}[htbp]
    \begin{minted}{text}
type sskey.
type spkey.

fun spk(sskey): spkey.
fun sign(bitstring, sskey): bitstring.

reduc forall m: bitstring, k: sskey; getmess(sign(m, k)) = m.
reduc forall m: bitstring, k: sskey; checksign(sign(m, k), k) = m.
    \end{minted}   
    \caption{Digital signatures schema presented in the ProVerif manual~\citep[Sec. 3.1.2.]{proverif-manual}.}
    \label{lst:proverif-digital-signatures}
\end{listing}

However, UniSUF requires a more complex signing schema, since data signing in UniSUF considers multiple steps and messages.
Therefore, the previously stated method of signing function fails to capture all our use cases.

We specify a more complex signing schema (see~\Cref{lst:digital-signatures}). This schema allows for building a signature step-by-step:
\begin{itemize}
    \item The \verb|sgnHash| is used to sign a hash.
    \item The \verb|createSgn| function takes the signed hash of a message and appends it to the message.
    \item The \verb|validateSgn| function checks that a properly signed message contains the message and the signed hash of the message.
    \item The \verb|equation| allows the more traditional \verb|sgn| signing function to be used interchangeably with \verb|createSgn|.
\end{itemize}

In summary, \verb|sgnHash| and \verb|createSgn| are used to properly model the signing process in UniSUF, while \verb|sgn| serves as a shorthand to create signatures.

\begin{listing}[ht]
    \begin{minted}{text}
type sskey.
type spkey.
fun spk(sskey): spkey.

fun sgnHash(bitstring, sskey): bitstring.
fun sgn(bitstring, sskey): bitstring.
fun createSgn(bitstring, bitstring): bitstring.

equation forall m: bitstring, ssk: sskey;
    createSgn(m, sgnHash(hash(m), ssk)) = sgn(m, ssk).
reduc forall m: bitstring, k: sskey;
    validateSgn(createSgn(m, sgnHash(hash(m), k)), spk(k)) = m.
    \end{minted}   
    \caption{Digital signatures \cite[Sec. 3.1.2]{proverif-manual}.}
    \label{lst:digital-signatures}
\end{listing}

\subsubsection{Certificates}

For certificates, we use a simplified version of the implementation presented by \citet[Appendix A.1]{junlangwang}. 
On line 2 (see~\Cref{lst:certificates}), the \verb|createCert| function outputs a certificate from a signing public key and a signing private key. 
We define two destructors \verb|validateCert| and \verb|getCert| for the function \verb|createCert|. 
The destructor \verb|validateCert| validates whether the public key corresponds to the entity that issued the certificate, and outputs the holder's public key if and only if this validation passes.
The destructor \verb|getCert| is similar to \verb|validateCert|, except that \verb|getCert| does not validate the certificate.

\begin{listing}[ht]
    \begin{minted}{text}
type cert.
fun createCert(spkey, sskey): cert.
reduc forall holderSpk: spkey, issuerSsk: sskey;
    validateCert(createCert(holderSpk, issuerSsk), spk(issuerSsk)) = (holderSpk, spk(issuerSsk)).
reduc forall holderSpk: spkey, issuerSsk: sskey;
    getCert(createCert(holderSpk, issuerSsk)) = (holderSpk, spk(issuerSsk)).
    \end{minted}   
    \caption{Certificates \citep[Appendix A.1]{junlangwang}.}
    \label{lst:certificates}
\end{listing}

%% file: representing_requirements_in_proverif.tex
\subsection{Representing Requirements in ProVerif}
\label{sec:methods-requirements}

ProVerif allows assertions of the system properties that the model must fulfill.
All such assertions are declared using the \verb|query| keyword~\citep{proverif-manual}.
In this section, we explain our methods for specifying the security requirements of UniSUF in ProVerif.

\subsubsection{Modeling Confidential Secrets}
\label{sec:confidential-secrets}

One of the main features of ProVerif is the ability to verify the secrecy of the variables~\citep{proverif-manual}. 
For secrecy queries, such as \mintinline{text}{query attacker(new x).}, ProVerif will attempt to prove that there exists no reachable state in which the attacker can access the local variable \verb|x|.
This query is a shorthand for \mintinline{text}{not attacker(new x).}, which means the query will output true if there is no state where the attacker can obtain the secret. 
More complex queries can also be created.
We use the signature function in \Cref{sec:dig-sign} as an example: \mintinline{text}{query attacker(sgn(new x, new ssk)).}. 
This checks that \verb|x| signed with the private key \verb|ssk| cannot be learned by the adversary.

\subsubsection{Modeling Integrity of Handling Events}
\label{sec:integrity-handling-events}

ProVerif allows users to define events using the keyword \verb|event|~\citep[Sec. 3.2.2]{proverif-manual}. 
These user events represent our handling events, with the event's name serving as the label ($\ell$). 
Also, data can be associated with user events, as shown in lines 1--3 in~\Cref{lst:correspondence}, thereby representing the cryptographic materials ($d$) in the handling events. The update round identifier ($r$) is also passed in as data to the user events. By using the same $r$ for all events, we ensure that the execution stays the same during the update round. However, $d$ is specified for each handling event in the execution.

\begin{listing}[htbp]
    \begin{minted}{text}
event A(bitstring, bitstring).
event B(bitstring, bitstring).
event C(bitstring, bitstring).

query r: bitstring, d1: bitstring; d2: bitstring, d3:bitstring;
    event(C(r, d1));
    event(C(r, d1)) ==> event((B(r, d2)) ==> event(A(r, d3))).
    \end{minted}
    \caption{Modeling integrity of handling events in ProVerif. Line 6 asserts that event \texttt{C} is reachable. The nested correspondence assertion in line 7 specifies that if the event \texttt{C} has happened, then event \texttt{B} must have occurred at an earlier time, and event \texttt{A} must have happened before event \texttt{B}.}
    \label{lst:correspondence}
\end{listing}

The integrity of handling events requires that handling events are executed according to a specified partial order.
ProVerif correspondence assestions~\citep[Sec. 3.2.2]{proverif-manual} represents this requirement because they specify relationships between the events ensuring they occur in the desired order.
For example, the assertion \mintinline{text}{event(e1) ==> event(e2)} states that whenever event \mintinline{text}{e1} has occurred, \mintinline{text}{e2} must have occurred previously. 

Correspondence assertions can be extended to model a chain of events by using nested correspondence assertions~\citep[Sec. 4.3.1]{proverif-manual}.
The nested correspondence assertion in~\Cref{lst:correspondence} specifies that if event \verb|C| has occurred, then event \verb|B| must have occurred previously, and event \verb|A| must have occurred before \verb|B|.
Note that line 6 is a reachability query that checks if the event \verb|C| is reachable. 
If event \verb|C| is never reached, then the entire nested correspondence assertion is vacuously true and, therefore, not meaningful. That is, any system will meet the requirement as long as it never executes \verb|C|.

The partial order of handling events can be modeled by creating different queries for chains of correspondence assertions. Lines 1 and 2 in~\Cref{lst:handeling-events-partial-order} give an example of two independent sequences that can occur concurrently: \verb|C| must have occurred before \verb|B| and \verb|E| must have occurred before \verb|D|. 
Additionally, line 3 in~\Cref{lst:handeling-events-partial-order} illustrates a method for asserting that multiple events have occurred before a specific event~\citep[Sec. 4.3.1]{proverif-manual}. 
By linking the events using conjunctions, the listed events must have happened before the target event but in no specific order.
This enables modeling concurrent events in a sequence of events.

\begin{listing}[htbp]
    \begin{minted}{text}
query event(B) ==> event(C).
query event(D) ==> event(E).
query event(A) ==> (event(B) && event(D)).
    \end{minted}
    \caption{Code example of partially ordered events being modeled in ProVerif~\citep[Sec. 4.3.1]{proverif-manual}. Note that declarations of events are excluded to avoid clutter.}
    \label{lst:handeling-events-partial-order}
\end{listing}

\subsubsection{Modeling Integrity of Cryptographic Materials}
\label{sec:integrity-cryptograhic-materials}

The integrity of cryptographic materials requires that the materials being processed remain the same during a sub-problem (see \Cref{sec:problem-definition}). This can be modelled by extending the approach used for the integrity of handling events (see~\Cref{sec:integrity-handling-events}). 
By using correspondence assertions, we model the relationships between cryptographic materials that are consumed by different handling events. Because we look at relationships, we can verify that the materials remain unchanged even as they are involved in various cryptographic transformations.

The code snippet in~\Cref{lst:integrity-crypto-mat} demonstrates this approach. Line 3 specifies that the cryptographic materials \verb|cm| have not been modified between the events. Additionally, by looking at the signing function \verb|sgn| (see \Cref{sec:dig-sign}) and the private key \verb|ssk|, we can see that the \verb|cm| was signed by a specific private key. Namely, the key belonging to the public key we specify in \verb|A| and \verb|B|: \verb|pk(ssk)|.

\begin{listing}[htbp]
    \begin{minted}{text}
query cm: bitstring, ssk: sskey; 
    event(C(sgn(aenc(cm, spk(ssk)), ssk)));
    event(C(sgn(aenc(cm, spk(ssk)), ssk))) ==> (event(B(aenc(cm, spk(ssk)))) ==> (event(A(cm, pk(ssk))))).
    \end{minted}
    \caption{Modeling integrity of cryptographic materials. The assertion verifies that if event \texttt{C} occurs (\texttt{cm} is encrypted and then signed), then event \texttt{B} must have occurred earlier (\texttt{cm} was encrypted), and event \texttt{A} must have happened even earlier (the original unencrypted \texttt{cm} was available).}
    \label{lst:integrity-crypto-mat}
\end{listing}

Because we divide UniSUF into sub-problems (see~\Cref{ch:subproblems}), there are assertions where we need cryptographic materials that are not used in the sub-problem. For example, some data could have been previously signed in another sub-problem. In that case, we need to create the signature with a private key only available during the system setup (see~\Cref{sec:methods-system-setup}), but not for any entities in the sub-problem. We create a special \mintinline{text}{event started(r, d).} that is executed during the system setup and contains cryptographic materials ($d$) not used in the sub-problem. This event allows us to make assertions with these materials.

As explained in~\Cref{sec:integrity-handling-events} and in this section, both integrity requirements can be specified using nested correspondence assertions.
Therefore, we use a single nested correspondence assertion to specify the requirements for both handling events and the cryptographic materials.

\subsubsection{Modeling Inter-Round Uniqueness}
\label{sec:inter-round-uniqueness}

Informally, inter-round uniqueness means that no two distinct update rounds produce the same data.
We can verify this property in ProVerif using correspondence assertions, similar to the approach  used by \citet{junlangwang}.
\Cref{lst:inter-uniqueness} demonstrates how the inter-round uniqueness property is expressed in ProVerif.
In this example, line 2 specifies a precondition that checks whether both instances of the event \verb|A| are reachable.
Since both instances of \verb|A| produce the same data, they must originate from the same update round.
Consequently, the assertion on line 3 requires that if both events are reachable, they must have occurred within the same update round.

\begin{listing}[htbp]
    \begin{minted}{text}
query round1: bitstring, round2: bitstring, data: bitstring;
    event(A(round1, data)) && event(A(round2, data))
    ==> round1 = round2.
    \end{minted}   
    \caption{Code example of the inter-round uniqueness query for event \texttt{A}.}
    \label{lst:inter-uniqueness}
\end{listing}

\subsubsection{Verification Summary}
\label{sec:verification_summary}

\begin{table}[htbp]
\caption{\label{tab:req-subproblems-pv}Summary of mapping from system-level requirements to ProVerif specifications and the corresponding verification results.}
\begin{tblr}{
  colspec = {|X[0.22,l,m] | X[0.64,l,m] | X[0.12,c,m]|},
  hlines,
  row{1} = {c}, 
}
\textbf{Requirement (\cref{sec:problem-definition})} &
\textbf{ProVerif formulation (\cref{sec:confidential-secrets,sec:integrity-handling-events,sec:integrity-cryptograhic-materials,sec:inter-round-uniqueness})} &
\textbf{Result} \\
\emph{Confidential Secrets} (Req.~\ref{req:confidential-secrets}) & Secrecy queries, e.g., \texttt{query attacker(new S)}, for all confidential materials \texttt{S}, such as cryptographic keys and software updates. & Verified \\
\emph{Integrity of Cryptographic Materials} (Req.~\ref{req:integrity-of-cryptographic-materials}) & The mapping is similar to that of \emph{Integrity of Handling Event}, but each event is extended with parameters specifying the expected cryptographic materials, e.g.,  \texttt{event(e\_i(d\_1,\dots,d\_j))}, for all handling events and cryptographic materials \{\texttt{d\_1,\dots,d\_j}\} specified. & Verified \\
\emph{Inter-Round Uniqueness} (Req.~\ref{req:inter-round-uniqueness}) & Correspondence assertions to check that no two distinct rounds produce the same data, e.g., \texttt{event(e\_i(r\_1, d)) \&\& event(e\_i(r\_2, d)) ==> r\_1 == r\_2}, for any pair of rounds \texttt{r\_1} and \texttt{r\_2}, some event \texttt{e\_i}, and all produced materials \texttt{d} in the current sub-problem. & Verified \\
\emph{Intra-Round Uniqueness} (Req.~\ref{req:intra-round-uniqueness}) & Not verified using ProVerif. See proof in \cref{sec:proof-intra-round-uniqueness}. & Proven \\
\emph{Integrity of Handling Event} (Req.~\ref{req:integrity-of-handling-events}) & Reachability queries, e.g.,  \texttt{event(e\_1)}, and correspondence assertions using event queries, e.g.,  \texttt{event(e\_1) ==> \dots\ ==> event(e\_j)}, for all handling events \{\texttt{e\_1,\dots,e\_j}\} specified. & Verified \\
\emph{Termination} (Req.~\ref{req:termination}) & Not verified using ProVerif. See proof in \cref{sec:proof-termination}. & Proven \\

\end{tblr}
\end{table}

\Cref{tab:req-subproblems-pv} summarizes the mapping between each system-level requirement from \Cref{sec:problem-definition} and the corresponding ProVerif artifacts used for verification, along with the outcomes obtained.

\begin{table}[htbp]
\centering
\caption{\label{tab:verif-times}ProVerif verification time in milliseconds for each of the three main problems (e.g., preparation, encapsulation, and decapsulation) and their respective sub-problems, see \cref{ch:subproblems}. As mentioned in \Cref{sec:our_contribution}, the verification overhead occurs only at design time.}
\begin{tblr}{
  colspec = {llrr},
  hline{1,Z} = {1.5pt},
  hline{2,Y} = {1pt},
  hline{5,13} = {0.5pt,dashed},
}
Task  & Sub-problem & Time [ms] & \% of Total \\
Preparation   & Step 1--4   &  $31 \pm 1$  & 1.47\%  \\
Preparation   & Step 5--6   &  $21 \pm 1$  & 1.00\%  \\
\textbf{Preparation Total} & & $52 \pm 2$ & 2.47\% \\

Encapsulation   & Step 1--2   &  $19 \pm 1$  & 0.90\%  \\
Encapsulation   & Step 3   &  $58 \pm 1$  & 2.76\%  \\
Encapsulation   & Step 4   &  $58 \pm 1$  & 2.76\%  \\
Encapsulation   & Step 5 and 7   &  $605 \pm 10$  & 28.75\%  \\
Encapsulation   & Step 6   &  $56 \pm 2$  & 2.66\%  \\
Encapsulation   & Step 8   &  $497 \pm 10$  & 23.62\%  \\
Encapsulation   & Step 9--11   &  $17 \pm 1$  & 0.81\%  \\
\textbf{Encapsulation Total} & & $1311 \pm 27$ & 62.31\% \\

Decapsulation   & Step 1--4 &  $314 \pm 6$  & 14.92\%  \\
Decapsulation   & Step 5--9 &  $129 \pm 3$  & 6.13\%  \\
Decapsulation   & Step 10--14 &  $112 \pm 2$  & 5.32\%  \\
Decapsulation   & Step 15--16 &  $29 \pm 1$  & 1.38\%  \\
Decapsulation   & Step 17   &  $156 \pm 3$  & 7.41\%  \\
\textbf{Decapsulation Total} & & $741 \pm 14$ & 35.22\% \\

\textbf{Total} &  & $2104 \pm 43$ & 100\% \\

\end{tblr}
\end{table}

To assess the computational overhead introduced by our formal verification model, we measured the verification time for each sub-problem described in \cref{ch:subproblems}.
The experiments were conducted on a Dell Latitude 5450 laptop using ProVerif 2.05.
For each sub-problem, we recorded the average runtime of the ProVerif process over 50 repeated runs, discarding the maximum and minimum values.
The resulting verification times are summarized in \cref{tab:verif-times}.
Notably, verifying the entire set of sub-problems requires about two seconds.
As discussed in \cref{sec:our_contribution}, the formal verification overhead occurs exclusively at design time and is therefore incurred only once.

%% file: simulation_of_assumptions.tex
\subsection{Simulation of System Settings and Assumptions}
We describe the techniques for simulating the assumptions listed in \Cref{ch:preliminaries}.

\subsubsection{Setting Up Cryptographic Materials and Starting Contexts}
\label{sec:methods-system-setup}

As mentioned in \Cref{sec:entities}, we consider a system with a single producer responsible for producing updates for multiple vehicles. Additionally, vehicles occasionally need to update their software with the latest versions, i.e., each vehicle can be updated multiple times. While some of the cryptographic materials we identify in \Cref{sec:crypto-mat} are used for multiple updates, others are ephemeral, i.e., they can only be used during their designated update round. Next, we show how to simulate the relationships between update rounds and cryptographic materials.

\label{sec:mapping-crypto-to-update-rounds}
We present the Mapping Tree in \Cref{fig:setup-crypto}, which is a tree consisting of different processes. 
The root \emph{process} invokes multiple vehicles in its child processes \emph{setupVehicle}, and each vehicle has multiple update rounds (invoked in the child process \emph{setupUpdateRound}). 
In each process, we create cryptographic materials. Because the materials for each sub-problem vary, we create a tree for each sub-problem. Additionally, the sub-problems related to software preparation (see \Cref{sec:preparation}) do not consider vehicles. Therefore, \emph{setupVehicle} is not included for these sub-problems. Instead, \emph{process} directly invokes \emph{setupUpdateRound}. Moreover, for \emph{Secure Software Files} (see \Cref{sec:prp-1-4}) and \emph{Order Initiation} (see \Cref{sec:stp-ecp-1-2}), \emph{setupUpdateRound} is omitted. 
This is because these sub-problems contain the initiation task, meaning the update round of the current execution has not been initiated (see \Cref{sec:modelling-unisuf}).

\begin{figure}[htbp]
    \centering
    \includegraphics[width=\columnwidth]{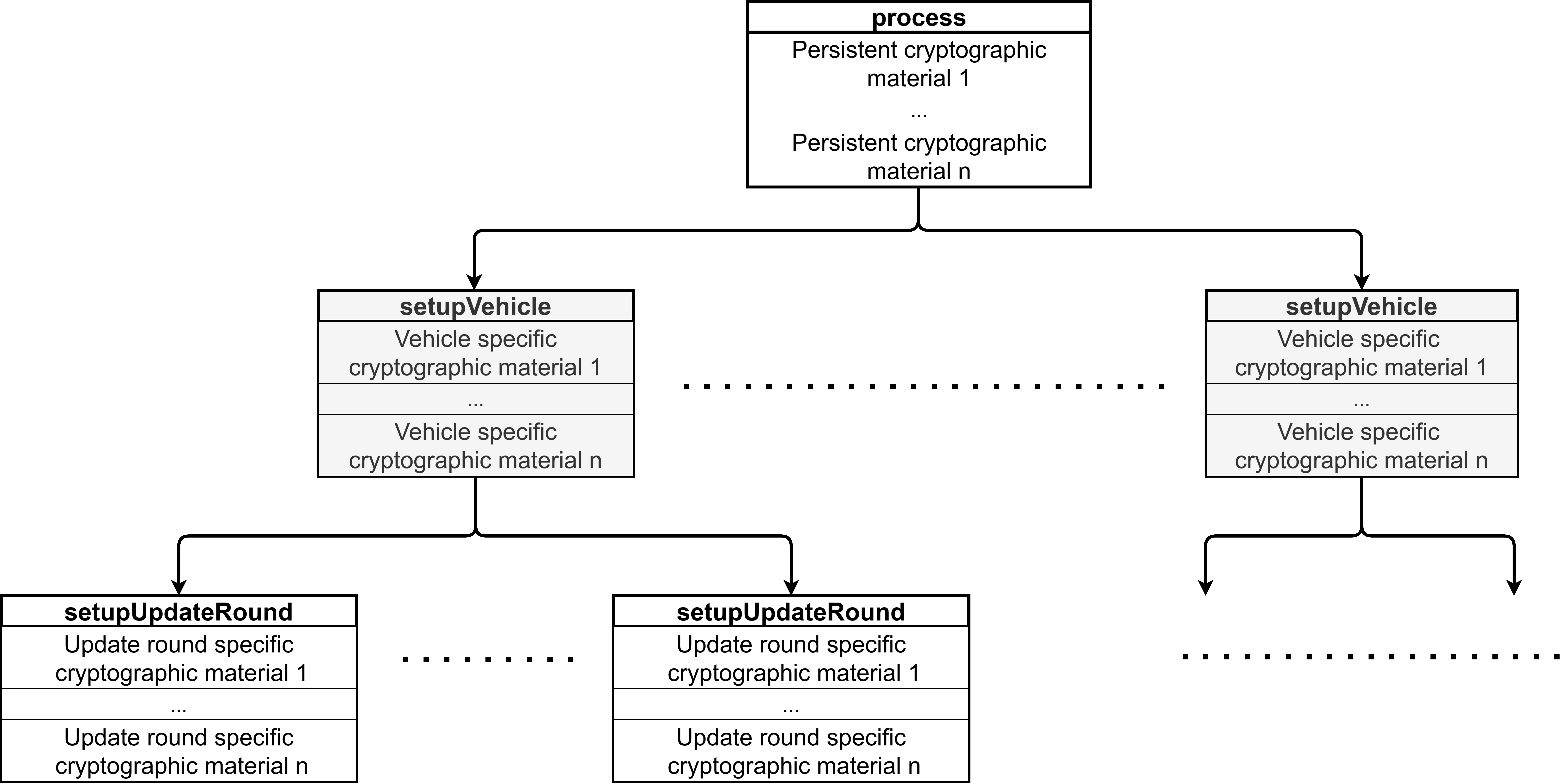}
    \caption{Mapping Tree creates multiple update rounds for each vehicle. The tree also ensures that cryptographic materials are used in their assigned update round. }
    \label{fig:setup-crypto}
\end{figure}

In the root \emph{process}, we create the software supplier certificate and all producer certificates because these are used in all update rounds. 
However, the vehicle certificate is only used in update rounds for its vehicle. 
Therefore, this certificate is created in \emph{setupVehicle}. All other cryptographic materials are created in \emph{setupUpdateRound} because they are only used in a single update round.

Each process passes down its cryptographic materials to its children processes. Therefore, \emph{setupUpdateRound} can access cryptographic materials from \emph{setupVehicle} and \emph{process}, while \emph{setupVehicle} can access materials from the root \emph{process}.

\subsubsection{Mapping Starting Contexts to Cryptographic Materials}
\label{sec:mapping-starting-contexts-to-cryptographic-materials}
As mentioned in \Cref{sec:modelling-unisuf}, an entity participating in a sub-problem can have a starting context. This occurs when the entity has previously run either the initiation or listening task for another sub-problem in the same update round. We, therefore, extend our mapping tree from \Cref{sec:mapping-crypto-to-update-rounds} to account for starting contexts (see \Cref{fig:setup-tree}).

\begin{figure}[!ht]
    \centering
    \includegraphics[width=\columnwidth]{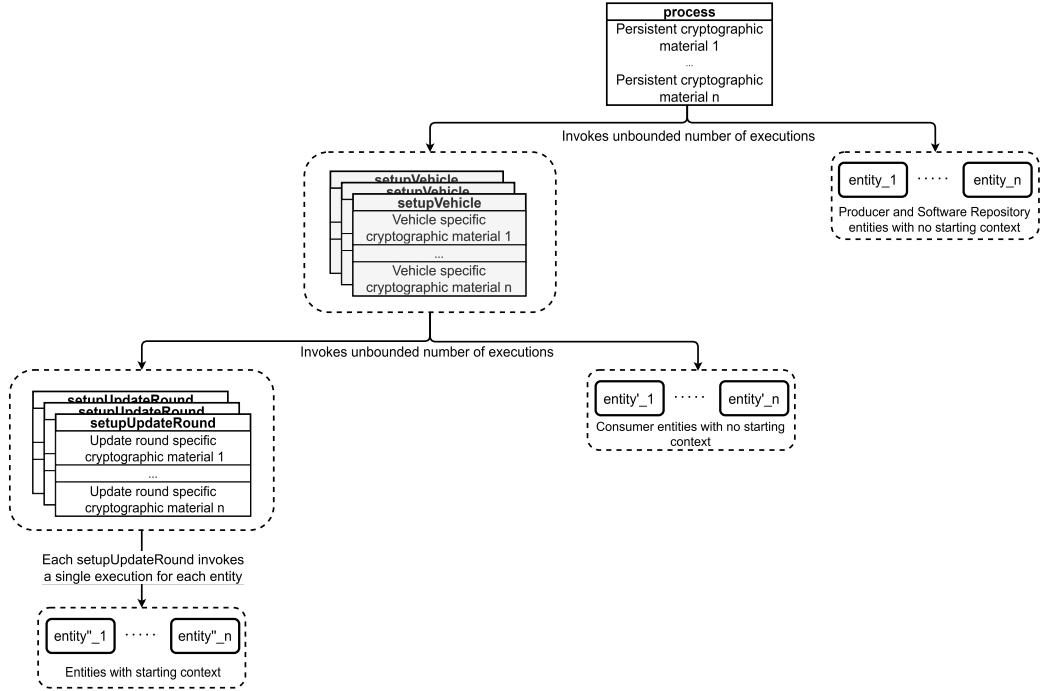}
    \caption{Extension of the mapping tree in \Cref{fig:setup-crypto}. This tree invokes the starting contexts for entities in the tree's sub-problem. The entities with starting contexts receive them when they are invoked by the tree contexts.}
    \label{fig:setup-tree}
\end{figure}

All producer and software repository entities with no starting context are invoked by \emph{process}. Therefore, they only have access to the cryptographic materials used in all update rounds, i.e., they cannot access materials coupled to a vehicle or update round. Consumer entities with no starting context also have access to vehicle-specific materials and are therefore invoked by \emph{setupVehicle}.

An entity with a starting context is invoked by \emph{setupUpdateRound}. These entities can access the cryptographic materials created for the update round. We make sure to specify the materials in accordance with \Cref{ch:subproblems}, such that only materials previously created in previous sub-problems are available. 

Note that we create $v_{id}$ and $t_e$ (see \Cref{sec:update-rounds}) in \emph{setupVehicle} and \emph{setupUpdateRound} respectively. This allows us to separate the different update rounds from each other.

\subsubsection{ProVerif Simulation of the Mapping Tree with Update Rounds}
\label{sec:simulation-mapping-tree-update-rounds}
ProVerif uses a single main process, defined with the reserved word \verb|process|, and sub-processes can be defined as macros by using the reserved word \verb|let|~\citep[Sec. 3.1]{proverif-manual}. We use these two reserved words to implement our mapping tree (see \Cref{fig:setup-tree}), as seen in \Cref{lst:setup-tree}. 
For easier comprehension, the names of the sub-processes follow the mapping tree's structure.

\begin{listing}[!ht]
     \begin{minted}{text}
let entity_1 ((*Parameters for entity_1*)) = 
    (*Algorithm for entity_1....*).
(*
Define entity_n, entity'_1, entity'_n, entity''_1 and entity''_n as entity_1 was defined
but with different parameters and algorithms.
*)

let setupVehicle((*Parameters for the vehicle*)) = 
    (*Initiate Vehicle specific material*)
    !setupUpdateRound((*Send material to each update round*)) |
    !entity'_1((*Send initial material to entity'_1*)) |
    (* .... | *)
    !entity'_n((*Send initial material to entity'_n*)).

let setupUpdateRound((*Parameters needed for the update round*)) = 
    (*Initiate update round specific material*)
    (*Running all entities in an update round*)
    entity''_1((*Send initial material to entity''_1*)) |
    (* .... | *)
    entity''_n((*Send initial material to entity''_n*)).

process
    (*Initiate persistent cryptographic material*)
    !setupVehicle((*Send initial material to each vehicle*)) |
    !entity_1((*Send initial material to entity_1*)) |
    (* .... | *)
    !entity_n((*Send initial material to entity_n*))
    \end{minted}
    \caption{Code example of how entities are instantiated with their corresponding data, thereby simulating the update round mapping tree in \Cref{fig:setup-tree}.}
    \label{lst:setup-tree}
\end{listing}

Note the exclamation operator, \verb|!|, in the listing, which in ProVerif invokes an unbounded number of replications of a process \citep[Sec. 3.1.4]{proverif-manual}.
This operator corresponds to the arrows we mark with \emph{Invokes unbounded amount of processes} in our tree. By using an unbounded amount of invocation per process, we ensure that each vehicle receives multiple updates.

To give an execution of an entity access to the cryptographic materials discussed in \Cref{sec:mapping-starting-contexts-to-cryptographic-materials}, we pass the materials as parameters when invoking the execution. In ProVerif, this is achieved by: \mintinline[bgcolor=backcolour]{text}{entity_1(cryptographicMaterial)} \citep[Sec. 3.1]{proverif-manual}.

Since all public cryptographic materials are available to all, the adversary is explicitly made aware of them.
To simulate this, we send the public material out on a channel that the adversary can read on, as such \citep[Sec. 3.1]{proverif-manual}: 
\begin{minted}{text}
    out(publicChannel, (publicData1, publicData2, ..., publicDataN))
\end{minted}
 Note that this is done for all public cryptographic materials created in \emph{process}, \emph{setupVehicle} and \emph{setupUpdateRound}.

\subsubsection{Reliable Communication}
\label{sec:reliable-communication}

We simulate reliable communication in which the receiving-side messages are delivered according to the order in which the sender fetched them.
The session identifier and message sequence number are included for all messages. 
The sequence number is incremented for each new message.
This simulates a connection establishment and ensures the correct ordering of messages.

\begin{listing}[!ht]
     \begin{minted}{text}
free alice_bob: channel.

table bobSessions(bitstring, bitstring).

let Alice() = 
    new sessionId: bitstring; (* initiate the session identifier *)
    new data1: bitstring; (* data to be sent *)
    out(alice_bob, (sessionId, data1, 1));
    in(alice_bob, (=sessionId, data2: bitstring, =2)).

let Bob() =
    new processId: bitstring; (* each instance of Bob has a unique process ID *)
    (* receive the session ID and data from Alice *)
    in(aliceBob, (sessionId: bitstring, data1: bitstring, =1)); 

    insert bobSessions(sessionId, processId);
    get bobSessions(=sessionId, processId': bitstring) suchthat processId <> processId' in
    (* another execution of Bob is already in the session, so we abort *)
        0
    else (
    (* no other execution of Bob is currently in the session, so we proceed *)
        new data2: bitstring;
        out(alice_bob, (sessionId, data2, 2))
    ).
    \end{minted}   
    \caption{Simulation of reliable communication in ProVerif.}
    \label{lst:rel-communication}
\end{listing}

Our simulation is illustrated in \Cref{lst:rel-communication}.
Alice initiates the session by sending a message that includes a session identifier, data, and sequence number 1 (line 8).
The session identifier tracks the session, while the hardcoded sequence number orders messages within it. 

In line 14, Bob listens to the first message of any session by using the ProVerif pattern matching operator (\verb|=|) \citep[Sec. 3.1.2]{proverif-manual}. When Bob receives such a message, we check if another execution of Bob is already in the session. 
This is because we use an unbounded number of processes (see \Cref{sec:simulation-mapping-tree-update-rounds}), which means that multiple executions of Bob can join the session.
To prevent another execution of the same process from joining the same session, we adapt Wang's solution~\citep[Sec. 9.3.2]{junlangwang}. 

Specifically, in line 3, we set up a table to store the session and the process identifiers. 
Then, in lines 16-17, we add the session identifier and Bob's process identifier to the table. 
We then check if another instance of Bob is already in the session, i.e., if there are at least two different process identifiers for the given session. 
If so, we abort the session establishment process. 
Otherwise, Bob proceeds to join the session.
\citet[Sec. 9.3.2]{junlangwang} states that this solution works because ProVerif schedules the processes such that, at most, one execution is allowed to proceed for a given session.
Furthermore, ProVerif explores all possible schedules, including the schedules in which only a single execution is permitted to continue.

\subsubsection{Secure and Reliable Communication}
We enhance the reliable communication channel with security guarantees.
In ProVerif, based on the Dolev-Yao threat model, the adversary has full control over the communication channels~\citep[Ch. 3]{proverif-manual}. 
They can freely read, update, or insert channel messages.
This means messages are vulnerable to eavesdropping and tampering by an attacker.
To prevent the adversary from reading, updating, or inserting channel messages, we declare the channel as private~\citep[Sec. 6.7.4]{proverif-manual}: \mintinline{text}{free c: channel[private].}.

\subsubsection{Update Rounds}
We describe how we simulate update rounds.

\subsubsection{Replacing Session Identifiers With Update Round Identifiers}
We include the update round identifier in all messages, as mentioned in \Cref{sec:update-rounds}.
The update round identifier provides context to the cryptographic materials transmitted in our code and removes the need for the session identifiers discussed in \Cref{sec:reliable-communication}.
This is because an update round encapsulates multiple peer-to-peer sessions.
This change is simple to implement. 
Instead of using session identifiers (\verb|sessionId|), we use the update round identifier (see \Cref{sec:update-rounds}), i.e., VIN and the expiration time (\verb|vin, expirationTime|), or just the expiration time (\verb|expirationTime|).

\subsubsection{Simulating the Listening Task in ProVerif}
As mentioned in \Cref{sec:modelling-unisuf}, the first task of all listeners is the listening task. We simulate this task so that listeners can discover new update rounds. As mentioned in \Cref{sec:reliable-communication}, we adapt a solution by Wang \citet[Sec. 9.3.2]{junlangwang}, to hinder multiple executions from working on the same session. For update rounds, we do the same, but we use a table containing update round identifiers instead of session identifiers.

 Listeners that are considered producer and software repository entities are made aware of the round's VIN and expiration time, as follows:
\begin{minted}{text}
in(c, (vin: bitstring, expirationTime: bitstring, ..., =1);
\end{minted}
In contrast, consumer listeners are already aware of the VIN of their vehicle (see \Cref{sec:system-settings}). 
Therefore, consumer listeners only learn the expiration time: 
\begin{minted}{text}
in(c, (=vin, expirationTime: bitstring, ..., =1));
\end{minted}
Moreover, the equals operator (\verb|=|), which is used for pattern matching in ProVerif~\citep[Sec. 3.1.4]{proverif-manual}, ensures that the consumer listeners only work on update rounds intended for their vehicle.

\subsubsection{Unbounded Number of Processes in ProVerif and Considerations}

Using unbounded processes running concurrently in ProVerif \citep{proverif-manual} means that our update rounds can also occur concurrently in our model.
However, \citet{kstrandberg} notes that this may not accurately represent reality, where update rounds occur sequentially for a given vehicle. 
The strict sequential update rounds could be achieved by allowing a vehicle to update only once.
Although this would simplify the solution, it would also trivialize the problem, as the adversary would not be able to replay messages to the same vehicle since there would be no other update rounds to target.
We believe that the use of unbounded concurrent processes is crucial to our proof.
We argue that the set of all sequential schedules is a subset of the set of all concurrent schedules.
Therefore, if our proof holds for all concurrent schedules, it will also hold for all sequential schedules.

\subsubsection{Well-Known Addresses}
We add a communication channel for each pair of entities communicating in a task. 
This allows us to separate the different peer-to-peer sessions inside an update round from each other. 
Without multiple channels, a message from Alice to Bob could end up at Charlie, that is, the channels simulate the well-known addresses discussed in \Cref{sec:system-settings}. Specifically, channels alone simulate well-known addresses for the producer and the software repository. Meanwhile, for the consumer entities that belong to a vehicle, the VIN is also needed to simulate well-known addresses.  
The reason is that the VIN separates the multiple vehicles considered in our system setup (see \Cref{sec:methods-system-setup}).

%% file: libraries.tex
\subsection{ProVerif Libraries}

ProVerif offers a method to organize frequently used functions and macros into a library file, allowing them to be imported into other files to minimize redundant code~\citep[Sec. 6.6]{proverif-manual}. In addition, the libraries ensure that the data structures appearing in multiple proofs are modeled consistently. We have developed libraries for implementing cryptographic primitives,
helper functions for setting up cryptographic materials,
and common message types. 

%% file: correctness_proof_for_inner_round_uniqueness.tex
\subsection{Correctness Proof for Intra-Round Uniqueness}
\label{sec:proof-intra-round-uniqueness}

\begin{lemma}
\label{lemma:intra-round-uniqueness}
Consider an entity $E$ and a handling $e(r,\ d,\ \ell)$, $E$ executes $e(r,\ d,\ \ell)$ at most once. 
\end{lemma}

\begin{proof}
Consider an entity, $E$, and a handling event, $e(r,\ d,\ \ell)$, which we denote $e$ for brevity. 
To execute $e$, entity $E$ must be in the update round $r$ and have previously generated the cryptographic material $d$.
Then, there can only be two cases: either the generation of $d$ depends on some cryptographic materials sent to $E$ in messages, $m'$, or $E$ generates $d$ entirely from scratch.
For the sake of simplicity, we assume that there is only one such message. 
Note that similar arguments are held when multiple messages are held.

In the first case, $m'$ must contain $(r,\ d')$ where all $d'$ were used to generate $d$.
To execute the event $e$ more than once, $E$ must have received multiple versions of message $m'$ containing $(r,\ d')$.
However, this is impossible because the entities never process duplicate messages, according to the assumption in \Cref{sec:update-rounds}, which specifies that entities omit any message already in their logs.

In the second case, $d$ is generated from scratch.
As we assume perfect cryptography and, hereby, randomness is based on a random oracle, two generated materials cannot be identical.
Therefore, executing $e$ with the same $d$ multiple times is impossible.

Thus, executing $e$ multiple times in both cases is impossible.
In other words, $E$ executes $e$ at most once.
\end{proof}

From the~\Cref{lemma:intra-round-uniqueness}, we derive that the \textit{Intra-Round Uniqueness} requirement always holds.
\begin{corollary}
The \Cref{req:intra-round-uniqueness} holds for all tasks.
\end{corollary}

%% file: correctness_proof_for_termination.tex
\subsection{Correctness Proof for Termination}
\label{sec:proof-termination}

\begin{lemma}
\label{lemma:termination}
All entity executions terminate eventually. Each termination is either timely or late.
\end{lemma}

\begin{proof}
Consider an update round; if an entity runs its \textit{halt} task for this update round before the expiration time, it terminates timely by the definition of termination in \Cref{sec:problem-definition}. 
Otherwise, it fails to run the halt task before the expiration time, and by assumption (see \Cref{sec:lifecycle-of-update-rounds}), the entity halts and terminates late.
Therefore, all entity executions always terminate and each termination is either timely or late.
\end{proof}

From the~\Cref{lemma:termination}, we derive that the \textit{Termination} requirement always holds because all update rounds must always terminate if all entities always terminate.
\begin{corollary}
The \Cref{req:termination} holds for all tasks.
\end{corollary}

%% file: conclusions.tex
\section{Conclusions}
\label{sec:Conclusions}
Our work scrutinizes UniSUF's requirements and our research questions (see \Cref{sec:problem-description}).
To validate UniSUF's requirements and architecture, we developed a formal model using ProVerif to ensure that the ProVerif program satisfies the specified requirements and the technological assumptions that UniSUF relies on.
Furthermore, we divided the UniSUF update process into smaller, more manageable sub-problems. 
We analyzed and created specific requirements for each sub-problem so that the requirements of the sub-problems together fulfill the System-level Requirements of UniSUF (see \Cref{sec:problem-definition}).

Our verification results show that our symbolic execution of UniSUF in ProVerif fulfils all requirements. 
The system-level requirements established for UniSUF collectively address the research questions posed in \Cref{sec:problem-description}.
The \textit{Confidential Secrets} requirement (\Cref{req:confidential-secrets}) ensures that the system's secrets, such as cryptographic keys and disseminated software, are not exposed to the adversary, addressing \textbf{RQ1}. 
The \textit{Integrity of Cryptographic Materials} requirement (\Cref{req:integrity-of-cryptographic-materials}) guarantees that the software obtained and used by UniSUF originates from the right source and has not been modified by any other entity, addressing \textbf{RQ2}.
The \textit{Inter-Round Uniqueness} and \textit{Intra-Round Uniqueness} (\Cref{req:inter-round-uniqueness} and \Cref{req:intra-round-uniqueness}) collaboratively prevent the use of obsolete software versions and the replay of cryptographic materials within and between update rounds.
This guarantees that UniSUF always performs software updates with the correct up-to-date versions, thus addressing \textbf{RQ3}.
The \textit{Integrity of Handling Events} (\Cref{req:integrity-of-handling-events}) ensures the operations in the software update process follow the order specified by UniSUF, thus addressing \textbf{RQ4}.
The last requirement, \textit{Termination} (\Cref{req:termination}), ensures that the update process always terminates, addressing \textbf{RQ5}.

\subsection{Discussion}
Although our work proves the UniSUF model in ProVerif to be secure, it does not necessarily guarantee the security of a real-world implementation of UniSUF.
There is an inherent discrepancy between the latter and our formal model.
Namely, formal models are intended to be complete, e.g., nothing outside the model's specification can occur, such as compromised components.
However, in real-world deployments, implementation errors and unexpected events, such as evolving attacker capabilities, can affect the system security.

In addition, like any formal analysis, our results depend on the correctness of the working assumptions and abstractions. Modeling errors or oversights could lead to missed vulnerabilities, and alternative attacker models (e.g., beyond the symbolic Dolev-Yao adversary considered here) may expose other risks that fall outside our current scope.
Thus, important challenges remain in the area. 

This does not imply that our formal verification has failed to establish meaningful security guarantees. 
On the contrary, our work has rigorously demonstrated the security properties of UniSUF in the formal model.
Our work establishes the provability of UniSUF's security, which can be a starting point for real-world implementations of UniSUF.

However, the security assurances provided by our model do not automatically transfer to a real-world implementation.
Verifying the correctness of an actual UniSUF implementation requires a substantially different and more comprehensive analysis that goes beyond the scope of our work.
Thus, bridging the gap between the formal model and a real-world implementation remains an open challenge and requires further investigation.

\subsection{Future Work}
Although UniSUF has been our case study, the methodology itself is general: by adapting the modeled roles, message formats, and requirements, a similar decomposition and verification approach can be applied to other software update frameworks in the automotive or IoT domains.
As additional directions for future work, we propose to extend our approach to cover multi-ECU and large-scale vehicular networks. The modular decomposition of UniSUF into sub-problems (\Cref{ch:subproblems}) and the requirement mapping in \Cref{tab:req-subproblems-pv}  naturally support compositional reasoning across ECUs, enabling scalability to larger deployments. Prior analyses of Uptane that incorporate secondary ECUs~\cite{DBLP:journals/jlap/KirkNBSW23,DBLP:conf/raid/LorchLTC24,boureanu2023uptane} highlight both the feasibility and the challenges of scaling symbolic verification, especially with respect to state explosion. We envision that assume-guarantee style reasoning and modular lemmas can address these challenges, while preserving strong guarantees at the system level.

%% file: appendix.tex
\section{Implementations in ProVerif}
\label{sec:Implementations}
\subsection{Cryptographic Primitives}
\label{app:crypto-primitives}
\vspace{10px}

\appendixcode{cryptography.pvl}
{ProVerif library file \texttt{cryptography.pvl} that contains all our cryptographic primitives.}{lst:cryptography.pvl}

\subsection{Utilities}
\label{app:utilities}
\vspace{10px}

\appendixcode{uniSufHelpers.pvl}
{ProVerif library file \texttt{uniSufHelpers.pvl} that contains all helper methods for setting up the cryptographic materials in \Cref{sec:crypto-mat}.}{lst:uniSufHelpers.pvl}

\subsection{Message Types}
\label{app:message-types}
\vspace{10px}

\appendixcode{uniSufMessageTypes.pvl}
{ProVerif library file \texttt{uniSufMessageTypes.pvl} that contains all the message types used in our implementation.}{lst:uniSufMessageTypes.pvl}